\documentclass[a4paper,11pt]{article} 
\usepackage[english]{babel}
\usepackage[a4paper]{geometry}
\usepackage{amsmath}
\usepackage{amsthm}
\usepackage{amssymb}

\def\bR{\mathbb{R}}

\def\bN{\mathbb{N}}

\def\bZ{\mathbb{Z}}

\def\cV{\mathcal{V}}
\def\cF{\mathcal{F}}
\def\cG{\mathcal{G}}
\def\cL{\mathcal{L}}
\def\cN{\mathcal{N}}
\def\cE{\mathcal{E}}
\def\cK{\mathcal{K}}

\def\ph{\varphi}

\def\wt{\widetilde}

\def\indic{\hbox{\raise-2pt \hbox{\indbf 1}}}

\let\==\equiv

\let\0=\noindent

\def\*{{\hfill\break\null\hfill\break}}

\def\tende#1{\,\vtop{\ialign{##\crcr\rightarrowfill\crcr
             \noalign{\kern-1pt\nointerlineskip}
             \hskip3.pt${\scriptstyle #1}$\hskip3.pt\crcr}}\,}
\def\otto{\,{\kern-1.truept\leftarrow\kern-5.truept\to\kern-1.truept}\,}

\def\tr{{\rm tr}}

\newtheorem{theorem}{Theorem}[section]  % use thm for %Theorems to keep numbering consistent

\newtheorem{prop}[theorem]{Proposition}
\newtheorem{lemma}[theorem]{Lemma}

\numberwithin{equation}{section}

\begin{document}

\title{Complete Bose-Einstein condensation in the Gross-Pitaevskii regime}

\author{Chiara Boccato$^1$, Christian Brennecke$^1$, Serena Cenatiempo$^2$, Benjamin Schlein$^1$ \\
\\
Institute of Mathematics, University of Zurich, \\
Winterthurerstrasse 190, 8057 Zurich, Switzerland$^1$ \\
\\
Gran Sasso Science Institute, \\ Viale Francesco Crispi 7 
67100 L'Aquila, Italy$^2$}

\date{September 11, 2017}

\maketitle

\begin{abstract}
We consider a gas of $N$ bosons in a box with volume one interacting through a two-body potential with scattering length of order $N^{-1}$ (Gross-Pitaevskii limit). Assuming the (unscaled) potential to be sufficiently weak, we prove complete Bose-Einstein condensation for the ground state and for many-body states with finite excitation energy in the limit of large $N$ with a uniform ($N$-independent) bound on the number of excitations.
\end{abstract}

\section{Introduction and main results}

We consider systems of $N$ bosons in the three dimensional box $\Lambda = [-1/2;1/2]^{\times 3}$, with periodic boundary conditions. In particular, we are interested in the Gross-Pitaevskii regime; the Hamilton operator has the form 
\begin{equation}\label{eq:ham0} 
H_N = \sum_{j=1}^N -\Delta_{x_j} + \kappa \sum_{i<j}^N N^2 V(N(x_i - x_j)) \end{equation}
and acts on the Hilbert space $L^2_s (\Lambda^N)$, the subspace of $L^2 (\Lambda^{N})$ consisting of functions that are symmetric with respect to permutations of the $N$ particles. We will assume $V \in L^3 (\bR^3)$ to be non-negative, spherically symmetric and compactly supported. In (\ref{eq:ham0}), we also introduced a coupling constant $\kappa > 0$, which we will later assume to be small enough. The scattering length $a_0$ of the potential $\kappa V$ is defined through the zero-energy scattering equation 
\begin{equation}\label{eq:0en} 
\left[ -\Delta + \frac{\kappa}{2} V \right] f = 0 
\end{equation}
with the boundary condition $f (x)  \to 1$ as $|x| \to \infty$ (note that (\ref{eq:0en}) is an equation on $\bR^3$, despite the fact that we consider particles moving on the torus $\Lambda$). Outside the support of $V$, $f$ has the form 
\begin{equation}\label{eq:0enf} f(x) = 1- \frac{a_0}{|x|} \end{equation}
The constant $a_0$ is known as the scattering length of $\kappa V$. By scaling, the scattering length of the interaction $\kappa N^2 V(Nx)$ appearing in (\ref{eq:ham0}) is given by $a_N = a_0/N$. 

It follows from \cite{LY,LSY,NRS} that the ground state energy $E_N$ of (\ref{eq:ham0}) is such that 
\begin{equation}\label{eq:LSY} \lim_{N \to \infty} \frac{E_N}{N} = 4 \pi a_0 \end{equation}
Moreover, it has been shown in \cite{LS,NRS} that the ground state of (\ref{eq:ham0}) exhibits Bose-Einstein condensation in the one-particle orbital $\ph_0 (x) \equiv 1$ on $\Lambda$. In other words, if $\psi_N$ is a normalized ground state vector for  (\ref{eq:ham0}), and if $\gamma_N^{(1)} = \tr_{2,\dots , N} |\psi_N \rangle \langle \psi_N|$ denotes its one-particle reduced density, it was proven in \cite{LS} that
\begin{equation}\label{eq:cond0} \gamma^{(1)}_N \to |\ph_0 \rangle \langle \ph_0| \end{equation}
as $N \to \infty$ (for example, in the trace-norm topology). Actually, results in \cite{LSY,LS} were more general and also applied to non-translation invariant bosonic systems in the Gross-Pitaevskii regime, where particles are trapped in a volume of order one by an  external confining potential. For rotating gases similar results have been obtained in \cite{LS2}. In fact, following the arguments of \cite{LS}, 
it is also possible to give a bound on the rate of the convergence (\ref{eq:cond0}), which is, however, very far from optimal. 

The main result of our paper is a proof of Bose-Einstein condensation (\ref{eq:cond0}), valid for sufficiently small values of the coupling constant $\kappa \geq 0$, with a presumably optimal bound on the rate of the convergence. This is the content of the next theorem.  
\begin{theorem}\label{thm:main}
Let $V \in L^3 (\bR^3)$ be non-negative, spherically symmetric and compactly supported and assume the coupling constant $\kappa \geq 0$ to be small enough. Let $\psi_N \in L^2_s (\Lambda^N)$ be a sequence with $\| \psi_N \| = 1$ and such that \begin{equation}\label{eq:excit-en} \langle \psi_N, H_N \psi_N \rangle \leq 4 \pi a_0 N + K \end{equation} for some $K > 0$. Let $\gamma_N^{(1)} = \tr_{2,\dots , N} \, |\psi_N \rangle \langle \psi_N|$ be the one-particle reduced density associated with $\psi_N$. Then there exists a constant $C > 0$, depending on $V$ and on $\kappa$ but independent of $K$ such that 
\begin{equation}\label{eq:conv-thm} 1- \langle \ph_0 , \gamma^{(1)}_N \ph_0 \rangle \leq \frac{C (K+1)}{N} \end{equation}
where $\ph_0 (x) = 1$ for all $x \in \Lambda$. 

Furthermore, the ground state energy $E_N$ of (\ref{eq:ham0}) is such that 
\begin{equation}\label{eq:gs-bd} \left| E_N - 4 \pi a_0 N \right| \leq D \end{equation}
for a $D >0$ independent of $N$ (depending only on $V$ and $\kappa$). Hence, the one-particle reduced density associated with the ground state of (\ref{eq:ham0}) satisfies (\ref{eq:conv-thm}), with $K$ replaced by the constant $D$.
\end{theorem}

{\it Remarks:} 
\begin{itemize}
\item[1)] The inequality (\ref{eq:conv-thm}) bounds the number of particles orthogonal to the condensate wave function $\ph_0$. It 
states that the number of orthogonal excitations in the state $\psi_N$ is bounded by $C(K+1)$. In particular, by (\ref{eq:gs-bd}), the number of orthogonal excitations in the ground state of (\ref{eq:ham0}) remains bounded, uniformly in $N$. 
\item[2)] The bound (\ref{eq:gs-bd}) improves the result (\ref{eq:LSY}) obtained in \cite{LY,LSY,NRS} by showing that $4 \pi a_0 N$ is the ground state energy of (\ref{eq:ham0}) up to an error of order one, uniform in $N$. This is a new result, of independent interest. 
\item[3)] The inequality (\ref{eq:conv-thm}) immediately implies convergence of the reduced density $\gamma^{(1)}_N$ towards the orthogonal projection $|\ph_0 \rangle \langle \ph_0|$ in the trace-class topology, since  
\[ \tr \, \left| \gamma^{(1)}_N - |\ph_0 \rangle \langle \ph_0| \right| \leq 2 \, \| \gamma_N^{(1)} - |\ph_0 \rangle \langle \ph_0| \|_\text{HS} \leq 2^{3/2} \left( 1 - \langle \ph_0 , \gamma^{(1)}_N \ph_0 \rangle \right)^{1/2} \leq \frac{C}{\sqrt{N}} \] 
\item[4)] We believe that the methods that we use to show Theorem \ref{thm:main} can be extended to prove an analogous result for non-translation-invariant bosonic systems trapped by confining external fields. The details will appear elsewhere.
\item[5)] We think that the smallness assumption on $\kappa > 0$ is technical; we expect the results of Theorem \ref{thm:main} to remain true, independently of the strength of the interaction (of course, assuming the interaction to scale as in (\ref{eq:ham0})). 
\item[6)] The threshold $K > 0$ can also be chosen depending on $N$. Of course, if we allow for a large excitation energy $K \simeq N^\alpha$ for some $\alpha < 1$, the bound (\ref{eq:conv-thm}) deteriorates and only shows that the number of orthogonal excitations is at most of the order $N^\alpha$. The statement (\ref{eq:conv-thm}) remains non-trivial for all $K \ll N$. 
\end{itemize}

Bounds similar to (\ref{eq:conv-thm}) have been obtained in \cite{Sei,GS,DN,Piz} for $N$-boson systems in the mean field limit, described by the Hamilton operator \begin{equation}\label{eq:mf} H_N^\text{mf} = \sum_{j=1}^N -\Delta_{x_j} + \frac{1}{N} \sum_{i<j}^N V (x_i -x_j) \end{equation}
acting again on $L^2_s (\Lambda^{3N})$. In \cite{Sei,GS,LNSS,DN} establishing an estimate on the number of particles orthogonal to the condensate was an important ingredient to show the validity of Bogoliubov theory for the mean-field Hamiltonian (\ref{eq:mf}). In this sense, (\ref{eq:conv-thm}) can be thought of as a first step towards a better mathematical understanding of the excitation spectrum of Bose gases in the Gross-Pitaevskii regime corresponding to (\ref{eq:ham0}).  

To prove Theorem \ref{thm:main} we  combine techniques from \cite{LNSS} with ideas developed in \cite{BDS} and recently in \cite{BS} to study the time-evolution in the Gross-Pitaevskii regime. First of all, following \cite{LNSS}, we observe that every normalized $\psi_N \in L^2_s (\Lambda^{N})$ can be represented uniquely as 
\begin{equation}\label{eq:UNdef} \psi_N = \sum_{n=0}^N \psi_N^{(n)} \otimes_s \ph_0^{\otimes (N-n)} 
\end{equation}
for a sequence $\psi_N^{(n)} \in L^2_{\perp} (\Lambda)^{\otimes_s n}$. Here $L^2_\perp (\Lambda)^{\otimes_s n}$ denotes the symmetric tensor product of $n$ copies of the orthogonal complement $L^2_{\perp} (\Lambda)$ of $\ph_0$ in $L^2 (\Lambda)$. This remark allows us to define a unitary map \begin{equation}\label{eq:Uph-def} U_N: L^2_s (\Lambda^{N}) \to \cF_{+}^{\leq N} \quad \text{ through } \quad  U_N \psi_N = \{ \psi^{(0)}_N, \psi^{(1)}_N, \dots , \psi^{(N)}_N  \}.\end{equation} 
Here $\cF^{\leq N}_+ = \bigoplus_{n = 0}^N L^2_\perp (\Lambda)^{\otimes n}$ denotes the bosonic Fock space constructed over $L^2_\perp (\Lambda)$, truncated to sectors 
with at most $N$ particles. The unitary map $U_N$ factors out the Bose-Einstein condensate described by $\ph_0$ and it let us focus on its orthogonal excitations, described on $\cF_+^{\leq N}$.

With $U_N$, we can define a first excitation Hamiltonian $\cL_N = U_N H_N U_N^* : \cF^{\leq N}_+ \to \cF_+^{\leq N}$. To compute $\cL_N$, it is convenient to rewrite the original Hamiltonian (\ref{eq:ham0}) in second quantized form as 
\begin{equation}\label{eq:HN-Fock} H_N = \sum_{p \in \Lambda^*} p^2 a_p^* a_p + \frac{\kappa}{N} \sum_{p,q, r \in \Lambda^*} \widehat{V} (r/N) a_{p+r}^* a_q^* a_p a_{q+r} \end{equation}
where $\Lambda^* = 2\pi \bZ^3$ is momentum space and where, for every $p \in \Lambda^*$, $a_p^* , a_p$ are the usual Fock space operators, creating and annihilating a particle with momentum $p$ (precise definitions will be given in Section \ref{sec:fock}). Roughly speaking, $\cL_N$ can be obtained from (\ref{eq:HN-Fock}) by replacing creation and annihilation operators $a_0^*, a_0$ in the zero-momentum mode by factors of $(N-\cN_+)^{1/2}$, where $\cN_+ = \sum_{p \in \Lambda^* \backslash \{ 0 \}} a_p^* a_p$ is the number of particles operator on the excitation space $\cF_+^{\leq N}$. This procedure can be thought of as a rigorous version of the Bogoliubov approximation, proposed already in \cite{Bog}. Conjugating $H_N$ with $U_N$ we effectively extract, from the original interaction term in (\ref{eq:HN-Fock}) (quartic in creation and annihilation operators), contributions that are constant (commuting numbers), quadratic and cubic in creation and annihilation operators (the precise form of $\cL_N$ is given in (\ref{eq:cLN}) and (\ref{eq:cLNj})). 

In the mean field regime described by the Hamilton operator (\ref{eq:mf}), assuming that $V$ is positive definite it turns out that, up to errors of order one, 
\begin{itemize}
\item[i)] the constant term in $\cL_N^\text{mf} = U_N H^\text{mf}_N U_N^*$ is given by $N \widehat{V} (0)/2$, which is (again up to errors of order one) the ground state energy of (\ref{eq:mf}), 
\item[ii)] the sum of all other contributions in $\cL_N^\text{mf}$ can be bounded below on $\cF_+^{\leq N}$ by the number of particles operator $\cN_+$.
\end{itemize}
We conclude that \begin{equation}\label{eq:low-mf} \cL_N^\text{mf} - N \widehat{V} (0)/2 \geq c \, \cN_+ - C \end{equation}
for appropriate constants $C,c > 0$. This bound shows that states with small excitation energy can be written as $\psi_N = U_N^* \xi_N$ for an excitation vector $\xi_N \in \cF_+^{\leq N}$ with $\langle \xi_N , \cN_+ \xi_N \rangle \leq C$, uniformly in $N$. It is easy to check that this estimate implies (\ref{eq:conv-thm}). 

In the Gross-Pitaevskii regime, on the other hand, conjugating with $U_N$ is not enough. The difference between the constant term in $\cL_N$ and the ground state energy of (\ref{eq:ham0}) is still of order $N$ and, moreover, the sum of the other contributions to $\cL_N$ cannot be bounded below by the number of particles operator. The problem, in the Gross-Pitaevskii regime, is the fact that the completely factorized wave function $U_N^* \Omega = \ph_0^{\otimes N}$ (with $\Omega = \{ 1, 0, \dots , 0 \}$ the vacuum vector in $\cF^{\leq N}_{+}$) is not a good approximation for the ground state vector of (\ref{eq:ham0}) or, more generally, for low-energy states. Instead, states with small energies in the Gross-Pitaevskii limit are characterized by a short scale correlation structure, which already played a crucial role in \cite{LSY,LS} and also in the analysis of the time-evolution; see \cite{ESY1,ESY2,ESY3,ESY4,P,EMS,BDS,CH,BS}. To take into account correlations we proceed as in \cite{BS}, conjugating $\cL_N = U_N H_N U_N^*$ with a generalized Bogoliubov transformation $T$. This idea stems from \cite{BDS}, where Bogoliubov transformations of the form 
\begin{equation}\label{eq:wtT} \wt{T} = \exp \left\{ \frac{1}{2} \sum_{q \in \Lambda^*_+}  \eta_q \left[ a_q^* a_{-q}^* - a_q a_{-q} \right] \right\} \end{equation}
with coefficients $\eta_q \in \bR$ related to the solution of the zero energy scattering equation (\ref{eq:0en}), have been used to model correlations (in fact, since \cite{BDS} studied the time-evolution in non-translation-invariant systems, a slightly more general version of (\ref{eq:wtT}) was used there). A nice property of the unitary map (\ref{eq:wtT}) is the fact that its action on creation and annihilation operators can be computed explicitly, i.e. 
\[ \wt{T}^* a_p \wt{T} = \cosh (\eta_p) \, a_p  + \sinh (\eta_p) \, a^*_{-p} \]
for all $p\in \Lambda^*_+$. Unfortunately, however, the Bogoliubov transformation $\wt{T}$ does not map $\cF^{\leq N}$ into itself (it does not preserve the constraint on the number of particles). To circumvent this obstacle, we follow \cite{BS} and introduce generalized Bogoliubov transformations, having the form
\begin{equation}\label{eq:T} T = \exp \left\{ \frac{1}{2} \sum_{p \in \Lambda^*_+} \eta_p \left[ b_p^* b_{-p}^* - b_p b_{-p} \right] \right\} \end{equation}
with the modified creation and annihilation operators 
\[ b_p = \sqrt{\frac{N-\cN_+}{N}} \, a_p  \qquad \text{and } \quad b^*_p =  a_p^* \, \sqrt{\frac{N-\cN_+}{N}} \]
We will choose $\eta_p = - N^{-2} \widehat{w}_\ell (p/N)$, where $\widehat{w}_\ell$ are the Fourier coefficients of $w_\ell = 1- f_\ell$ and $f_\ell$ is a modification of the solution 
$f$ of the zero-energy scattering equation (\ref{eq:0en}) (more precisely, $f_\ell$ is going to be the Neumann ground state on the ball of radius $N \ell$, for an $\ell$ of order one). We will show in Lemma \ref{3.0.sceqlemma} that, with this definition, $\eta_p \simeq C \kappa / |p|^2$ for $|p| \ll N$, with fast decay for $|p| \gtrsim N$ guaranteeing that $\sum_p p^2 \eta_p^2 \simeq C N$ (the large $p$ behavior of $\eta_p$ corresponds to the $|x|^{-1}$ singularity of (\ref{eq:0enf}), regularized on a length scale of order $N^{-1}$). 

Let us point out that the idea of using unitary operators of the form (\ref{eq:T}) already appeared in \cite{Sei}, in the analysis of the excitation spectrum of mean-field Hamiltonians. In \cite{Sei}, however, these generalized Bogoliubov transformations were used to diagonalize the quadratic part of the excitation Hamiltonian $\cL^\text{mf}_N$, and not, as we do here, to extract additional contributions from cubic and quartic terms in $\cL_N$; as a consequence, in \cite{Sei} the choice of the coefficients $\eta_p$ was very different than in (\ref{eq:T}). 

Since $T$ maps $\cF^{\leq N}_+$ back into itself, we can use it to define a new, modified, excitation Hamiltonian $\cG_N = T^* U_N H_N U_N^* T : \cF_+^{\leq N} \to \cF_+^{\leq N}$. While conjugation with $T$ only creates a finite number of excitations (because $\eta_p$ is square summable; see Lemma \ref{lm:Ngrow}), it extracts an additional energy of order $N$ (because $\sum_p p^2 \eta_p^2 \simeq CN$). Choosing $\eta_p$ as indicated above makes sure that the constant term in $\cG_N$ is exactly $4 \pi a_0 N$ and that all other contributions can be bounded below by the number of particles operator, up to errors of order one. In Proposition \ref{prop:gene} we will conclude that, similarly to (\ref{eq:low-mf}), 
\begin{equation}\label{eq:GN-bd} \cG_N - 4 \pi a_0 N \geq c \cN_+ - C  \end{equation}
for appropriate constants $C,c > 0$ (the proof of Proposition \ref{prop:gene} is given in Section \ref{sec:prop} and represents 
the longest part of the paper). Conjugating (\ref{eq:GN-bd}) 
with $T$ and $U$ (and using the fact that, as discussed in Lemma \ref{lm:Ngrow}, $T$ only changes the number of particles by a multiplicative constant), we arrive at the estimate
\begin{equation}\label{eq:dGamma} H_N - 4 \pi a_0 N \geq c \sum_{j=1}^N  (1-|\ph_0 \rangle \langle \ph_0|)_j - C \end{equation}
between operators acting on the $N$-particle Hilbert space $L^2_s (\Lambda^N)$. For $j=1,\dots , N$, $(1-|\ph_0 \rangle \langle \ph_0|)_j$ indicates the projection $1-|\ph_0 \rangle \langle \ph_0|$ onto the orthogonal complement of the condensate wave function $\ph_0$ acting on the $j$-th particle. In other words, the operator on the r.h.s. of (\ref{eq:dGamma}) measures the number of orthogonal excitations of the condensate. It is then easy to see that (\ref{eq:dGamma}) implies complete Bose-Einstein condensation in the precise sense of (\ref{eq:conv-thm}). 

Technically, the main challenge that we have to face is the fact that the action of the generalized Bogoliubov transformations (\ref{eq:T}) on creation and annihilation operators is not explicit, as it was for (\ref{eq:wtT}). Instead, we will have to expand operators of the form $T^* a_p T$ in absolutely convergent infinite series and we will need to bound several contributions. The main tool we use to control these expansions is Lemma \ref{lm:indu} below, which we take from \cite{BS}.

\section{Fock space}
\label{sec:fock}

Let 
\[ \cF = \bigoplus_{n \geq 0} L^2_s (\Lambda^{n}) = \bigoplus_{n \geq 0} L^2 (\Lambda)^{\otimes_s n} \]
denote the bosonic Fock space over the one-particle space $L^2 (\Lambda)$. Here $L^2_s (\Lambda^{n}) \simeq L^2 (\Lambda)^{\otimes_s n}$ is the subspace of $L^2 (\Lambda^{n})$ consisting of all functions that are symmetric w.r.t. permutations. We use the notation $\Omega = \{ 1, 0, \dots \} \in \cF$ for the vacuum vector. 

For $g \in L^2 (\Lambda)$, we define on $\cF$ the creation operator $a^* (g)$ and the annihilation operator $a(g)$ by 
\[ \begin{split} 
(a^* (g) \Psi)^{(n)} (x_1, \dots , x_n) &= \frac{1}{\sqrt{n}} \sum_{j=1}^n g (x_j) \Psi^{(n-1)} (x_1, \dots , x_{j-1}, x_{j+1} , \dots , x_n) 
\\
(a (g) \Psi)^{(n)} (x_1, \dots , x_n) &= \sqrt{n+1} \int_\Lambda  \overline{g (x)} \Psi^{(n+1)} (x,x_1, \dots , x_n) \, dx \end{split} \]
Creation and annihilation operators satisfy canonical commutation relations
\begin{equation}\label{eq:ccr} [a (g), a^* (h) ] = \langle g,h \rangle , \quad [ a(g), a(h)] = [a^* (g), a^* (h) ] = 0 \end{equation}
for all $g,h \in L^2 (\Lambda)$ (here $\langle g,h \rangle$ denotes the usual inner product on $L^2 (\Lambda)$). 

Since we consider a translation invariant system, it will be useful to work in momentum space. Let $\Lambda^* = 2\pi \bZ^3$. For $p \in \Lambda^*$, we define the normalized wave function $\ph_p  (x) = e^{-ip\cdot x}$ in $L^2 (\Lambda)$ and we set
\begin{equation}\label{eq:ap} a^*_p = a^* (\ph_p), \quad \text{and } \quad  a_p = a (\ph_p) \end{equation} 
In other words, $a^*_p$ and $a_p$ create, respectively, annihilate a particle with momentum $p$. 

In some occasions, it will be also convenient to work in position space (it is easier to make use of the condition that the interaction potential $V(x)$ is pointwise positive when working in position space). To this end, we introduce operator valued distributions $\check{a}_x, \check{a}_x^*$ defined so that
\begin{equation}\label{eq:axf} a(f) = \int \bar{f} (x) \,  \check{a}_x \, dx , \quad a^* (f) = \int f(x) \, \check{a}_x^* \, dx  \end{equation}

On $\cF$, we also introduce the number of particles operator, defined by $(\cN\Psi)^{(n)} = n \Psi^{(n)}$. Notice that 
\[ \cN = \sum_{p \in \Lambda^*} a_p^* a_p  = \int \check{a}^*_x \check{a}_x \, dx \, . \]
It is useful to observe that creation and annihilation operators are bounded by the square root of the number of particles operator, i.e.   
\begin{equation}\label{eq:abd} \| a (f) \Psi \| \leq \| f \| \| \cN^{1/2} \Psi \|, \quad \| a^* (f) \Psi \| \leq \| f \| \| (\cN+1)^{1/2} \Psi \| 
\end{equation}
for all $f \in L^2 (\Lambda)$. 

We will often have to deal with quadratic translation invariant operators on $\cF$ (quadratic in creation and annihilation operators). For $f \in \ell^2 (\Lambda^*)$, we define
\begin{equation}\label{eq:AAdef} A_{\sharp_1, \sharp_2} (f) = \sum_{p \in \Lambda^*} f_p \, a_{\alpha_1 p}^{\sharp_1} a_{\alpha_2 p}^{\sharp_2} \end{equation}
where $\sharp_1, \sharp_2 \in \{ \cdot, * \}$, and we use the notation $a^\sharp = a$, if $\sharp = \cdot$, and $a^\sharp = a^*$ if $\sharp = *$. Also, $\alpha_j \in \{ \pm 1 \}$ is chosen so that $\alpha_1 = 1$, if $\sharp_1 = *$, $\alpha_1 = -1$ if $\sharp_1 = \cdot$, $\alpha_2 = 1$ if $\sharp_2 = \cdot$ and $\alpha_2 = -1$ if $\sharp_2 = *$. Notice that, in position space  
\[ A_{\sharp_1, \sharp_2} (j) = \int dx dy \, \check{f} (x-y) \, \check{a}_x^{\sharp_1} \, \check{a}_y^{\sharp_2} \]
with the inverse Fourier transform \[ \check{f} (x) = \sum_{p \in \Lambda^*} f_p \, e^{ip\cdot x} \, . \]  
\begin{lemma} \label{lm:Abds} Let $f \in \ell^2 (\Lambda^*)$ and, if $\sharp_1 = \cdot$ and $\sharp_2 = *$ assume additionally that $f \in \ell^1 (\Lambda^*)$. Then we have, for any $\Psi \in \cF$, 
\[ \begin{split} 
\| A_{\sharp_1, \sharp_2} (f) \Psi \| &\leq \sqrt{2} \, \| (\cN+1) \Psi \| \left\{ \begin{array}{ll}  \| f \|_2 + \| f \|_1 \, \qquad &\text{if } \sharp_1 = \cdot, \sharp_2 = * \\  \| f \|_2  \qquad &\text{otherwise} \end{array} \right. \end{split} \]
\end{lemma} 

We will need to work on certain subspaces of $\cF$. 
Recall that $\ph_0 \in L^2 (\Lambda)$ is the constant wave function $\ph_0 (x) = 1$ for all $x \in \Lambda$. We denote by $L^2_{\perp} (\Lambda)$ the orthogonal complement of the one dimensional space spanned by $\ph_0$ in $L^2 (\Lambda)$. We define then 
\[ \cF_{+} = \bigoplus_{n \geq 0} L^2_{\perp} (\Lambda)^{\otimes_s n} \]
as the Fock space constructed over $L^2_{\perp} (\Lambda)$. A vector $\Psi = \{ \psi^{(0)}, \psi^{(1)}, \dots \} \in \cF$ lies in $\cF_{+}$, if $\psi^{(n)}$ is orthogonal to $\ph_0$, in each of its coordinate, for all $n \geq 1$, i.e. if
\[ \int_\Lambda \psi^{(n)} (x,y_1, \dots , y_{n-1})\, dx = 0 \]
for all $n \geq 1$. In momentum space, it is very easy to characterize the orthogonal complement of $\ph_0$; it consists of all functions in $\ell^2 (\Lambda^*)$ vanishing at $p=0$. Hence, $\cF_{+}$ is the Fock space generated by creation and annihilation operators $a_p^*, a_p$, for $p \in \Lambda^*_+ := 2\pi \bZ^3 \backslash \{ 0 \}$. On $\cF_+$, we denote the number of particles operator by 
\[ \cN_+ = \sum_{p \in \Lambda_+^*} a_p^* a_p \, . \]

We will also need a truncated version of the Fock space $\cF_+$. For $N \in \bN$, we define 
\[ \cF_{+}^{\leq N} = \bigoplus_{n=0}^N L^2_{\perp} (\Lambda)^{\otimes_s n} \]
as the Fock spaces constructed over $L^2_\perp (\Lambda)$, describing states with at most $N$ particles.  

On $\cF_+^{\leq N}$, we will consider modified creation and annihilation operators. For $f \in L^2_\perp (\Lambda)$, we define 
\[ b (f) = \sqrt{\frac{N- \cN_+}{N}} \, a (f), \qquad \text{and } \quad  b^* (f) = a^* (f) \, \sqrt{\frac{N-\cN_+}{N}} \]
We have $b(f), b^* (f) : \cF_+^{\leq N} \to \cF_+^{\leq N}$. As we will discuss in the next section, the importance of these fields arises from the application of the map $U_N$, defined in (\ref{eq:Uph-def}), since, for example,  
\begin{equation}\label{eq:UaU} U_N a^* (f) a(\ph_0) U_N^* = a^* (f) \sqrt{N-\cN_+}  = \sqrt{N} \, b^* (f) 
\end{equation}
Eq. (\ref{eq:UaU}) clarifies the action of the modified creation and annihilation operators; $b^* (f)$  excites a particle from the condensate into its orthogonal complement while $b(f)$ annihilates an excitation back into the condensate. Compared with the standard fields $a^*,a$, the modified creation and annihilation operators $b^*,b$ have an important  advantage. They create or annihilate an excitation of the condensate but, at the same time, they preserve the total number of particles  (this is why they map $\cF^{\leq N}_+$ into itself). 

It is also convenient to introduce modified creation and annihilation operators in momentum space, setting
\[ b_p = \sqrt{\frac{N-\cN_+}{N}} \, a_p, \qquad \text{and } \quad  b^*_p = a^*_p \, \sqrt{\frac{N-\cN_+}{N}} \]
for all $p \in \Lambda^*_+$ and operator valued distributions 
in position space 
\[ \check{b}_x = \sqrt{\frac{N-\cN_+}{N}} \, \check{a}_x, \qquad \text{and } \quad  \check{b}^*_x = \check{a}^*_x \, \sqrt{\frac{N-\cN_+}{N}} \]
for all $x \in \Lambda$. 

Modified creation and annihilation operators satisfy the commutation relations 
\begin{equation}\label{eq:comm-bp} \begin{split} [ b_p, b_q^* ] &= \left( 1 - \frac{\cN_+}{N} \right) \delta_{p,q} - \frac{1}{N} a_q^* a_p 
\\ [ b_p, b_q ] &= [b_p^* , b_q^*] = 0 
\end{split} \end{equation}
and, in position space, 
\begin{equation}\label{eq:comm-b}
\begin{split}  [ \check{b}_x, \check{b}_y^* ] &= \left( 1 - \frac{\cN_+}{N} \right) \delta (x-y) - \frac{1}{N} \check{a}_y^* \check{a}_x \\ 
[ \check{b}_x, \check{b}_y ] &= [ \check{b}_x^* , \check{b}_y^*] 
= 0 
\end{split} \end{equation}
Furthermore, we find 
\begin{equation}\label{eq:comm-b2}
\begin{split}
[\check{b}_x, \check{a}_y^* \check{a}_z] &=\delta (x-y)\check{b}_z, \qquad 
[\check{b}_x^*, \check{a}_y^* \check{a}_z] = -\delta (x-z) \check{b}_y^*
\end{split} \end{equation}
These expressions easily lead us to $[ \check{b}_x, \cN_+ ] = \check{b}_x$, $[ \check{b}_x^* , \cN_+ ] = - \check{b}_x^*$ and, in momentum space, to $[b_p , \cN_+] = b_p$, $[b_p^*, \cN_+] = - b_p^*$. {F}rom (\ref{eq:abd}), we immediately find that
\begin{equation}\label{lm:bbds}
\begin{split} 
\| b(f) \xi \| &\leq \| f \| \left\| \cN_+^{1/2} \left( \frac{N+1-\cN_+}{N} \right)^{1/2} \xi \right\| \\ 
\| b^* (f) \xi \| &\leq \| f \| \left\| (\cN_+ +1)^{1/2} \left( \frac{N-\cN_+}{N} \right)^{1/2} \xi \right\|
\end{split} \end{equation}
for all $f \in L^2_\perp (\Lambda)$ and $\xi \in \cF_+^{\leq N}$. Since $\cN_+ \leq N$ on $\cF_+^{\leq N}$, it follows that $b(f), b^* (f) : \cF_+^{\leq N} \to \cF_+^{\leq N}$ are bounded operators with $\| b(f) \|, \| b^* (f) \| \leq (N+1)^{1/2} \| f \|$. 

We will also consider quadratic expressions in the $b$-fields. Also in this case, we restrict our attention to translation invariant operators. For $f \in \ell^2 (\Lambda^*_+)$, we define, similarly to (\ref{eq:AAdef}),
\[ B_{\sharp_1, \sharp_2} (f) = \sum_{p \in \Lambda^*} f_p \, b_{\alpha_1 p}^{\sharp_1} \, b_{\alpha_2 p}^{\sharp_2} \]
with $\alpha_1 = 1$ if $\sharp_1 = *$, $\alpha_1 = -1$ if $\sharp_1 = \cdot$, $\alpha_2 = 1$ if $\sharp_2 = \cdot$ and $\alpha_2 = -1$ if $\sharp_2 = *$. By construction, $B_{\sharp_1, \sharp_2} (f) : \cF_+^{\leq N} \to \cF_+^{\leq N}$. In position space, we find 
\[ \begin{split} 
B_{\sharp_1, \sharp_2} (f) &= \int \check{f} (x-y) \, \check{b}^{\sharp_1}_x \check{b}^{\sharp_2}_y \, dx dy  
\end{split}  \]
From Lemma \ref{lm:Abds}, we obtain the following bounds.
\begin{lemma}\label{lm:Bbds}
Let $f \in \ell^2 (\Lambda^*_+)$. If $\sharp_1 = \cdot$ and $\sharp_2 = *$, we assume additionally that $f \in \ell^1 (\Lambda^*)$. Then
\[ \begin{split} \frac{\| B_{\sharp_1,\sharp_2} (f) \xi \|}{\left\| (\cN_+ +1) \left(\frac{N-\cN_+ +2}{N} \right) \xi \right\|} &
\leq \sqrt{2}  \left\{ \begin{array}{ll}  \| f \|_2 + \| f \|_1 \qquad &\text{if } \sharp_1 = \cdot, \sharp_2 = * \\  \| f \|_2  \qquad &\text{otherwise} \end{array} \right. 
\end{split} \]
for all $\xi \in \cF^{\leq N}_+$. Since $\cN_+  \leq N$ on $\cF^{\leq N}_+$, the operator $B_{\sharp_1 , \sharp_2} (f)$ is bounded, with 
\[ \begin{split} \| B_{\sharp_1, \sharp_2} (f) \| &\leq \sqrt{2} N \left\{ \begin{array}{ll} \| f \|_2 + \| f \|_1 \quad &\text{if } \sharp_1 = \cdot, \sharp_2 = * \\  \| f \|_2  \qquad &\text{otherwise} \end{array} \right. \end{split} \]
\end{lemma}

We will need to consider products of several creation and annihilation operators. In particular, two types of monomials in creation and annihilation operators will play an important role in our analysis. For $f_1, \dots , f_n \in \ell_2 (\Lambda^*_+)$, $\sharp = (\sharp_1, \dots , \sharp_n), \flat = (\flat_0, \dots , \flat_{n-1}) \in \{ \cdot, * \}^n$, we set 
\begin{equation}\label{eq:Pi2}
\begin{split}  
\Pi^{(2)}_{\sharp, \flat} &(f_1, \dots , f_n) \\ &= \sum_{p_1, \dots , p_n \in \Lambda^*}  b^{\flat_0}_{\alpha_0 p_1} a_{\beta_1 p_1}^{\sharp_1} a_{\alpha_1 p_2}^{\flat_1} a_{\beta_2 p_2}^{\sharp_2} a_{\alpha_2 p_3}^{\flat_2} \dots  a_{\beta_{n-1} p_{n-1}}^{\sharp_{n-1}} a_{\alpha_{n-1} p_n}^{\flat_{n-1}} b^{\sharp_n}_{\beta_n p_n} \, \prod_{\ell=1}^n f_\ell (p_\ell)  \end{split} \end{equation}
where, for every $\ell=0,1, \dots , n$, we set $\alpha_\ell = 1$ if $\flat_\ell = *$, $\alpha_\ell =    -1$ if $\flat_\ell = \cdot$, $\beta_\ell = 1$ if $\sharp_\ell = \cdot$ and $\beta_\ell = -1$ if $\sharp_\ell = *$. In (\ref{eq:Pi2}), we impose the condition that for every $j=1,\dots, n-1$, we have either $\sharp_j = \cdot$ and $\flat_j = *$ or $\sharp_j = *$ and $\flat_j = \cdot$ (so that the product $a_{\beta_\ell p_\ell}^{\sharp_\ell} a_{\alpha_\ell p_{\ell+1}}^{\flat_\ell}$ always preserves the number of particles, for all $\ell =1, \dots , n-1$). With this assumption, we find that the operator $\Pi^{(2)}_{\sharp,\flat} (f_1, \dots , f_n)$ maps $\cF_+^{\leq N}$ into itself. If, for some $\ell=1, \dots , n$, $\flat_{\ell-1} = \cdot$ and $\sharp_\ell = *$ (i.e. if the product $a_{\alpha_{\ell-1} p_\ell}^{\flat_{\ell-1}} a_{\beta_\ell p_\ell}^{\sharp_\ell}$ for $\ell=2,\dots , n$, or the product $b_{\alpha_0 p_1}^{\flat_0} a_{\beta_1 p_1}^{\sharp_1}$ for $\ell=1$, is not normally ordered) we require additionally that $f_\ell  \in \ell^1 (\Lambda^*_+)$. In position space, the same operator can be written as 
\begin{equation}\label{eq:Pi2-pos} \Pi^{(2)}_{\sharp, \flat} (j_1, \dots , j_n) = \int   \check{b}^{\flat_0}_{x_1} \check{a}_{y_1}^{\sharp_1} \check{a}_{x_2}^{\flat_1} \check{a}_{y_2}^{\sharp_2} \check{a}_{x_3}^{\flat_2} \dots  \check{a}_{y_{n-1}}^{\sharp_{n-1}} \check{a}_{x_n}^{\flat_{n-1}} \check{b}^{\sharp_n}_{y_n} \, \prod_{\ell=1}^n \check{f}_\ell (x_\ell - y_\ell) \, dx_\ell dy_\ell \end{equation}
An operator of the form (\ref{eq:Pi2}), (\ref{eq:Pi2-pos}) with all the properties listed above, will be called a $\Pi^{(2)}$-operator of order $n$.

For $g, f_1, \dots , f_n \in \ell_2 (\Lambda^*_+)$, $\sharp = (\sharp_1, \dots , \sharp_n)\in \{ \cdot, * \}^n$, $\flat = (\flat_0, \dots , \flat_{n}) \in \{ \cdot, * \}^{n+1}$, we also define the operator 
\begin{equation}\label{eq:Pi1}
\begin{split} \Pi^{(1)}_{\sharp,\flat} &(f_1, \dots , f_n;g) \\ &= \sum_{p_1, \dots , p_n \in \Lambda^*}  b^{\flat_0}_{\alpha_0, p_1} a_{\beta_1 p_1}^{\sharp_1} a_{\alpha_1 p_2}^{\flat_1} a_{\beta_2 p_2}^{\sharp_2} a_{\alpha_2 p_3}^{\flat_2} \dots a_{\beta_{n-1} p_{n-1}}^{\sharp_{n-1}} a_{\alpha_{n-1} p_n}^{\flat_{n-1}} a^{\sharp_n}_{\beta_n p_n} a^{\flat n} (g) \, \prod_{\ell=1}^n f_\ell (p_\ell) \end{split} \end{equation}
where $\alpha_\ell$ and $\beta_\ell$ are defined as above. Also here, we impose the condition that, for all $\ell = 1, \dots , n$, either $\sharp_\ell = \cdot$ and $\flat_\ell = *$ or $\sharp_\ell = *$ and $\flat_\ell = \cdot$. This implies that $\Pi^{(1)}_{\sharp,\flat} (f_1, \dots , f_n;g)$ maps $\cF_+^{\leq N}$ back into $\cF_+^{\leq N}$. Additionally, we assume that $f_\ell \in \ell^1 (\Lambda^*)$, if $\flat_{\ell-1} = \cdot$ and $\sharp_\ell = *$ for some $\ell = 1,\dots , n$ (i.e. if the pair $a_{\alpha_{\ell-1} p_\ell}^{\flat_{\ell-1}} a^{\sharp_\ell}_{\beta_\ell p_\ell}$ is not normally ordered). In position space, the same operator can be written as
\begin{equation}\label{eq:Pi1-pos} \Pi^{(1)}_{\sharp,\flat} (f_1, \dots ,f_n;g) = \int \check{b}^{\flat_0}_{x_1} \check{a}_{y_1}^{\sharp_1} \check{a}_{x_2}^{\flat_1} \check{a}_{y_2}^{\sharp_2} \check{a}_{x_3}^{\flat_2} \dots  \check{a}_{y_{n-1}}^{\sharp_{n-1}} \check{a}_{x_n}^{\flat_{n-1}} \check{a}^{\sharp_n}_{y_n} \check{a}^{\flat n} (g) \, \prod_{\ell=1}^n \check{f}_\ell (x_\ell - y_\ell) \, dx_\ell dy_\ell \end{equation}
An operator of the form (\ref{eq:Pi1}), (\ref{eq:Pi1-pos}) will be called a $\Pi^{(1)}$-operator of order $n$. Operators of the form $b(\check{f})$, $b^* (\check{f})$, for a $f \in \ell^2 (\Lambda^*_+)$, will be called $\Pi^{(1)}$-operators of order zero. 

In the next lemma we show how to bound $\Pi^{(2)}$- and $\Pi^{(1)}$-operators. The simple proof, based on Lemma \ref{lm:Abds}, can be found in \cite{BS}. 
\begin{lemma}\label{lm:Pi-bds}
Let $n \in \bN$, $g, f_1, \dots , f_n \in \ell^2 (\Lambda^*_+)$, $\xi \in \cF_+^{\leq N}$. Let $\Pi^{(2)}_{\sharp,\flat} (f_1,\dots , f_n)$ and $\Pi^{(1)}_{\sharp,\flat} (f_1,\dots, f_n ; g)$ be defined as in (\ref{eq:Pi2}), (\ref{eq:Pi1}). Then  
\begin{equation}\label{eq:Pi-bds} \begin{split} 
\left\| \Pi^{(2)}_{\sharp,\flat} (f_1,\dots ,f_n) \xi \right\| &\leq 6^n \prod_{\ell=1}^n K_\ell^{\flat_{\ell-1}, \sharp_\ell} \left\| (\cN_+ +1)^n \left(1- \frac{\cN_+ -2}{N} \right) \xi \right\| \\
 \left\| \Pi^{(1)}_{\sharp,\flat} (f_1,\dots , f_n;g) \xi \right\| 
 &\leq 6^n \| g \| \prod_{\ell=1}^n K_\ell^{\flat_{\ell-1}, \sharp_\ell} \left\| (\cN_+ +1)^{n+1/2} \left(1- \frac{\cN_+ -2}{N} \right)^{1/2} \xi \right\| 
\end{split} \end{equation} 
where
\[ K_\ell^{\flat_{\ell-1}, \sharp_\ell} = \left\{ \begin{array}{ll} \|f_\ell \|_2 + \| f_\ell \|_1  \quad &\text{if } \flat_{\ell-1} = \cdot \text{ and } \sharp_\ell = * \\
\| f_\ell \|_2 \quad &\text{otherwise} \end{array} \right. \]
Since $\cN_+  \leq N$ on $\cF_+^{\leq N}$, it follows that  
\[ \begin{split} 
\left\|  \Pi^{(2)}_{\sharp,\flat} (f_1,\dots, f_n) \right\| &\leq (12 N)^n \prod_{\ell=1}^n K_\ell^{\flat_{\ell-1}, \sharp_\ell} \\ 
\left\|  \Pi^{(1)}_{\sharp,\flat} (f_1,\dots ,f_n;g) \right\| &\leq (12 N)^n \sqrt{N} \| g \| \prod_{\ell=1}^n K_\ell^{\flat_{\ell-1}, \sharp_\ell} 
\end{split} \]
\end{lemma}

To conclude this section, we introduce generalized Bogoliubov transformations, and we discuss their main properties. For $\eta \in \ell^2 (\Lambda^*_+)$ with $\eta_{-p} = \eta_{p}$ for all $p \in \Lambda^*_+$, we define 
\begin{equation}\label{eq:defB} 
B(\eta) = \frac{1}{2} \sum_{q\in \Lambda^*_+}  \left( \eta_q b_q^* b_{-q}^* - \overline{\eta}_q b_q b_{-q} \right)  \end{equation}
and the unitary operator 
\begin{equation}\label{eq:eBeta} 
e^{B(\eta)} = \exp \left\{ \frac{1}{2} \sum_{q \in \Lambda^*_+}   \left( \eta_q b_q^* b_{-q}^* - \overline{\eta}_q  b_q b_{-q} \right) \right\}
\end{equation}
Notice that $B(\eta), e^{B(\eta)} : \cF_{+}^{\leq N} \to \cF_{+}^{\leq N}$. We will call unitary operators of the form (\ref{eq:eBeta}) generalized Bogoliubov transformations. The name arises from the observation that, on states with $\cN_+ \ll N$, we can expect that $b_q \simeq a_q$, $b^*_q \simeq a_q^*$ and therefore that 
\[ B(\eta) \simeq \wt{B} (\eta) = \frac{1}{2} \sum_{q \in \Lambda^*_+} \left( \eta_q a_q^* a_{-q}^* - \overline{\eta}_q a_q a_{-q} \right)  \]
Since $\wt{B} (\eta)$ is quadratic in creation and annihilation operators, the unitary operator $\exp (\wt{B} (\eta))$ is a standard Bogoliubov transformation, whose action on creation and annihilation operators is explicitly given by 
\begin{equation}\label{eq:act-Bog} e^{-\wt{B} (\eta)} a_p e^{\wt{B} (\eta)} = \cosh (\eta_p) a_p + \sinh (\eta_p) a_{-p}^*  
\end{equation}
As explained in the introduction, since the Bogoliubov transformation in (\ref{eq:act-Bog}) does not map $\cF_+^{\leq N}$ in itself, in the following it will be convenient for us to work with generalized Bogoliubov transformations of the form (\ref{eq:eBeta}). The price we have to pay is the fact that there is no explicit expression like (\ref{eq:act-Bog}) for the action of (\ref{eq:eBeta}). Hence, we  need other tools to control the action of generalized Bogoliubov transformations. A first result, whose proof can be found in \cite{BS} and which will play an important role in the sequel, is the fact that conjugating with (\ref{eq:eBeta}) does not change the momenta of the number of particles operator substantially, if $\eta \in \ell^2 (\Lambda^*_+)$ (the same result was previously established in \cite{Sei}).
\begin{lemma}\label{lm:Ngrow}
Let $\eta \in \ell^2 (\Lambda^*_+)$ and $B(\eta)$ as in (\ref{eq:defB}). Then, for every $n_1, n_2 \in \bZ$, there exists a constant $C > 0$ (depending also on $\| \eta\|$)  such that 
\[ e^{-B(\eta)} (\cN_+ +1)^{n_1} (N+1-\cN_+ )^{n_2} e^{B(\eta)} \leq C (\cN_+ +1)^{n_1} (N+1- \cN_+ )^{n_2} \]
on $\cF^{\leq N}_+$.
\end{lemma}

Controlling the change of the number of particles operator is not enough for our purposes. Instead, we will often need to express the action of generalized Bogoliubov transformations by means of convergent series of nested commutators. We start by noticing that, for any $p \in \Lambda^*_+$, 
\[\begin{split} e^{-B(\eta)} b_p e^{B(\eta)} &= b_p + \int_0^1 ds \, \frac{d}{ds}  e^{-sB(\eta)} b_p e^{sB(\eta)} \\ &= b_p - \int_0^1 ds \, e^{-sB(\eta)} [B(\eta), b_p] e^{s B(\eta)} \\ &= b_p - [B(\eta),b_p] + \int_0^1 ds_1 \int_0^{s_1} ds_2 \, e^{-s_2 B(\eta)} [B(\eta), [B(\eta),b_p ]] e^{s_2 B(\eta)} 
\end{split} \]
Iterating $m$ times, we obtain
\begin{equation}\label{eq:BCH} \begin{split} 
e^{-B(\eta)} b_p e^{B(\eta)} = &\sum_{n=1}^{m-1} (-1)^n \frac{\text{ad}^{(n)}_{B(\eta)} (b_p)}{n!} \\ &+ \int_0^{1} ds_1 \int_0^{s_1} ds_2 \dots \int_0^{s_{m-1}} ds_m \, e^{-s_m B(\eta)} \text{ad}^{(m)}_{B(\eta)} (b_p) e^{s_m B(\eta)} \end{split} \end{equation}
where we introduced the notation $\text{ad}_{B(\eta)}^{(n)} (A)$  defined recursively by
\[ \text{ad}_{B(\eta)}^{(0)} (A) = A \quad \text{and } \quad \text{ad}^{(n)}_{B(\eta)} (A) = [B(\eta), \text{ad}^{(n-1)}_{B(\eta)} (A) ]  \]
We will show later that, under suitable assumptions on $\eta$, the error term on the r.h.s. of (\ref{eq:BCH}) is negligible in the limit $m \to \infty$. This means that the action of the generalized Bogoliubov transformation $e^{B(\eta)}$ on $b_p$ and similarly on $b^*_p$ can be described in terms of the nested commutators $\text{ad}^{(n)}_{B(\eta)} (b_p)$ and $\text{ad}^{(n)}_{B(\eta)} (b^*_p)$. In the next lemma, we give a detailed analysis of these operators. 
\begin{lemma}\label{lm:indu}
Let $\eta \in \ell^2 (\Lambda^*_+)$ be such that $\eta_p = \eta_{-p}$ for all $p \in \ell^2 (\Lambda^*)$. To simplify the notation, assume also $\eta$ to be real-valued (as it will be in applications). Let $B(\eta)$ be defined as in (\ref{eq:defB}), $n \in \bN$ and $p \in \Lambda^*_+$. Then the nested commutator $\text{ad}^{(n)}_{B(\eta)} (b_p)$ can be written as the sum of exactly $2^n n!$ terms, with the following properties. 
\begin{itemize}
\item[i)] Possibly up to a sign, each term has the form
\begin{equation}\label{eq:Lambdas} \Lambda_1 \Lambda_2 \dots \Lambda_i \, N^{-k} \Pi^{(1)}_{\sharp,\flat} (\eta^{j_1}, \dots , \eta^{j_k} ; \eta^{s}_p \ph_{\alpha p}) 
\end{equation}
for some $i,k,s \in \bN$, $j_1, \dots ,j_k \in \bN \backslash \{ 0 \}$, $\sharp \in \{ \cdot, * \}^k$, $ \flat \in \{ \cdot, * \}^{k+1}$ and $\alpha \in \{ \pm 1 \}$ chosen so that $\alpha = 1$ if $\flat_k = \cdot$ and $\alpha = -1$ if $\flat_k = *$ (recall here that $\ph_p (x) = e^{-ip \cdot x}$). In (\ref{eq:Lambdas}), each operator $\Lambda_w : \cF^{\leq N}_+ \to \cF^{\leq N}_+$, $w=1, \dots , i$, is either a factor $(N-\cN_+ )/N$, a factor $(N+1-\cN_+ )/N$ or an operator of the form
\begin{equation}\label{eq:Pi2-ind} N^{-h} \Pi^{(2)}_{\sharp',\flat'} (\eta^{z_1}, \eta^{z_2},\dots , \eta^{z_h}) \end{equation}
for some $h, z_1, \dots , z_h \in \bN \backslash \{ 0 \}$, $\sharp,\flat  \in \{ \cdot , *\}^h$. 
\item[ii)] If a term of the form (\ref{eq:Lambdas}) contains $m \in \bN$ factors $(N-\cN_+ )/N$ or $(N+1-\cN_+ )/N$ and $j \in \bN$ factors of the form (\ref{eq:Pi2-ind}) with $\Pi^{(2)}$-operators of order $h_1, \dots , h_j \in \bN \backslash \{ 0 \}$, then 
we have
\begin{equation}\label{eq:totalb} m + (h_1 + 1)+ \dots + (h_j+1) + (k+1) = n+1 \end{equation}
\item[iii)] If a term of the form (\ref{eq:Lambdas}) contains (considering all $\Lambda$- and $\Pi^{(1)}$-operators) the arguments $\eta^{i_1}, \dots , \eta^{i_m}$ and the factor $\eta^{s}_p$ for some $m, s \in \bN$ and some $i_1, \dots , i_m \in \bN \backslash \{ 0 \}$, then \[ i_1 + \dots + i_m + s = n .\]
\item[iv)] There is exactly one term having the form  
\begin{equation}\label{eq:iv1} \left(\frac{N-\cN_+ }{N} \right)^{n/2} \left(\frac{N+1-\cN_+ }{N} \right)^{n/2} \eta^{n}_p b_p 
\end{equation}
if $n$ is even, and 
\begin{equation}\label{eq:iv2} - \left(\frac{N-\cN_+ }{N} \right)^{(n+1)/2} \left(\frac{N+1-\cN_+ }{N} \right)^{(n-1)/2} \eta^{n}_p b^*_{-p}  \end{equation}
if $n$ is odd.
\item[v)] If the $\Pi^{(1)}$-operator in (\ref{eq:Lambdas}) is of order $k \in \bN \backslash \{ 0 \}$, it has either the form  
\[ \sum_{p_1, \dots , p_k}  b^{\flat_0}_{\alpha_0 p_1} \prod_{i=1}^{k-1} a^{\sharp_i}_{\beta_i p_{i}} a^{\flat_i}_{\alpha_i p_{i+1}}  a^*_{-p_k} \eta^{2r}_p  a_p \prod_{i=1}^k \eta^{j_i}_{p_i}  \]
or the form 
\[\sum_{p_1, \dots , p_k} b^{\flat_0}_{\alpha_0 p_1} \prod_{i=1}^{k-1} a^{\sharp_i}_{\beta_i p_{i}} a^{\flat_i}_{\alpha_i p_{i+1}}  a_{p_k} \eta^{2r+1}_p a^*_{-p} \prod_{i=1}^k \eta^{j_i}_{p_i}  \]
for some $r \in \bN$, $j_1, \dots , j_k \in \bN \backslash \{ 0 \}$. If it is of order $k=0$, then it is either given by $\eta^{2r}_p b_p$ or by $\eta^{2r+1}_p b_{-p}^*$, for some $r \in \bN$. 
\item[vi)] For every non-normally ordered term of the form 
\[ \begin{split} &\sum_{q \in \Lambda^*} \eta^{i}_q a_q a_q^* , \quad \sum_{q \in \Lambda^*} \, \eta^{i}_q b_q a_q^* \\  &\sum_{q \in \Lambda^*} \, \eta^{i}_q a_q b_q^*, \quad \text{or } \quad \sum_{q \in \Lambda^*} \, \eta^{i}_q b_q b_q^*  \end{split} \]
appearing either in the $\Lambda$-operators or in the $\Pi^{(1)}$-operator in (\ref{eq:Lambdas}), we have $i \geq 2$.
\end{itemize}
\end{lemma}

\begin{proof}
The proof is a translation in momentum space of the proof of Lemma 3.2 in \cite{BS}. For completeness, we repeat here the main steps. We proceed by induction. For $n=0$ the claims are clear. For the induction from $n$ to 
$n+1$ we will repeatedly use the relations
\begin{equation}\label{2.2.Betacommutators}
    \begin{split}
    [B(\eta), b_p ] &= - \frac{(N-\cN_+)}{N} \eta_p b^*_{-p}  +  \frac{1}{N}\sum_{q\in\Lambda_+^*} b^*_q a^*_{-q} a_p  \eta_q \\
    &=-\eta_p b^*_{-p}\frac{(N-\cN_+ +1)}{N}  + \frac{1}{N}\sum_{q\in\Lambda_+^*}  a_p a^*_{-q} b^*_q \eta_q ,\\
    [B(\eta), b^*_p ] &= -\eta_p b_{-p}\frac{(N-\cN_+)}{N}  +  \frac{1}{N} \sum_{q\in\Lambda_+^*}a^*_p  a_{-q} b_q  \eta_q\\
    &=- \frac{(N-\cN_+ +1)}{N} \eta_p b_{-p}  +  \frac{1}{N}
    \sum_{q\in\Lambda_+^*}b_q a_{-q}   a^*_p   \eta_q ,\\
    [B(\eta), a^*_p a_q] & = [B(\eta), a_q a^*_p]= - b^*_p b_{-q}^*\eta_q - \eta_p b_{-p} b_q,\\
    [B(\eta), N-\cN_+ ]&= \sum_{q\in\Lambda_+^*}\eta_q(b^*_qb^*_{-q} +b_qb_{-q} ).
    \end{split}
    \end{equation}
Since $ \operatorname {ad}_{B(\eta)}^{(n+1)}(b_p)=[B(\eta),\operatorname{ad}_{B(\eta)}^{(n)}(b_p)]$, by linearity it is enough to analyze
    \begin{equation} \label{eq:comm-step}
    \left[ B(\eta), \Lambda_1 \Lambda_2\dots \Lambda_i N^{-k} \Pi^{(1)}_{\sharp,\flat} (\eta^{j_1}, \dots , \eta^{j_k} ; \eta^{s}_p \ph_{\alpha p}) \right] 
\end{equation}
with $ \Lambda_1 \Lambda_2\dots \Lambda_i N^{-k} \Pi^{(1)}_{\sharp,\flat} (\eta^{j_1}, \dots , \eta^{j_k} ; \eta^{s}_p \ph_{\alpha p})$ satisfying properties (i) to (vi). By Leibniz, the commutator (\ref{eq:comm-step}) is a sum of terms, where $B(\eta)$ is either commuted with a $\Lambda$-operator, or with the $\Pi^{(1)}$-operator. 

First, consider the case that $B(\eta)$ is commuted with a $\Lambda$-operator. If  $\Lambda$ is either equal $(N-\cN_+)/N$ or to $(N+1-\cN_+)/N$, the last identity in \eqref{2.2.Betacommutators} implies that, after commutation with $B(\eta)$, $\Lambda$ should be replaced by \begin{equation}\label{eq:Pi2-repl} N^{-1} \Pi^{(2)}_{*,*} (\eta) + N^{-1}\Pi^{(2)}_{\cdot, \cdot} ({\eta})
\end{equation}
This generates two terms contributing to $\text{ad}^{(n+1)}_{B(\eta)} (b_p)$. Let us check that these new terms satisfy (i)-(vi), with $n$ replaced by $(n+1)$. (i) is obviously true. Also (ii) remains true because, when replacing $(N-\cN_+)/N$ or $(N+1-\cN_+)/N$ by one of the two summands in (\ref{eq:Pi2-repl}), the index $m$ decreases by one but, at the same time, we have one more $\Pi^{(2)}$-operator of order one (which means that $j$ is replaced by $j+1$, and that there is an additional factor $h_{j+1} + 1 = 2$ in the sum (\ref{eq:totalb})). Since exactly one additional factor $\eta$ is inserted, also (iii) remains true. The $\Pi^{(1)}$-operator is not affected by the replacement, so also (v) continues to hold true. Since both terms in (\ref{eq:Pi2-repl}) are normally ordered, (vi) remains valid as well, by the induction assumption. Finally, the two terms generated in (\ref{eq:Pi2-repl}) are not of the form appearing in (iv).
    
Next, we consider the commutator of $B(\eta)$ with an operator of the form (\ref{eq:Pi2-ind}) for some $h\in \bN$, with $h\leq n$ by (ii). By definition
    \begin{equation}\label{eq:L2}
    \Lambda =  N^{-h} \sum_{p_1, \dots , p_h \in \Lambda^*}  b^{\flat'_0}_{\alpha_0 p_1} a_{\beta_1 p_1}^{\sharp'_1} a_{\alpha_1 p_2}^{\flat'_1} a_{\beta_2 p_2}^{\sharp'_2} a_{\alpha_2 p_3}^{\flat'_2} \dots  a_{\beta_{h-1} p_{h-1}}^{\sharp'_{h-1}} a_{\alpha_{h-1} p_h}^{\flat'_{h-1}} b^{\sharp'_h}_{\beta_h p_h} \, \prod_{\ell=1}^h \eta^{z_l}_{p_l}  
    \end{equation}
When $ [B(\eta), \cdot]$ hits $ b^{\flat'_0}_{\alpha_0p_1}$, the first two equations in \eqref{2.2.Betacommutators} imply that $\Lambda$ is replaced by 
the sum of two operators. The first operator is either 
\begin{equation}\label{eq:first-L} \begin{split} -\frac{N-\cN_+}{N} N^{-h} \Pi^{(2)}_{\sharp^{'},\wt{\flat^{'}}} &(\eta^{z_1+1}, \eta^{z_2},\dots , \eta^{z_h}) \quad \text{or } \\  &\quad -\frac{N-\cN_++1}{N} N^{-h} \Pi^{(2)}_{\sharp^{'},\wt{\flat^{'}}} (\eta^{z_1+1}, \eta^{z_2},\dots , \eta^{z_h}) \end{split} \end{equation}
depending on whether $\flat^{'}_0 = \cdot$ or $\flat^{'}_0 = *$ (here $\wt{\flat^{'}} = (\bar{\flat'_0}, \flat^{'}_1, \dots , \flat^{'}_{h-1})$ with $\bar{\flat^{'}_0} = \cdot$ if $\flat^{'}_0 = *$ and $\bar{\flat^{'}_0} = *$ if $\flat^{'}_0 = \cdot$). The second operator is a $\Pi^{(2)}$-operator of order $(h+1)$, given by
\begin{equation}\label{eq:second-L} N^{-(h+1)} \Pi^{(2)}_{\wt{\sharp'},\wt{\flat'}} (\eta, \eta^{z_1}, \eta^{z_2},\dots , \eta^{z_h}) \end{equation}
where $\wt{\sharp} = (\bar{\flat^{'}_0}, \sharp^{'}_1, \dots , \sharp^{'}_h)$, $\wt{\flat} = (\bar{\flat^{'}_0}, \flat^{'}_0, \dots , \flat^{'}_{h-1})$. 

In both cases (i) is clearly correct and (ii) remains true as well (when we replace (\ref{eq:L2}) with (\ref{eq:first-L}), the number of $(N-\cN_+)/N$ or $(N-\cN_++1)/N$-operators increases by one, while everything else remains unchanged; similarly, when we replace (\ref{eq:L2}) with (\ref{eq:second-L}), the order of the $\Pi^{(2)}$-operator increases by one, while the rest remains unchanged). (iii) also remains true, since in (\ref{eq:first-L}) the power $z_1 +1$ of the first $\eta$-kernel is increased by one unit and in (\ref{eq:second-L}) there is one additional factor $\eta$, compared with (\ref{eq:L2}). (v) remains valid, since the $\Pi^{(1)}$-operator on the right is not affected by this commutator. (vi) remains true in (\ref{eq:first-L}), because $z_1+1 \geq 2$ . It remains true also in (\ref{eq:second-L}). In fact, according to (\ref{2.2.Betacommutators}), when switching from (\ref{eq:L2}) to (\ref{eq:second-L}), we are effectively replacing $b \to b^* a^* a$ or $b^* \to b a a^*$. Hence, the first pair of operators in (\ref{eq:second-L}) is always normally ordered. As for the second pair of creation and annihilation operators (the one associated with the function $\eta^{z_1}$ in (\ref{eq:second-L})), the first field is of the same type as the original $b$-field appearing in (\ref{eq:L2}); non-normally ordered pairs cannot be created. Finally, we remark that the terms we generated here are certainly not of the form in (iv).

The same arguments can be applied if $B(\eta)$ hits the factor $b^{\sharp'_h}_{\beta_h p_h}$ on the right of (\ref{eq:L2}) (in this case, we use the identities for the first two commutators in \eqref{2.2.Betacommutators} having the $b$-field to the left of the factors $(N+1-\cN_+)/N$ and $(N-\cN_+)/N$ and to the right of the $a_{p} a_{-q}^*$ and $a_p^* a_{-q}$ operators). 

If instead $B(\eta)$ hits a term $a^*_{p_{r}}a_{p_{r+1}}$ or $a_{p_{r}} a^*_{p_{r+1}}$ in (\ref{eq:L2}), for an $r=1, \dots, h-1$, then, by \eqref{2.2.Betacommutators}, $\Lambda$ is replaced by the sum of 
the two terms, given by
\begin{equation}\label{eq:replaa} - \left[ N^{-r} \Pi^{(2)}_{\sharp^{''},\flat^{''}} (\eta^{z_1}, \eta^{z_2},\dots , \eta^{z_r+1}) \right] \left[ N^{-(h-r)}  \Pi^{(2)}_{\sharp^{'''},\flat^{'''} } (\eta^{z_{r+1}}, \eta^{z_2},\dots , \eta^{z_h}) \right] \end{equation}
and by
\begin{equation}\label{eq:replaa2} - \left[ N^{-r} \Pi^{(2)}_{\sharp^{''''}, \flat^{''}} (\eta^{z_1}, \eta^{z_2},\dots , \eta^{z_r}) \right] \left[ N^{-(h-r)} \Pi^{(2)}_{\sharp^{'''},\flat^{''''}} (\eta^{z_{r+1}+1}, \eta^{z_2},\dots , \eta^{z_h}) \right] \end{equation}
with $\flat^{''} = (\flat^{'}_0, \dots , \flat^{'}_{r-1})$, $\flat^{'''} = (\flat^{'}_r, \dots , \flat^{'}_{h-1})$, $\flat^{''''} = (\bar{\flat^{'}}_r, \flat^{'}_{r+1}, \dots , \flat^{'}_{h-1})$ and with $\sharp^{''} = (\sharp^{'}_1, \dots , \sharp^{'}_{r-1},\bar{\sharp^{'}}_r)$, $\sharp^{'''} = (\sharp^{'}_{r+1}, \dots , \sharp^{'}_h)$, $\sharp^{''''} = (\sharp^{'}_1, \dots , \sharp^{'}_r)$ (here, we denote $\bar{\sharp^{'}}_r = *$ if $\sharp^{'}_r = \cdot$ and $\bar{\sharp^{'}}_r = \cdot$ if $\sharp^{'}_r = *$, and similarly for $\bar{\flat^{'}}_{r-1}$). 
Obviously, the new terms containing (\ref{eq:replaa}) and (\ref{eq:replaa2}) satisfy (i). (ii) remains valid since the contribution of the original $\Lambda$ to the sum in (\ref{eq:totalb}), which was given by $(h+1)$ is now given by $(r+1)+(h-r+1) = h+2$. Also (iii) continues to be true, because for both terms (\ref{eq:replaa}) and (\ref{eq:replaa2}), there is one new additional factor $\eta$. Moreover, the terms we generated do not have the form (iv). Since the $\Pi^{(1)}$-operator is unaffected, (v) remains true. As for (vi), we observe that non-normally ordered pairs can only be created where $\sharp^{'}_r$ is changed to $\bar{\sharp^{'}}_{r}$ (in the term where $\sharp^{''}$ appears) or where $\flat^{'}_r$ is changed to $\bar{\flat^{'}}_r$ (in the term where $\flat'''$ appears). In both cases, however, the change $\sharp^{'}_r \to \bar{\sharp^{'}}_r$ and $\flat^{'}_r \to \bar{\flat^{'}}_r$ comes together with an increase in the power of $\eta$ (i.e. $\eta^{z_r}$ is changed to $\eta^{z_r+1}$ in the first case, while $\eta^{z_{r+1} }$ is changed to $\eta^{z_{r+1} + 1}$ in the second case). Since $z_r + 1, z_{r+1} + 1 \geq 2$, (vi) is still satisfied.

Next, let us consider the terms arising from commuting $B(\eta)$ with the operator 
\begin{equation} \label{eq:Pi1-term}
\begin{split} 
& N^{-k} \Pi^{(1)}_{\sharp,\flat} (\eta^{j_1}, \dots , \eta^{j_k} ; \eta^{s}_p \ph_{\alpha p}) \\ &= \sum_{p_1, \dots , p_k \in \Lambda^*}  b^{\flat_0}_{\alpha_0 p_1} a_{\beta_1 p_1}^{\sharp_1} a_{\alpha_1 p_2}^{\flat_1} a_{\beta_2 p_2}^{\sharp_2} a_{\alpha_2 p_3}^{\flat_2} \dots a_{\beta_{k-1} p_{k-1}}^{\sharp_{k-1}} a_{\alpha_{k-1} p_k}^{\flat_{k-1}} a^{\sharp_k}_{\beta_k p_k} a^{\flat_k}_{\alpha p}  \eta^{s}_p \, \prod_{\ell=1}^n \eta^{j_l}_{p_l} \end{split}\end{equation}
The arguments are very similar to the case when $B(\eta)$ is commuted with a $\Pi^{(2)}$-operator of the form (\ref{eq:L2}). In particular, if $B(\eta)$ hits $b^{\flat_0}_{\alpha_0 p_1} $, (\ref{eq:Pi1-term}) is replaced by the sum of two terms, the first one being 
\[\begin{split}  - \frac{N-\cN_+}{N} N^{-k} \Pi^{(1)}_{\sharp,\wt{\flat}} (\eta^{j_1+1},& \dots , \eta^{j_k} ; \eta^{s}_p \ph_{\alpha p}) \qquad \text{or } \\ &- \frac{N-\cN_++1}{N} N^{-k} \Pi^{(1)}_{\sharp,\wt{\flat}} (\eta^{j_1+1}, \dots , \eta^{j_k} ; \eta^{s}_p \ph_{\alpha p}) \end{split} \]
depending on whether $\flat_0 = \cdot$ or $\flat_0 = *$ (with $\wt{\flat} = (\bar{\flat}_0, \flat_1, \dots,  \flat_{k-1})$) and the second one being 
\[ N^{-(k+1)} \Pi^{(1)}_{\wt{\sharp}, \wt{\flat}} (\eta, \eta^{j_1}, \dots , \eta^{j_k} ; \eta^{s}_p \ph_{\alpha p})\]
with $\wt{\sharp} = (\bar{\flat}_0, \sharp_1, \dots , \sharp_k)$ and $\wt{\flat} = (\bar{\flat}_0, \flat_1, \dots , \flat_{k})$. As for (\ref{eq:first-L}) and (\ref{eq:second-L}) above, one can show that (i), (ii), (iii), (v), (vi) remain valid. Property (iv) will be discussed below.  

If $B(\eta)$ is commuted with one of the factors $a_{p_r}^{\sharp_r} a_{p_{r+1}}^{\flat_r}$ for an $r=1, \dots , k-1$, the resulting two terms will be given by 
\begin{equation}\label{eq:replaa-Pi1} - \left[ N^{-r} \Pi^{(2)}_{\sharp^{''}, \flat^{''} } (\eta^{j_1}, \dots , \eta^{j_r+1} ; \eta^{s}_p \ph_{\alpha p}) \right] \left[ N^{-(k-r)} \Pi^{(1)}_{\sharp^{'''},\flat^{'''}} (\eta^{j_{r+1}}, \dots , \eta^{j_k} ; \eta^{s}_p \ph_{\alpha p}) \right] 
\end{equation}
and by
\begin{equation}\label{eq:replaa2-Pi1} - \left[ N^{-r} \Pi^{(2)}_{\sharp^{''''}, \flat^{''}} (\eta^{j_1}, \dots , \eta^{j_r} ; \eta^{s}_p \ph_{\alpha p}) \right] \left[ N^{-(k-r)} \Pi^{(1)}_{\sharp^{'''},\flat^{''''}} (\eta^{j_{r+1}+1}, \dots , \eta^{j_k} ; \eta^{s}_p \ph_{\alpha p}) \right] \end{equation}
with $\sharp^{''}, \sharp^{'''},\sharp^{''''}$ and $\flat^{''},\flat^{'''}, \flat^{''''}$ as defined after (\ref{eq:replaa2}). Proceeding analogously as for (\ref{eq:replaa2}), these terms satisfy (i),(ii),(iii),(v),(vi).

Let us next consider the case that (\ref{eq:Pi1-term}) hits the last pair of operators appearing in (\ref{eq:Pi1-term}). From the induction assumption, this pair either equals $\eta^{2r} a_{p_k}^* a_p$ or $\eta^{2r+1} a_{p_k} a^*_{-p}$. In the first case, (\ref{eq:Pi1-term}) is replaced by
\begin{equation}\label{eq:restaa-Pi1l} - \Pi^{(2)}_{\sharp, \flat'} (\eta^{j_1}, \dots , \eta^{j_k}) \, \eta_p^{2r+1}b^*_{-p}   - \Pi^{(2)}_{\sharp',\flat'} (\eta^{j_1}, \dots , \eta^{j_k+1}) \, \eta_p^{2r}  b_p \end{equation}
In the second case, it is replaced by 
\begin{equation}\label{eq:restaa-Pi1ll} - \Pi^{(2)}_{\sharp',\flat'} (\eta^{j_1}, \dots , \eta^{j_k+1}) \, \eta_p^{2r+1}b^*_{-p}  - \Pi^{(2)}_{\sharp, \flat'} (\eta^{j_1}, \dots , \eta^{j_k}) \, \eta_p^{2r+2}  b_p \end{equation}
In (\ref{eq:restaa-Pi1l}), (\ref{eq:restaa-Pi1ll}), we used the notation $\flat' = (\flat_0, \dots , \flat_{k-1})$, $\sharp' = (\sharp_1, \dots , \bar{\sharp}_k)$. {F}rom the expression (\ref{eq:restaa-Pi1l}), (\ref{eq:restaa-Pi1ll}), we infer that also here (i), (ii), (iii), (v), (vi) are satisfied.

As for (iv), from the induction assumption there is exactly one term, in the expansion for $\text{ad}_{B(\eta)}^{(n)}( b_p)$, given by (\ref{eq:iv1}) if $n$ is even and by (\ref{eq:iv2}) if $n$ is odd. As an example, let us consider (\ref{eq:iv1}). If we commute the zero-order $\Pi^{(1)}$-operator $\eta_p^{n} b_p$ in (\ref{eq:iv1}) with $B(\eta)$, we obtain exactly  the term in (\ref{eq:iv2}), with $n$ replaced by $(n+1)$ (together with a second term, containing a $\Pi^{(1)}$-operator of order one). Similarly, if we take (\ref{eq:iv2}) and commute the $\Pi^{(1)}$-operator $\eta_p^{n} b^*_{-p}$ with $B(\eta)$, we get (\ref{eq:iv1}), with $n$ replaced by $(n+1)$. Considering the terms above, it is clear that there can be only exactly one term with this form. This shows that also in the expansion for $\text{ad}_{B(\eta)}^{(n+1)} (b_p)$, there is precisely one term of the form given in (iv). 

We conclude the proof by counting the number of terms in the expansion for the nested commutator $\text{ad}^{(n+1)}_{B(\eta)} (b_p)$. By the inductive assumption, $\operatorname{ad}_{B(\eta)}^{(n)}(b_p)$ can be expanded in a sum of exactly $2^n n!$ terms. $(ii)$ implies that each of these terms is a product of exactly $(n+1)$ operators, each of them being either $(N-\cN_+)$, $(N-(\cN_+-1))$, a field operator $b^{\sharp}_q$ or a quadratic factor $a_u^\sharp a_q^\flat$ commuting with the number of particles operator. By \eqref{2.2.Betacommutators}, the commutator of $B(\eta)$ with each such factor gives a sum of two terms. Therefore, by the product rule, $\operatorname{ad}_{B(\eta)}^{(n+1)}(b_p)$ contains $ 2^n(n!)\times 2(n+1) = 2^{(n+1)}((n+1)!)$ summands. 
\end{proof}

Using Lemma \ref{lm:indu} the remainder terms in the expansion (\ref{eq:BCH}) can be estimated in the same way as in Lemma \cite[Lemma 3.3]{BS}. The outcome is stated in the next lemma, whose proof is a translation into momentum space of the proof of \cite[Lemma 3.3]{BS}.
\begin{lemma}\label{lm:conv-series}
Let $\eta \in \ell^2 (\Lambda^*_+)$ be symmetric, with $\| \eta \|$ sufficiently small. Then we have 
\begin{equation}\label{eq:conv-serie}
\begin{split} e^{-B(\eta)} b_p e^{B (\eta)} &= \sum_{n=0}^\infty \frac{(-1)^n}{n!} \text{ad}_{B(\eta)}^{(n)} (b_p) \\
e^{-B(\eta)} b^*_p e^{B (\eta)} &= \sum_{n=0}^\infty \frac{(-1)^n}{n!} \text{ad}_{B(\eta)}^{(n)} (b^*_p) \end{split} \end{equation}
where the series on the r.h.s. are absolutely convergent. 
\end{lemma}

\section{The excitation Hamiltonian} 
\label{sec:ex}

We define the unitary operator $U_N: L^2_s (\Lambda^N) \to \cF_+^{\leq N}$ as in (\ref{eq:Uph-def}). In terms of creation and annihilation operators, the map $U_N$ is given by 
\[ U_N \psi_N = \bigoplus_{n=0}^N  (1-|\ph_0 \rangle \langle \ph_0|)^{\otimes n} \frac{a_0^{N-n}}{\sqrt{(N-n)!}} \psi_N \]
for all $\psi_N \in L^2_s (\Lambda^N)$ (here we identify $\psi_N \in L^2_s (\Lambda^N)$ with the vector $\{ \dots , 0, \psi_N, 0, \dots \} \in \cF$). The map $U_N^* : \cF_{+}^{\leq N} \to L^2_s (\Lambda^N)$ is given by 
\[ U_N^* \{ \psi^{(0)}, \dots , \psi^{(N)} \} = \sum_{n=0}^N \frac{a^* (\ph_0)^{N-n}}{\sqrt{(N-n)!}} \psi^{(n)} \]
It is useful to compute the action of $U_N$ on the product of a creation and an annihilation operators. We find (see \cite{LNSS}):
\begin{equation}\label{eq:U-rules}
\begin{split} 
U_N a^*_0 a_0 U_N^* &= N- \cN_+  \\
U_N a^*_p a_0 U_N^* &= a^*_p \sqrt{N-\cN_+ } \\
U_N a^*_0 a_p U_N^* &= \sqrt{N-\cN_+ } a_p \\
U_N a^*_p a_q U_N^* &= a^*_p a_q 
\end{split} \end{equation}
for all $p,q \in \Lambda^*_+ = \Lambda^* \backslash \{ 0 \}$. Writing the Hamiltonian (\ref{eq:ham0}) in momentum space, we find
\[ H_N = \sum_{p \in \Lambda^*} p^2 a_p^* a_p + \frac{\kappa}{2N} \sum_{p,q,r \in \Lambda^*} \widehat{V} (r/N) a_p^* a_q^* a_{q-r} a_{p+r} \]
With (\ref{eq:U-rules}), we can conjugate $H_N$ with the map $U_N$, defining $\cL_N = U_N H_N U_N^* : \cF_+^{\leq N} \to \cF_+^{\leq N}$. We find
\begin{equation}\label{eq:cLN} \cL_N =  \cL^{(0)}_N + \cL^{(2)}_N + \cL^{(3)}_N + \cL^{(4)}_N \end{equation}
with
\begin{equation}\label{eq:cLNj} \begin{split} 
\cL_N^{(0)} =\;& \frac{(N-1)}{2N} \kappa \widehat{V} (0) (N-\cN_+ ) + \frac{\kappa \widehat{V} (0)}{2N} \cN_+  (N-\cN_+ ) \\
\cL^{(2)}_N =\; &\sum_{p \in \Lambda^*_+} p^2 a_p^* a_p + \sum_{p \in \Lambda_+^*} \kappa  \widehat{V} (p/N) \left[ b_p^* b_p - \frac{1}{N} a_p^* a_p \right] + \frac{\kappa}{2} \sum_{p \in \Lambda^*_+} \widehat{V} (p/N) \left[ b_p^* b_{-p}^* + b_p b_{-p} \right] \\
\cL^{(3)}_N =\; &\frac{\kappa}{\sqrt{N}} \sum_{p,q \in \Lambda_+^* : p+q \not = 0} \widehat{V} (p/N) \left[ b^*_{p+q} a^*_{-p} a_q  + a_q^* a_{-p} b_{p+q} \right] \\
\cL^{(4)}_N =\; & \frac{\kappa}{2N} \sum_{p,q \in \Lambda_+^*, r \in \Lambda^*: r \not = -p,-q} \widehat{V} (r/N) a^*_{p+r} a^*_q a_p a_{q+r} 
\end{split} \end{equation}
The superscript $j=0,2,3,4$ indicates the number of creation and annihilation operators appearing in $\cL_N^{(j)}$. As explained in the introduction, in the mean-field regime the term $\cL_N^{(0)}$ is the ground state energy of the Bose gas and the sum of the quadratic, cubic and quartic contributions can be bounded below by $\cN_+$, up to errors of order one (at least for positive definite interaction). This is not the case in the Gross-Pitaevskii regime we are considering here. To extract the important contributions to the energy that are still hidden in $\cL_N^{(2)}, \cL_N^{(3)}, \cL_N^{(4)}$, we need to conjugate $\cL_N$ with a generalized Bogoliubov transformation, as defined in (\ref{eq:eBeta}). 

To choose the function $\eta \in \ell^2 (\Lambda^*_+)$ entering (\ref{eq:defB}) and (\ref{eq:eBeta}), we consider the solution of the Neumann problem 
\begin{equation}\label{eq:scatl} \left(-\Delta + \frac{\kappa}{2} V \right) f_{\ell} = \lambda_{\ell} f_\ell \end{equation}
on the  ball $|x| \leq N\ell$ (we omit the $N$-dependence in the notation for $f_\ell$ and for $\lambda_\ell$; notice that $\lambda_\ell$ scales as $N^{-3}$), with the normalization $f_\ell (x) = 1$ if $|x| = N \ell$. It is also useful to define $w_\ell = 1-f_\ell$ (so that $w_\ell (x) = 0$ if $|x| > N \ell$). By scaling, we observe that $f_\ell (N.)$ satisfies the equation 
\[ \left( -\Delta + \frac{\kappa N^2}{2} V (Nx) \right) f_\ell (Nx) = N^2 \lambda_\ell f_\ell (Nx) \]
on the ball $|x| \leq \ell$. We choose $0 < \ell < 1/2$, so that the ball of radius $\ell$ is contained in the box $\Lambda$. We extend then $f_\ell (N.)$ to $\Lambda$, by choosing $f_\ell (Nx) = 1$ for all $|x| > \ell$. Then  
\begin{equation}\label{eq:scatlN}
 \left( -\Delta + \frac{\kappa N^2}{2} V (Nx) \right) f_\ell (Nx) = N^2 \lambda_\ell f_\ell (Nx) \chi_\ell (x) 
\end{equation}
where $\chi_\ell$ is the characteristic function of the ball of radius $\ell$. In particular, $x \to w_\ell (Nx)$ is compactly supported and it can be extended to a periodic function on the torus $\Lambda$. The Fourier coefficients of the function $x \to w_\ell (Nx)$ are given by 
\[ \frac{1}{(2\pi)^3} \int_{\Lambda} w_\ell (Nx) e^{-i p \cdot x} dx = \frac{1}{N^3} \widehat{w}_\ell (p/N) \]
where \[ \widehat{w}_\ell (p) = \frac{1}{(2\pi)^3} \int_{\bR^3} w_\ell (x) e^{-ip \cdot x} dx \] is the Fourier transform of the function $w_\ell$. {F}rom (\ref{eq:scatlN}), we find the following relation for the Fourier coefficients of $w_\ell (Nx)$: 
\begin{equation}\label{eq:wellp}
\begin{split}  - p^2 \widehat{w}_\ell (p/N) +  \frac{\kappa N^2}{2} \widehat{V} (p/N) - \frac{\kappa}{2N} &\sum_{q \in \Lambda^*} \widehat{V} ((p-q)/N) \widehat{w}_\ell (q/N) \\ &= N^5 \lambda_\ell \widehat{\chi}_\ell (p) - N^2 \lambda_\ell \sum_{q \in \Lambda^*} \widehat{\chi}_\ell (p-q) \widehat{w}_\ell (q/N) \end{split} \end{equation}

In the next lemma we collect some important properties of $w_\ell, f_\ell$; its proof can be found in \cite[Lemma A.1]{ESY0} and in \cite[Lemma 4.1]{BS} (exchanging $V$ with $\kappa V$ and following the $\kappa$-dependence of the bounds). Notice that this lemma is the reason why we require that $V \in L^3 (\bR^3)$; for the rest of the analysis $V \in L^2 (\bR^3)$ would be enough. 
\begin{lemma} \label{3.0.sceqlemma}
Let $V \in L^3 (\bR^3)$ be non-negative, compactly supported and spherically symmetric. Fix $\ell > 0$ and let $f_\ell$ denote the solution of \eqref{eq:scatl}. 
\begin{enumerate}
\item [i)] We have 
\[ \lambda_\ell = \frac{3a_0}{N^3 \ell^3} \left(1 + \mathcal{O} (a_0 / N\ell) \right) \]
\item[ii)] We have $0\leq f_\ell, w_\ell \leq 1$ and \begin{equation}\label{eq:Vfa0} \left| \kappa \int  V(x) f_\ell (x) dx - 8\pi a_0 \right| \leq \frac{C\kappa}{N} \, . 
\end{equation}    
\item[iii)] There exists a constant $C>0 $ such that 
	\begin{equation}\label{3.0.scbounds1} 
	w_\ell(x)\leq \frac{C\kappa}{|x|+1} \quad\text{ and }\quad |\nabla w_\ell(x)|\leq \frac{C \kappa}{x^2+1}. 
	\end{equation}
for all $|x| \leq N \ell$.
\item[iv)] There exists a constant $C > 0$ such that 
\[ |\widehat{w}_\ell (p)| \leq \frac{C \kappa}{p^2} \]
for all $p \in \Lambda^*_+$.    
\end{enumerate}        
\end{lemma}
Using the solution $f_\ell$ of (\ref{eq:scatl}) and recalling that $w_\ell = 1 - f_\ell$, we define $\eta : \Lambda^* \to \bR$ through  
\begin{equation}\label{eq:ktdef} \eta_p = - \frac{1}{N^2}  \widehat{w}_\ell (p/N) 
\end{equation}
{F}rom Lemma \ref{3.0.sceqlemma}, it follows that 
\begin{equation}\label{eq:etap} |\eta_p | \leq \frac{C \kappa}{p^2} \end{equation}
and also that
\begin{equation}\label{eq:wteta0}  |\eta_0| \leq N^{-2} \int_{\bR^3} w_\ell (x) dx \leq C \kappa \end{equation}
Hence $\eta \in \ell^2 (\Lambda^*_+)$, uniformly in $N$. Another useful bound which can be proven with Lemma \ref{3.0.sceqlemma} (part iii)) is given by 
\begin{equation}\label{eq:etapN} \sum_{p \in \Lambda^*} p^2 |\eta_p|^2 = \| \nabla \check{\eta} \|_2^2 \leq C N \kappa^2 
\end{equation}
{F}rom (\ref{eq:wellp}), we obtain  
\begin{equation}\label{eq:eta-scat}
\begin{split} 
p^2 \eta_p + \frac{\kappa}{2} \widehat{V} (p/N) & + \frac{\kappa}{2N} \sum_{q \in \Lambda^*} \widehat{V} ((p-q)/N) \eta_q \\ &= N^3 \lambda_\ell \widehat{\chi}_\ell (p) + N^2 \lambda_\ell \sum_{q \in \Lambda^*} \widehat{\chi}_\ell (p-q) \eta_q
\end{split} \end{equation}

Using the coefficients $\eta_p$, for $p \not = 0$, we construct the generalized Bogoliubov transformation $e^{B(\eta)} : \cF_+^{\leq N} \to \cF^{\leq N}_+$ as in (\ref{eq:eBeta}). With it, we define the excitation Hamiltonian $\cG_N : \cF^{\leq N}_+ \to \cF^{\leq N}_+$ by setting (recall the definition (\ref{eq:cLN}) of the operator $\cL_N$)  
\begin{equation}\label{eq:GN} \cG_N = e^{-B(\eta)} \cL_N e^{B(\eta)} = e^{-B(\eta)} U_N H_N U_N^* e^{B(\eta)} 
\end{equation}
In the next proposition, we collect important properties of the self-adjoint operator $\cG_N$. 
\begin{prop}\label{prop:gene}
Let $V \in L^3 (\bR^3)$ be non-negative, compactly supported and  spherically symmetric and assume that the coupling constant $\kappa \geq 0$ is small enough. Then there exists a constant $C > 0$ such that, on $\cF_+^{\leq N}$,   
\begin{equation}\label{eq:bds-prop} 2\pi^2 \cN_+  - C \leq \frac{1}{2} (\cK + \cV_N) - C  \leq \cG_N - 4\pi a_0 N  \leq C (\cK + \cV_N + 1)  \end{equation}
where we used the notation 
\[ \cK = \sum_{p \in \Lambda_+^*} p^2 a_p^* a_p \qquad \text{and } \quad \cV_N = \frac{\kappa}{2N} \sum_{\substack{p,q \in \Lambda^*_+, r \in \Lambda^* \\ r \not = -p,-q}} \widehat{V} (r/N) \, a_{p+r}^* a_q^* a_p a_{q+r} \, . \]   
\end{prop}

The proof of Proposition \ref{prop:gene} is, from the technical point of view, the main part of our paper. It is deferred to Section \ref{sec:prop} below. Using Proposition \ref{prop:gene} we can now complete the proof of Theorem \ref{thm:main}.

\begin{proof}[Proof of Theorem \ref{thm:main}]
{F}rom the upper bound in (\ref{eq:bds-prop}), taking the expectation in the vacuum $\Omega = \{ 1, 0, \dots , 0 \} \in \cF_+^{\leq N}$, we find
\[ \langle U_N^* e^{B(\eta)} \Omega, H_N  U_N^* e^{B(\eta)} \Omega \rangle = \langle \Omega, \cG_N \Omega \rangle \leq 4 \pi a_0 N + C \]
In particular, this implies that the ground state energy $E_N$ of $H_N$ is such that 
\begin{equation}\label{eq:EN-up} E_N \leq 4 \pi a_0 N + C \, .  \end{equation}

{F}rom the lower bound 
\[ 2 \pi^2 \cN_+ - C \leq \cG_N - 4 \pi a_0 N \]
in (\ref{eq:bds-prop}), conjugating with $e^{B(\eta)}$ and then with $U_N^*$ we find, using Lemma \ref{lm:Ngrow}, the inequality 
\begin{equation}\label{eq:HN-ref} H_N \geq 4 \pi a_0 N + c \, U_N^* \cN_+ U_N - C = 4 \pi a_0 N + c \sum_{j=1}^N (1- |\ph_0 \rangle \langle \ph_0|)_j - C \end{equation}
between operators on $L^2_s (\Lambda^N)$. Here $(1- |\ph_0 \rangle \langle \ph_0|)_j$ denotes the orthogonal projection $1- |\ph_0 \rangle \langle \ph_0|$ acting on the $j$-th particle. On the one hand, (\ref{eq:HN-ref}) implies that $H_N \geq 4 \pi a_0 N - C$ and therefore that 
\[ E_N \geq 4 \pi a_0 N - C \, . \]
Combined with (\ref{eq:EN-up}), this bound implies (\ref{eq:gs-bd}). On the other hand, (\ref{eq:HN-ref}) implies that for a normalized $\psi_N \in L^2_s (\Lambda^N)$ with
\[ \langle \psi_N, H_N \psi_N \rangle \leq 4 \pi a_0 N + K \]
and with one-particle reduced density $\gamma^{(1)}_N$ we must have
\[ K + C \geq c \sum_{j=1}^N \langle \psi_N, (1-|\ph_0 \rangle \langle \ph_0| )_j \psi_N \rangle = c \, N \left[ 1 - \langle \ph_0, \gamma^{(1)}_N \ph_0 \rangle \right] \]
which implies that 
\[ 1 -  \langle \ph_0, \gamma^{(1)}_N \ph_0 \rangle \leq \frac{C(K+1)}{N} \]
for an appropriate $C >0$. This shows (\ref{eq:conv-thm}) and concludes the proof of Theorem~\ref{thm:main}.
\end{proof}

\section{Analysis of the excitation Hamiltonian $\cG_N$}
\label{sec:prop}

In this section, we prove Proposition \ref{prop:gene}. To this end, 
we use (\ref{eq:cLN}) to decompose the excitation Hamiltonian (\ref{eq:GN}) as 
\begin{equation}\label{eq:cGN-deco} \cG_N = \cG_N^{(0)} + \cG_N^{(2)} + \cG^{(3)}_N + \cG^{(4)}_N \end{equation}
with 
\[ \cG_N^{(j)} = e^{-B(\eta)} \cL^{(j)}_N e^{B(\eta)} \]
and with $\cL_N^{(j)}$ as defined in (\ref{eq:cLNj}), for $j=0,2,3,4$. 

\subsection{Preliminary results}
\label{sub:prel}

Before analyzing the operators on the r.h.s. of (\ref{eq:cGN-deco}), we collect in the following Lemma some preliminary bounds that will be used frequently in the next subsections.
\begin{lemma}\label{lm:prel}
Let $\xi \in \cF^{\leq N}_+$, $p,q \in \Lambda^*_+$, $i_1, i_2, k_1, k_2, \ell_1, \ell_2 \in \bN$, $j_1, \dots , j_{k_1}$, $m_1, \dots , m_{k_2} \in \bN \backslash \{0 \}$ and $\alpha_i = (-1)^{\ell_i}$ for $i=1,2$. For $s \in \{1, \dots ,i_1 \}, s' \in \{ 1, \dots , i_2 \}$, let $\Lambda_s, \Lambda'_{s'}$ be either a factor $(N-\cN_+ )/N$, a factor $(N+1-\cN_+ )/N$ or a $\Pi^{(2)}$-operator of the form 
\begin{equation}\label{eq:Pi2-prel}
N^{-h} \Pi^{(2)}_{\underline{\sharp}, \underline{\flat}} (\eta^{z_1}, \dots , \eta^{z_h}) 
\end{equation}
for some $h \in \bN \backslash \{ 0 \}$, $z_1, \dots, z_h \in \bN \backslash \{ 0 \}$ and $\underline{\sharp}, \underline{\flat} \in \{ \cdot, * \}^h$. Suppose that the operators 
\begin{equation}\label{eq:Pi1-prel} \begin{split} &\Lambda_1 \dots \Lambda_{i_1} N^{-k_1} \Pi^{(1)}_{\sharp,\flat} (\eta^{j_1}, \dots , \eta^{j_{k_1}} ; \eta^{\ell_1}_p \ph_{\alpha_1 p} ) \\ 
&\Lambda'_1 \dots \Lambda'_{i_2} N^{-k_2} \Pi^{(1)}_{\sharp',\flat'} (\eta^{m_1}, \dots , \eta^{m_{k_2}} ; \eta^{\ell_2}_q \ph_{\alpha_2 q}) 
\end{split} \end{equation}
with some $\sharp \in \{ \cdot , * \}^{k_1}, \flat \in \{ \cdot, * \}^{k_1 + 1}, \sharp' \in \{ \cdot, * \}^{k_2}, \flat' \in \{ \cdot, * \}^{k_2+1}$, appear in the expansion of $\text{ad}^{(n)}_{B(\eta)} (b_p)$ and of  $\text{ad}^{(k)}_{B(\eta)} (b_q)$ for some $n,k \in \bN$, as described in Lemma \ref{lm:indu}. 
\begin{itemize}
\item[i)] For any $\beta \in \bZ$, let 
\[ \text{D} = (\cN_+ +1)^{(\beta-1)/2} \Lambda_1 \dots \Lambda_{i_1} N^{-k_1} \Pi^{(1)}_{\sharp , \flat} (\eta^{j_1} , \dots , \eta^{j_{k_1}} ; \eta^{\ell_1}_p \ph_{\alpha_1 p} ) \xi \]
and
\[\wt{D} = (\cN_+ +1)^{(\beta-1)/2} \Pi^{(1)}_{\sharp , \flat} (\eta^{j_1} , \dots , \eta^{j_{k_1}} ; \eta^{\ell_1}_p \ph_{\alpha_1 p} )^* \Lambda_{i_1}^* \dots \Lambda^*_1 \xi \] 
Then, we have
\begin{equation}\label{eq:D1} \| \text{D} \| , \| \wt{D}  \|  \leq C^{n} \kappa^{n} p^{-2\ell_1} \| (\cN_+ +1)^{\beta/2} \xi \| \end{equation}
If $\ell_1$ is even, we also find  
\begin{equation}\label{eq:D0} \| \text{D} \| \leq C^{n} \kappa^{n} p^{-2\ell_1} \| a_p (\cN_+ +1)^{(\beta-1)/2} \xi \| \end{equation}
\item[ii)] For $\beta \in \bZ$, let
\[ \begin{split} \text{E} = \; &(\cN_+ +1)^{(\beta-1)/2} \Lambda_1 \dots \Lambda_{i_1} N^{-k_1} \Pi^{(1)}_{\sharp,\flat} (\eta^{j_1}, \dots , \eta^{j_{k_1}} ; \eta^{\ell_1}_p \ph_{\alpha_1 p}) \\ &\hspace{3cm} \times \Lambda'_1 \dots \Lambda'_{i_2} N^{-k_2} \Pi^{(1)}_{\sharp',\flat'} (\eta^{m_1}, \dots , \eta^{m_{k_2}} ; \eta^{\ell_2}_q \ph_{\alpha_2 q} ) \xi \end{split} \]
Then, we have 
\begin{equation}\label{eq:E11} \| \text{E} \| \leq C^{n+k} \kappa^{n+k} p^{-2\ell_1} q^{-2\ell_2} \| (\cN_+ +1)^{(\beta + 1)/2} \xi \| \end{equation}
If $\ell_2$ is even, we find 
\begin{equation}\label{eq:E10} \| \text{E} \| \leq C^{n+k} \kappa^{n+k} p^{-2\ell_1} q^{-2\ell_2} \| a_q (\cN_+ +1)^{\beta/2} \xi \| \end{equation}
If $\ell_1$ is even, we have
\begin{equation}\label{eq:E01-b} \begin{split} \| \text{E} \| \leq \; &C^{n+k} k N^{-1} \kappa^{n+k} p^{-2(\ell_1+1)} q^{-2\ell_2} \| (\cN_+ +1)^{(\beta+1)/2} \xi \| \\ &+ C^{n+k} \kappa^{n+k}  p^{-2(\ell_1 + \ell_2)} \mu_{\ell_2} \delta_{p,-q}  \| (\cN_+ +1)^{(\beta-1)/2}\xi \| \\ &+ C^{n+k} \kappa^{n+k} p^{-2\ell_1} q^{-2\ell_2} \| a_p (\cN_+ +1)^{\beta/2} \xi \| \end{split} 
\end{equation}
where $\mu_{\ell_2} = 1$ if $\ell_2$ is odd and $\mu_{\ell_2}= 0$ if $\ell_2$ is even. If $\ell_1$ is even and either $k_1 > 0$ or $k_2 >0$ or there is at least one $\Lambda$- or $\Lambda'$-operator having the form (\ref{eq:Pi2-prel}), we obtain the improved bound 
\begin{equation}\label{eq:E01} \begin{split} \| \text{E} \| \leq \; &C^{n+k} k N^{-1} \kappa^{n+k} p^{-2(\ell_1+1)} q^{-2\ell_2} \| (\cN_+ +1)^{(\beta+1)/2} \xi \| \\ &+ C^{n+k} N^{-1} \kappa^{n+k}  p^{-2(\ell_1 + \ell_2)} \mu_{\ell_2} \delta_{p,-q} \| (\cN_+ +1)^{(\beta+1)/2}\xi \| \\ &+ C^{n+k} \kappa^{n+k} p^{-2\ell_1} q^{-2\ell_2} \| a_p (\cN_+ +1)^{\beta/2} \xi \| \end{split} \end{equation}
Finally, if $\ell_1 = \ell_2 = 0$, we can write
\begin{equation}\label{eq:Edeco}  \text{E} = \text{E}_1 (p,q) + E_2 \, a_p a_q \xi 
\end{equation}
where 
\[ \| \text{E}_1 (p,q) \| \leq C^{n+k} k N^{-1} \kappa^{n+k}  p^{-2} \| a_q (\cN_+ +1)^{\beta/2} \xi \| \]
and $E_2$ is a bounded operator on $\cF^{\leq N}_+$ with \begin{equation}\label{eq:E2-est} \| E_2^\natural \zeta \| \leq C^{n+k} \kappa^{n+k} \| (\cN_+ +1)^{(\beta-1)/2} \zeta \| \end{equation} 
for $\natural \in \{ \cdot , * \}$ and for all $\zeta \in \cF^{\leq N}_+$. If $k_1 > 0$ or $k_2 > 0$ or at least one of the $\Lambda$- or $\Lambda$'-operators has the form (\ref{eq:Pi2-prel}), we also have the improved bound \begin{equation}\label{eq:E2impr} \| E_2^{\natural} \zeta \| \leq C^{n+k} N^{-1} \kappa^{n+k}   \| (\cN_+ +1)^{(\beta+1)/2} \zeta \| \end{equation}
for $\natural \in \{ \cdot ,* \}$ and all $\zeta \in \cF^{\leq N}_+$.
\end{itemize}
\end{lemma}

\begin{proof}
Let us start with part i). If $\Lambda_1$ is either the operator $(N-\cN_+ )/N$ or $(N-\cN_+ +1)/N$, then, on $\cF^{\leq N}_+$, 
\begin{equation}\label{eq:Lambd1} 
\begin{split} 
&\left\| (\cN_+ +1)^{(\beta-1)/2} \Lambda_1 \dots  \Lambda_{i_1} N^{-k_1} \Pi^{(1)}_{\sharp,\flat} (\eta^{j_1}, \dots , \eta^{j_{k_1}} ; \eta^{\ell_1}_{p} \ph_{\alpha_1 p}) \xi \right\| \\ &\hspace{2cm} \leq 2 \left\| (\cN_+ +1)^{(\beta-1)/2} \Lambda_2 \dots  \Lambda_{i_1} N^{-k_1} \Pi^{(1)}_{\sharp,\flat} (\eta^{j_1}, \dots , \eta^{j_k} ; \eta^{\ell_1}_p \ph_{\alpha_1 p}) \xi \right\| \end{split} 
\end{equation}
If instead $\Lambda_1$ has the form (\ref{eq:Pi2-prel}) for a $h \geq 1$, we apply Lemma \ref{lm:Pi-bds} and we find (using part vi) in Lemma \ref{lm:indu})
\begin{equation}\label{eq:Lambd2} \begin{split} &\left\| (\cN_+ +1)^{(\beta-1)/2}  \Lambda_1 \dots \Lambda_{i_1} N^{-k_1} \Pi^{(1)}_{\sharp,\flat} (\eta^{j_1}, \dots , \eta^{j_{k_1}} ; \eta^{\ell_1}_{p} \ph_{\alpha_1 p}) \xi \right\|
\\ &\hspace{1cm} \leq  C^h  \kappa^{\bar{h}} \| (\cN_+ +1)^{(\beta-1)/2}  \Lambda_2 \dots  \Lambda_{i_1} N^{-k_1} \Pi^{(1)}_{\sharp,\flat} (\eta^{j_1}, \dots , \eta^{j_{k_1}} ; \eta^{\ell_1}_p \ph_{\alpha_1 p}) \xi \| 
\end{split} \end{equation}
where we used the notation $\bar{h} = z_1 + \dots + z_h$ for the total number of factors $\eta$'s appearing in (\ref{eq:Pi2-prel}). Iterating the bounds (\ref{eq:Lambd1}) and (\ref{eq:Lambd2}), we find 
\begin{equation}\label{eq:bd1-Lam} 
\begin{split} &\| (\cN_+ +1)^{(\beta-1)/2}  \Lambda_1
 \dots \Lambda_{i_1} N^{-k_1} \Pi^{(1)}_{\sharp,\flat} (\eta^{j_1}, \dots , \eta^{j_k} ; \eta^{\ell_1}_{p} \ph_{\alpha_1 p}) \xi \| \\ &\hspace{.2cm} \leq C^{r+h_1 + \dots + h_s} \kappa^{\bar{h}_1 + \dots + \bar{h}_s} \| (\cN_+ +1)^{(\beta-1)/2} N^{-k_1} \Pi^{(1)}_{\sharp,\flat} (\eta^{j_1}, \dots , \eta^{j_{k_1}} ; \eta^{\ell_1}_{p} \ph_{\alpha_1 p}) \xi \| \end{split}  \end{equation}
if $r$ of the operators $\Lambda_1, \dots , \Lambda_{i_1}$ have either the form $(N-\cN_+ )/N$ or the form $(N-\cN_+ +1)/N$, and the other $s=i_1-r$ are $\Pi^{(2)}$-operators of the form (\ref{eq:Pi2-prel}) of order $h_1, \dots , h_s$ , containing $\bar{h}_1, \dots \bar{h}_s$ factors $\eta$. Again with Lemma \ref{lm:Pi-bds} and with (\ref{eq:etap}), we obtain (using also Lemma \ref{lm:indu}, part iii), and the fact that $(\cN_+ +1)^{(\beta-1)/2} \Pi^{(1)}_{\sharp,\flat} (\dots) = \Pi^{(1)}_{\sharp,\flat} (\dots) (\cN_+ +1 \pm 1)^{(\beta-1)/2}$)
\begin{equation}\label{eq:prel-i} \begin{split}  &\| (\cN_+ +1)^{(\beta-1)/2}  \Lambda_1
 \dots \Lambda_{i_1} N^{-k_1} \Pi^{(1)}_{\sharp,\flat} (\eta^{j_1}, \dots , \eta^{j_{k_1}} ; \eta^{\ell_1}_{p} \ph_{\alpha_1 p}) \xi \| \\ &\hspace{1cm} \leq C^{r+h_1 +\dots + h_s+j_1 + \dots + j_{k_1}+\ell_1} \kappa^{\bar{h}_1 + \dots + \bar{h}_s+j_1 +\dots +j_{k_1}+\ell_1} p^{-2\ell_1} \| (\cN_+ +1)^{\beta/2} \xi \| \\ &\hspace{1cm} \leq C^{n} \kappa^n p^{-2\ell_1} \| (\cN_+ +1)^{\beta/2} \xi \| \, . \end{split} 
\end{equation}
This shows the bound (\ref{eq:D1}) for $\| D \|$. The bound (\ref{eq:D1}) for $\| \wt{D} \|$ can be proven similarly. If we now assume that $\ell_1$ is even, the last field on the right in the $\Pi^{(1)}$ operator in the term $\text{D}$ must be an annihilation operator $a_p$ (see Lemma \ref{lm:indu}, part v)). Proceeding as above, but  estimating \[ \begin{split} 
\| (\cN_+ +1)^{(\beta-1)}/2 N^{-k_1} &\Pi^{(1)}_{\sharp,\flat} (\eta^{j_1}, \dots , \eta^{j_{k_1}} ; \eta^{\ell_1}_p \ph_p) \xi \| \\ &\leq C^{j_1 + \dots + j_{k_1}+\ell_1} \kappa^{j_1 + \dots + j_{k_1}+\ell_1} p^{-2\ell_1} \| a_p (\cN_+ +1)^{(\beta-1)/2} \xi \| \end{split} \]
we also obtain (\ref{eq:D0}). 

Let us now consider part ii). The bounds (\ref{eq:E11}) and (\ref{eq:E10}) follow applying (\ref{eq:D1}) twice and, respectively, (\ref{eq:D1}) and then (\ref{eq:D0}). We focus therefore on (\ref{eq:E01-b}). Here, we assume that $\ell_1$ is even. This implies that the field operator on the right of the first $\Pi^{(1)}$-operator is an annihilation operator $a_p$. To bound $\| \text{E} \|$, we have to commute $a_p$ to the right, until it hits $\xi$. To commute $a_p$ through factors of $\cN_+ $, we use the pull-through formula $a_p \cN_+  = (\cN_+ +1) a_p$. On the other hand, when we commute $a_p$ through a pair of creation and/or annihilation operators associated with a function $\eta^{j}$ for  some $j \geq 1$ (like the pairs appearing in the $\Pi^{(2)}$-operators of the form (\ref{eq:Pi2-prel}) or in the $\Pi^{(1)}$-operators in (\ref{eq:Pi1-prel})), we generate a creation or an annihilation operator $a_p$ or $a_{-p}^*$ together with an additional factor $\eta_p^j$. Furthermore, since the commutator erases a creation and an annihilation operator, we can save a factor $N^{-1}$ (taken from the factor $N^{-h}$ in (\ref{eq:Pi2-prel}) or from the factor $N^{-k_2}$ in (\ref{eq:Pi1-prel})). For example,  
\[ \left[ a_p , \sum_{r \in \Lambda^*} \eta^{j}_r a^*_{r} a_{r}  \right] = \eta^{j}_p a_p \]
There are at most $k$ pairs of creation and/or annihilation operators through which $a_p$ needs to be commuted (because every such pair carries a factor $\eta^j$, and the total number of $\eta$ factors on the right of $a_p$ is $k$). At the end, we also have to pass $a_p$ through the field operator appearing on the right of the second $\Pi^{(1)}$-operator; this is  either the annihilation operator $a_q$ if $\ell_2$ is even, or the creation operator $a^*_{-q}$, if $\ell_2$  is odd. Hence, the commutator vanishes if $\ell_2$ is even, while it is given by 
\begin{equation}\label{eq:comm-ell2} [ a_p, a^*_{-q} ] = \delta_{p,-q} 
\end{equation} 
if $\ell_2$ is odd. This leads to the estimate (\ref{eq:E01-b}). If we additionally assume that either $k_1 > 0$ or $k_2 > 0$ or that there is at least one $\Lambda$- or $\Lambda'$-operator having the form (\ref{eq:Pi2-prel}), in the contribution arising from the commutator of $a_p$ and $a_{-q}^*$ (which is only present if $\ell_2$ is odd), we can extract an additional factor $(\cN_+ +1)/N$ (this additional factor can be used here and not elsewhere, because in this term, after commuting $a_p$ and $a_{-q}^*$, there is one less factor of $\cN_+$). This observation leads to (\ref{eq:E01}). Finally, let us consider $\ell_1 = \ell_2 =0$. In this case we proceed as before, commuting the annihilation operator $a_p$ to the right. The contribution of the commutators of $a_p$ with the pairs of creation and annihilation fields appearing in the $\Pi^{(1)}$-operator and possibly in the $\Pi^{(2)}$-operators 
lying on the right of $a_p$ is collected in the term $\text{E}_1$ (this term can be estimated as on the first line on the r.h.s. of (\ref{eq:E01-b}) or (\ref{eq:E01})). After commuting $a_p$ all the way to the right, we are left with the second term on the r.h.s. of (\ref{eq:Edeco}), with the operator $\text{E}_2$ containing all $\Lambda$- and $\Lambda'$-operators as well as all pairs of annihilation and/or creation operators appearing in the two $\Pi^{(1)}$-operator which can be estimated, following Lemma \ref{lm:Pi-bds} as in (\ref{eq:E2-est}) or (\ref{eq:E2impr}). 
\end{proof}

\subsection{Analysis of $\cG^{(0)}_N$}
\label{subsec:L0}

{F}rom (\ref{eq:cLNj}), we have
\[ \cG^{(0)}_{N} = e^{-B(\eta)} \cL^{(0)}_N e^{B(\eta)} = \frac{(N-1)}{2} \kappa \widehat{V} (0)+ \cE_N^{(0)} \]
with 
\[ \cE_N^{(0)} = \frac{\kappa \widehat{V} (0)}{2N} e^{-B(\eta)} \cN_+  e^{B(\eta)} - \frac{\kappa \widehat{V} (0)}{2N} e^{-B(\eta)} \cN_+ ^2 e^{B(\eta)} \]
In the next Proposition, we estimate the error term $\cE_N^{(0)}$. 
\begin{prop}\label{prop:E0}
Let the assumptions of Proposition \ref{prop:gene} be satisfied.
Then there exists a constant $C > 0$ such that 
\begin{equation}\label{eq:E0} \pm \cE_N^{(0)} \leq C \kappa (\cN_+ +1) \end{equation}
as operator inequality on $\cF^{\leq N}_+$.
\end{prop}
\begin{proof}
Eq. (\ref{eq:E0}) follows from Lemma \ref{lm:Ngrow} and the fact that, on $\cF^{\leq N}_+$, $\cN_+  \leq N$. 
\end{proof}

\subsection{Analysis of $\cG_N^{(2)}$} 
\label{subsec:L2}

{F}rom (\ref{eq:cLNj}), we recall that
\[ \cL^{(2)}_N = \cK + \wt{\cL}_N^{(2)} \]
where
\[ \cK = \sum_{p \in \Lambda^*_+} p^2 a_p^* a_p \]
is the kinetic energy operator and
\begin{equation}\label{eq:wtL2} \wt{\cL}_N^{(2)} = \sum_{p \in \Lambda_+^*} \kappa \widehat{V} (p/N) \left[ b_p^* b_p - \frac{1}{N} a_p^* a_p \right] + \frac{\kappa}{2} \sum_{p \in \Lambda^*_+} \widehat{V} (p/N) \left[ b_p^* b_{-p}^* + b_p b_{-p} \right] \end{equation}

\subsubsection{Analysis of $e^{-B(\eta)} \cK e^{B(\eta)}$}
\label{subsec:K}

We write
\begin{equation}\label{eq:def-EK} e^{-B(\eta)} \cK e^{B(\eta)} = \cK + \sum_{p \in \Lambda_+^*} p^2 \eta_p^2 + \sum_{p \in \Lambda_+^*} p^2 \eta_p \left[ b^*_p b^*_{-p} + b_p b_{-p} \right] + \cE^{(K)}_N \end{equation}
In the next proposition, we bound the error term $\cE^{(K)}_N$. 
\begin{prop}\label{prop:K}
Let the assumptions of Proposition \ref{prop:gene} be satisfied (in particular, suppose $\kappa \geq 0$ is small enough). Then, for every $\delta > 0$ there exists a constant $C > 0$ such that, on $\cF^{\leq N}_+$, 
\[ \pm \cE^{(K)}_N \leq \delta (\cK + \cV_N) + C \kappa (\cN_+ +1) \, . \]
\end{prop}

\begin{proof}
We write
\[ \begin{split} 
e^{-B(\eta)} \cK e^{B(\eta)} &= \cK + \int_0^1 e^{-s B(\eta)} [\cK , B(\eta)] e^{sB(\eta)} ds \\ &=\cK+ \int_0^1 \sum_{p \in \Lambda_+^*} p^2 \eta_p \left[ e^{-s B(\eta)} b_p b_{-p} e^{s B(\eta)} + e^{-s B(\eta)} b_p^* b_{-p}^* e^{sB (\eta)} \right] \, ds \end{split}  \]
Lemma \ref{lm:conv-series}, together with $\text{ad}^{(n)}_{sB(\eta)} (A) = s^n \text{ad}^{(n)}_{B(\eta)} (A)$, implies that 
\[ e^{-B(\eta)} \cK e^{B(\eta)} = \cK + \sum_{n,k \geq 0} \frac{(-1)^{n+k}}{n!k!(n+k+1)} \sum_{p \in \Lambda_+^*} p^2 \eta_p \left[ \text{ad}^{(n)}_{B(\eta)} (b_p) \text{ad}^{(k)}_{B(\eta)} (b_{-p}) + \text{h.c.} \right] \]
We separate the summands with $(n,k) = (0,0), (0,1)$; we find 
\[ \begin{split} e^{-B(\eta)} \cK e^{B(\eta)} = \; &\cK + \sum_{p \in \Lambda_+^*} p^2 \eta_p \left[ b_p b_{-p} + b^*_p b^*_{-p} \right] - \frac{1}{2} \sum_{p \in \Lambda_+^*} p^2 \eta_p \left( b_p [B(\eta), b_{-p}] + \text{h.c} \right) \\ &+\sum^*_{n,k} \frac{(-1)^{n+k}}{n!k!(n+k+1)} \sum_{p \in \Lambda_+^*} p^2 \eta_p \left[ \text{ad}^{(n)}_{B(\eta)} (b_p) \text{ad}^{(k)}_{B(\eta)} (b_{-p}) + \text{h.c.} \right] \end{split} \]
where $\sum^*_{n,k}$ indicates the sum over all pairs $(n,k) \not = (0,0), (0,1)$. With (\ref{eq:comm-bp}) and (\ref{eq:def-EK}) we obtain 
\begin{equation}\label{eq:Kterms} 
\begin{split} \cE^{(K)}_N   = \; & \sum_{p \in \Lambda_+^*} p^2 \eta^2_p \left[ b_p^* b_p - \frac{1}{N} a_p^* a_p \right] - \frac{\cN_+ }{N} \sum_{p \in \Lambda_+^*} p^2 \eta^2_p \\ &- \frac{1}{N} \sum_{p \in \Lambda^*_+} p^2 \eta_p^2 b_p \cN_+  b_p^* -\frac{1}{2N} \sum_{p,q \in \Lambda^*_+} p^2 \eta_p \eta_q \left( b_p b_q^* a^*_{-q} a_{-p} + \text{h.c.} \right) \\ &+\sum^*_{n,k} \frac{(-1)^{n+k}}{n!k!(n+k+1)} \sum_{p \in \Lambda_+^*} p^2 \eta_p \left[ \text{ad}^{(n)}_{B(\eta)} (b_p) \text{ad}^{(k)}_{B(\eta)} (b_{-p}) + \text{h.c.} \right] \\
=: \; &\text{G}_1 + \text{G}_2 + \text{G}_3 + \text{G}_4 \\ &+\sum^*_{n,k} \frac{(-1)^{n+k}}{n!k!(n+k+1)} \sum_{p \in \Lambda_+^*} p^2 \eta_p \left[ \text{ad}^{(n)}_{B(\eta)} (b_p) \text{ad}^{(k)}_{B(\eta)} (b_{-p}) + \text{h.c.} \right] 
\end{split} 
\end{equation}
The expectation of the first term on the r.h.s. of (\ref{eq:Kterms}) can be estimated by 
\begin{equation}\label{eq:G1fin} \begin{split} \left| \langle \xi , \text{G}_1 \xi \rangle \right| &\leq  \sum_{p \in \Lambda^*_+} p^2 \eta_p^2 \| b_p \xi \|^2 + \frac{1}{N}  \sum_{p \in \Lambda^*_+} p^2 \eta_p^2 \| a_p \xi \|^2 \\ &\leq \sup_{p \in \Lambda^*_+} \, ( p^2 \eta_p^2 ) \, \| \cN_+ ^{1/2} \xi \|^2 \leq C \kappa^2 \| \cN_+ ^{1/2} \xi \|^2 \end{split} \end{equation} 
with (\ref{eq:etap}). To bound the second term on the r.h.s. of (\ref{eq:Kterms}) we remark that, by (\ref{eq:etapN}),   
\begin{equation}\label{eq:etaH1} \sum_p p^2 \eta_p^2 = \| \nabla \check{\eta} \|^2 \leq C N \kappa^2 
\end{equation}
This implies that 
\begin{equation}\label{eq:G2fin} \left| \langle \xi, \text{G}_2 \xi \rangle \right| \leq C \kappa^2 \| \cN_+ ^{1/2} \xi \|^2 \end{equation}
To estimate the contribution of the third term on the r.h.s. of  (\ref{eq:Kterms}), we commute $b_p$ to the right of $b_p^*$. We find, using the fact that $\cN_+ \leq N$ on $\cF_+^{\leq N}$ and again (\ref{eq:etap}), that 
\begin{equation}\label{eq:G3fin}\begin{split} \left| \langle \xi , \text{G}_3 \xi \rangle \right| &\leq \frac{2}{N} \sum_{p \in \Lambda^*_+} p^2 \eta_p^2 \| (\cN_+ +1)^{1/2} \xi \|^2 + \frac{1}{N} \sum_{p \in \Lambda_+^*} p^2 \eta_p^2 \| a_p (\cN_+ +1)^{1/2} \xi \|^2 \\ &\leq C \kappa^2 \| (\cN_+ +1)^{1/2} \xi \|^2 \end{split} \end{equation}
As for the fourth term on the r.h.s. of (\ref{eq:Kterms}), we write it as 
\begin{equation}\label{eq:A4}  
\begin{split} \text{G}_4 = \; &-\frac{1}{2N} \sum_{p,q \in \Lambda^*_+} p^2 \eta_p \eta_q  \left[ b_q^* a_{-q}^* a_{-p} b_p + \text{h.c.} \right] + \frac{1}{2N^2} \sum_{p,q \in \Lambda^*_+} p^2 \eta_p \eta_q \, \left[ a_q^* a_p a^*_{-q} a_{-p} + \text{h.c.} \right]  \\ &- \frac{1}{N} \sum_{p \in \Lambda_+^*} p^2 \eta_p^2 \left[ b_p^* b_p + \frac{N-\cN_+ }{N} a_p^* a_p \right] \\ =: \; & \text{G}_{41} + \text{G}_{42} + \text{G}_{43} \end{split} 
\end{equation}
While it is easy to bound  
\begin{equation}\label{eq:E42} \begin{split}\left| \langle \xi , \text{G}_{42} \xi \rangle \right|  &\leq \frac{1}{2N^2} \sum_{p,q \in \Lambda^*_+} p^2 \eta_p \eta_q \| a_q (\cN_+ +1)^{1/2} \xi \| \| a_p (\cN_+ +1)^{1/2} \xi \|  \\
&\leq \frac{1}{2N^2} \left[ \sum_{p,q \in \Lambda^*_+} p^2 \eta_p^2 \| a_q (\cN_+ +1)^{1/2} \xi \|^2 \right]^{\frac{1}{2}} \left[ \sum_{p,q \in \Lambda^*_+} p^2 \eta_q^2 \| a_p (\cN_+ +1)^{1/2} \xi \|^2 \right]^{\frac{1}{2}} 
\\
&\leq C N^{-1/2} \kappa^2 \| (\cN_+ +1)^{1/2} \xi \| \| \cK^{1/2} \xi \| \end{split} \end{equation}
and 
\begin{equation}\label{eq:E43} \left| \langle \xi , \text{G}_{43} \xi \rangle \right| \leq C N^{-1} \kappa^2 \| (\cN_+ +1)^{1/2} \xi \|^2 ,\end{equation}
in order to control the term $\text{G}_{41}$ we need to use Eq. (\ref{eq:eta-scat}). We find
\begin{equation}\label{eq:E41} \begin{split} \text{G}_{41} = \; &\frac{\kappa}{4N} \sum_{p,q \in \Lambda^*_+} \widehat{V} (p/N) \eta_q \, \left[ b_q^* a_{-q}^* a_{-p} b_p + \text{h.c.} \right] \\ &- \frac{\kappa}{4N^2} \sum_{p,q \in \Lambda^*_+, r\in \Lambda^*} \widehat{V} ((p-r)/N) \eta_r \eta_q \left[ b_q^* a_{-q}^* a_{-p} b_p + \text{h.c.} \right] \\ &+ N^2 \lambda_\ell \sum_{p,q \in \Lambda_+^*} \widehat{\chi}_\ell (p) \eta_q  \left[ b_q^* a_{-q}^* a_{-p} b_p + \text{h.c.} \right] \\ &- N \lambda_\ell \sum_{p,q \in \Lambda^*_+, r\in \Lambda^*} \widehat{\chi}_\ell (p-r) \eta_r \eta_q \left[ b_q^* a_{-q}^* a_{-p} b_p + \text{h.c.} \right] \\
=: \; & \text{G}_{411}+ \text{G}_{412} + \text{G}_{413} + \text{G}_{414} \end{split} \end{equation}
We estimate
\[ \begin{split} |\langle \xi, \text{G}_{413} \xi \rangle | &\leq \frac{C\kappa}{N} \sum_{p,q \in \Lambda^*_+} |\widehat{\chi}_\ell (p)| |\eta_q| \| a_q (\cN_+ +1)^{1/2} \xi \| \| a_p (\cN_+ +1)^{1/2} \xi \| \\ &\leq \frac{C\kappa}{N} \| \widehat{\chi}_\ell \|_2 \| \eta \| \| (\cN_+ +1) \xi \|^2 \\ &\leq C \kappa^2 \| (\cN_+ +1)^{1/2} \xi \|^2 \end{split} \]
Furthermore
\[ \begin{split} |\langle \xi, \text{G}_{414} \xi \rangle | &\leq  \frac{C\kappa}{N^2} \sum_{p,q \in \Lambda^*_+} |g (p)| |\eta_q| \| a_q (\cN_+ +1)^{1/2} \xi \| \| a_p (\cN_+ +1)^{1/2} \xi \| \\ &\leq C \kappa \| \eta \| \| g \| \| (\cN_+ +1)^{1/2} \xi \|^2 \end{split} \]
where we defined $g (p) = \sum_{r \in \Lambda^*} \widehat{\chi}_\ell (p-r) \eta_r$. Since 
\[ \| g \| = \| \chi_\ell \check{\eta} \| \leq \| \check{\eta} \| = \| \eta \| \leq C \kappa \]
we conclude that
\[ |\langle \xi, \text{G}_{414} \xi \rangle | \leq C \kappa^2 \| (\cN_+ +1)^{1/2} \xi \|^2 \]
Let us now consider the first term on the r.h.s. of (\ref{eq:E41}). Switching to position space we find, on $\cF^{\leq N}_+$, 
\[ \begin{split} \text{G}_{411} &= \frac{\kappa}{4N} \int_{\Lambda^{\times 4}} dx dy dz dw \sum_{p,q \in \Lambda^*_+} \widehat{V} (p/N) \eta_q e^{iq (z-w)} e^{ip (x-y)} \check{b}_z^* \check{a}_w^* \check{a}_x \check{b}_y \\ &= \frac{\kappa}{4} \int_{\Lambda^{\times 4}} dx dy dz dw \, N^2 V (N (x-y)) \check{\eta} (z-w) \check{b}_z^* \check{a}_w^* \check{a}_x \check{b}_y \end{split} \]
Hence
\[ \begin{split} |\langle \xi , \text{G}_{411} \xi \rangle | &\leq C \kappa \int_{\Lambda^{\times 4}} dx dydz dw \, N^2 V(N(x-y)) |\check{\eta} (z-w)| \| \check{a}_x \check{a}_y \xi \| \| \check{a}_w \check{a}_z \xi \| \\ &\leq C \kappa \left[ \int_{\Lambda^{\times 4}} dx dydz dw \, N^2 V(N(x-y)) |\check{\eta} (z-w)|^2 \| \check{a}_x \check{a}_y \xi \|^2 \right]^{1/2} \\ &\hspace{4cm} \times \left[ 
 \int_{\Lambda^{\times 4}} dx dydz dw \, N^2 V(N(x-y)) \| \check{a}_z \check{a}_w \xi \|^2 \right]^{1/2} \\ &\leq C \kappa^{3/2} \| (\cN_+ +1)^{1/2} \xi \| \| \cV_N^{1/2} \xi \| \end{split} \] 
The term $\text{G}_{412}$ can also be estimated similarly. We conclude that
\[ |\langle \xi , \text{G}_{41} \xi \rangle | \leq C \kappa^2 \| (\cN_+ +1)^{1/2} \xi \|^2 + C \kappa^{3/2} \| (\cN_+ +1)^{1/2} \xi \| \| \cV_N^{1/2} \xi \| \]
and therefore, together with (\ref{eq:E42}), (\ref{eq:E43}), we find  
\begin{equation}\label{eq:G4fin} |\langle \xi , \text{G}_{4} \xi \rangle | \leq C  \kappa^2 \| (\cN_+ +1)^{1/2} \xi \| \| (\cK + \cN_+  +1)^{1/2} \xi \| + C \kappa^{3/2} \| (\cN_+ +1)^{1/2} \xi \| \| \cV_N^{1/2} \xi \| \end{equation}
We consider next the last term in (\ref{eq:Kterms}), namely the sum over all pairs $(n,k) \not = (0,0), (0,1)$. According to Lemma \ref{lm:conv-series}, the operator
\begin{equation}\label{eq:adadK} \sum_{p\in \Lambda^*_+} p^2 \eta_p \, \text{ad}^{(n)}_{B(\eta)} (b_p) \text{ad}^{(k)} (b_{-p})\end{equation}
can be written as the sum of $2^{n+k} n!k!$ terms having the form
\begin{equation}\label{eq:E} 
\begin{split} \text{G} = \; &\sum_{p \in \Lambda^*_+} p^2 \eta_p \, \Lambda_1 \dots \Lambda_{i_1} N^{-k_1} \Pi^{(1)}_{\sharp,\flat} (\eta^{j_1}, \dots , \eta^{j_{k_1}} ; \eta^{\ell_1}_p \ph_{\alpha_1 p}) \\ &\hspace{3cm} \times \Lambda'_1 \dots \Lambda'_{i_2} N^{-k_2} \Pi^{(1)}_{\sharp', \flat'} (\eta^{m_1}, \dots , \eta^{m_{k_2}} ; \eta^{\ell_2}_p \ph_{-\alpha_2 p}) 
\end{split} 
\end{equation}
with $i_1, i_2, k_1, k_2, \ell_1, \ell_2 \in \bN$, $j_1, \dots , j_{k_1}, m_1, \dots, m_{k_2} \in \bN \backslash \{ 0 \}$, $\alpha_i = (-1)^{\ell_i}$ for $i=1,2$, and where each $\Lambda_r$, $\Lambda'_r$ is either a factor $(N-\cN_+ )/N$, $(N+1-\cN_+ )/N$ or a $\Pi^{(2)}$-operator of the form
\begin{equation}\label{eq:Pi2-typ} N^{-h} \Pi^{(2)}_{\underline{\sharp}, \underline{\flat}} (\eta^{(z_1)}, \dots , \eta^{(z_h)}) \end{equation}
with $h, z_1, \dots , z_h \in \bN \backslash \{ 0 \}$. We estimate the expectation of operators of the form (\ref{eq:E}).

Let us first assume that $\ell_1 + \ell_2 \geq 1$. With Lemma \ref{lm:prel}, part ii), we find (using the bounds \eqref{eq:E11} if $\ell_1 + \ell_2 \geq 2$, \eqref{eq:E10} if $(\ell_1, \ell_2)=(1,0)$ and \eqref{eq:E01} if $(\ell_1, \ell_2)=(0,1)$) 
\begin{equation}\label{eq:l1l2} 
\begin{split}  |\langle \xi , \text{G} \xi \rangle | \leq \; &C^{n+k} \| (\cN_+ +1)^{1/2} \xi \| \\ &\times \sum_{p \in \Lambda^*_+} p^2 \eta_p  \Big\{ (1+k/N) \eta_p^2 \kappa^{n+k-2} \| (\cN_+ +1)^{1/2} \xi \| \\ &\hspace{2cm} + N^{-1} \eta_p \kappa^{n+k-1} \| (\cN_+ +1)^{1/2} \xi \| + \eta_p \kappa^{n+k-1} \| a_p \xi \|  \Big\}  \end{split} \end{equation}
To apply (\ref{eq:E01}) in the case $(\ell_1, \ell_2) = (0,1)$, we use here the fact that the pairs $(n,k) = (0,0), (0,1)$ are excluded. The choice $(n,k) = (1,0)$ is not compatible with $(\ell_1, \ell_2) = (0,1)$ (by Lemma \ref{lm:indu}, $\ell_1 \leq n$ and $\ell_2 \leq k$). Hence $n+k \geq 2$, while $\ell_1 + \ell_2 = 1$; this implies by Lemma \ref{lm:indu}, part iii), that either $k_1 > 0$ or $k_2 > 0$ or at least one of the $\Lambda$- or $\Lambda'$-operators is a $\Pi^{(2)}$-operator of the form (\ref{eq:Pi2-typ}). With (\ref{eq:etap}) and (\ref{eq:etapN}), we conclude from (\ref{eq:l1l2}) that
\begin{equation}\label{eq:Gfin} |\langle \xi , \text{G} \xi \rangle | \leq C^{k+n} (1+k/N) \kappa^{n+k+1} \| (\cN_+ +1)^{1/2} \xi \|^2 \end{equation}

Let us now consider the case $\ell_1 = \ell_2 = 0$. With (\ref{eq:Edeco}) in Lemma \ref{lm:prel}, we can write 
\begin{equation}\label{eq:Edeco-kin} \langle \xi , \text{G} \xi \rangle = \sum_{p \in \Lambda^*_+} p^2 \eta_p \langle (\cN_+ +1)^{1/2} \xi , \text{E}_1 (p,-p) \rangle  + \sum_{p\in \Lambda^*_+} p^2 \eta_p  \langle (\cN_+ +1)^{1/2} \xi , \text{E}_2 a_p a_{-p} \xi \rangle \end{equation}
where the first term can be bounded by 
\[ \begin{split} 
\left|\sum_{p \in \Lambda^*_+} p^2 \eta_p \langle (\cN_+ +1)^{1/2} \xi , \text{E}_1 (p,-p) \rangle \right| &\leq \sum_{p \in \Lambda^*_+} p^2 \eta_p \| (\cN_+ +1)^{1/2} \xi \| \| \text{E}_1 (p,-p) \| \\ &\leq C^{n+k} k N^{-1} \kappa^{n+k+1} \| (\cN_+ +1)^{1/2} \xi \| \sum_{p \in \Lambda^*_+} p^{-2} \| a_p \xi \|  \\ &\leq C^{n+k} k N^{-1} \kappa^{n+k+1} \| (\cN_+ +1)^{1/2} \xi \|^2 \end{split} \]
As for the second term on the r.h.s. of (\ref{eq:Edeco-kin}), we  use the relation (\ref{eq:eta-scat}) to replace
\begin{equation}\label{eq:eta-scat2} p^2 \eta_p = -\frac{\kappa}{2} \widehat{V} (p/N) - \frac{\kappa}{2N} \sum_{q \in \Lambda^*} \widehat{V} ((p-q)/N) \eta_q + N^3 \lambda_\ell \widehat{\chi}_\ell (p) + N^2 \lambda_\ell \sum_{q \in \Lambda^*} \widehat{\chi}_\ell (p-q) \eta_q \end{equation}
To bound the contribution proportional to $\kappa \widehat{V} (p/N)$, we switch to position space. We find , for $\xi \in \cF^{\leq N}_+$, 
\[ \begin{split} \Big| \kappa \sum_{p \in \Lambda^*_+} \widehat{V} (p/N) &\langle (\cN_+ +1)^{1/2} \xi , E_2  a_p a_{-p} \xi \rangle \Big| \\ 
&= \left| \kappa \int_{\Lambda \times \Lambda} dx dy N^3 V(N(x-y))\langle E^*_2 (\cN_+ +1)^{1/2} \xi,  \check{a}_x \check{a}_y \xi \rangle \right| \\ &\leq \kappa \int_{\Lambda \times \Lambda} dx dy N^3 V(N(x-y)) \| \text{E}_2^* (\cN_+ +1)^{1/2} \xi \| \| 
\check{a}_x \check{a}_y \xi  \| \end{split} \] 
Since we are excluding the term with $(n,k) = (0,0)$, we have either $k_1 > 0$ or $k_2 >0$ or at least one of the $\Lambda$-operators has the form (\ref{eq:Pi2-typ}); this allows us to apply the bound (\ref{eq:E2impr}). We obtain 
\[ \begin{split} &\left| \kappa \sum_{p \in \Lambda^*_+} \widehat{V} (p/N) \langle \text{E}_2^* (\cN_+ +1)^{1/2} \xi , a_p a_{-p} \xi \rangle \right| \\ &\hspace{3cm} \leq  C^{n+k} \kappa^{n+k+1} \int_{\Lambda \times \Lambda} dx dy N^{5/2} V(N(x-y)) \| (\cN_+ +1)^{1/2} \xi \| \| \check{a}_x \check{a}_y \xi  \| \\ &\hspace{3cm} \leq C^{n+k} \kappa^{n+k+1}  \left[ \int_{\Lambda \times \Lambda} dx dy N^{2} V(N(x-y))  \| \check{a}_x \check{a}_y \xi  \|^2 \right]^{1/2} \\ &\hspace{4cm} \times \left[ \int_{\Lambda \times \Lambda} dx dy N^3 V(N(x-y)) \| (\cN_+ +1)^{1/2} \xi \|^2 \right]^{1/2} \\ &\hspace{3cm} \leq C^{n+k} \kappa^{n+k+1/2} \| (\cN_+ +1)^{1/2} \xi \| \| \cV_N^{1/2} \xi \| \end{split} \] 
The contribution of the other terms on the r.h.s. of (\ref{eq:eta-scat2}) can be bounded similarly. We conclude that, in the case $\ell_1 = \ell_2 = 0$,  
\begin{equation}\label{eq:G00} \begin{split}  |\langle \xi , \text{G} \xi \rangle | \leq \; &C^{k+n} (1+k/N) \kappa^{n+k+1} \| (\cN_+ +1)^{1/2} \xi \|^2 \\ &+ C^{k+n} \kappa^{n+k+1/2} \| (\cN_+ +1)^{1/2} \xi \|  \| \cV_N^{1/2} \xi \| \end{split}  \end{equation}
Combining this bound with (\ref{eq:Gfin}) we obtain from (\ref{eq:adadK}), for sufficiently small $\kappa$, \[ \begin{split} \Big| \sum^*_{n,k} \frac{(-1)^{n+k}}{n!k!(n+k+1)}& \sum_{p \in \Lambda_+^*} p^2 \eta_p \left\langle \xi, \left[ \text{ad}^{(n)}_{B(\eta)} (b_p) \text{ad}^{(k)}_{B(\eta)} (b_{-p}) + \text{h.c.} \right]  \xi \right\rangle \Big| \\ & \leq C \kappa  \| (\cN_+ +1)^{1/2} \xi \|^2 + C \kappa^{1/2} \| (\cN_+ +1)^{1/2} \xi \|  \| \cV_N^{1/2} \xi \| \end{split} \]
Together with (\ref{eq:G1fin}), (\ref{eq:G2fin}), (\ref{eq:G3fin}), (\ref{eq:G4fin}), we finally estimate (\ref{eq:Kterms}) by 
\[ |\langle \xi , \cE^{(K)}_N \xi \rangle | \leq C \kappa \| (\cN_+ +1)^{1/2} \xi \| \| (\cK + \cN_+ +1)^{1/2} \xi \| + C \kappa^{1/2} \| (\cN_+ +1)^{1/2} \xi \| \| \cV_N^{1/2} \xi \|  \]
Hence, for any $\delta > 0$, we can find $C > 0$ such that 
\[ \pm \cE_N^{(K)} \leq \delta (\cK+\cV_N) + C \kappa (\cN_+ +1) \]
as claimed.
\end{proof}

\subsubsection{Analysis of $e^{-B(\eta)} \wt{\cL}_N^{(2)} e^{B(\eta)}$}
\label{sebsec:wtL2}

With $\wt{\cL}_N^{(2)}$ as in (\ref{eq:wtL2}), we write 
\begin{equation}\label{eq:cE2def} e^{-B(\eta)} \wt{\cL}_N^{(2)} e^{B(\eta)} = \kappa \sum_{p \in \Lambda_+^*} \widehat{V} (p/N) \eta_p + \frac{\kappa}{2} \sum_{p \in \Lambda_+^*} \widehat{V} (p/N) \left[ b_p b_{-p} + b^*_p b^*_{-p} \right] + \cE^{(2)}_N 
\end{equation}
In the next proposition, we estimate the error term $\cE^{(2)}_N$.
\begin{prop}\label{prop:cL2}
Let the assumptions of Proposition \ref{prop:gene} be satisfied (in particular, suppose $\kappa \geq 0$ is small enough). Then, for every $\delta > 0$, there exists a constant $C > 0$ such that, on $\cF^{\leq N}_+$, 
\[ \pm \cE^{(2)}_N \leq \delta \cV_N + C \kappa (\cN_+ +1) \]
\end{prop}

\begin{proof}
Recall that 
\begin{equation}\label{eq:wtcL2-de} \wt{\cL}_N^{(2)} = \kappa \sum_{p \in \Lambda_+^*} \widehat{V} (p/N) \left( b_p^* b_p - \frac{1}{N} a_p^* a_p \right) + \frac{\kappa}{2} \sum_{p \in \Lambda^*_+} \widehat{V} (p/N) \left( b_p b_{-p} + b^*_p b^*_{-p} \right) \end{equation}
The expectation of the conjugation of the first term 
can be estimated by 
\begin{equation}\label{eq:b*b-in}
\begin{split} \left| \kappa \sum_{p \in \Lambda_+^*} \widehat{V} (p/N) \langle \xi , e^{-B(\eta)} b_p^* b_p e^{B(\eta)} \xi \rangle \right| &\leq \kappa \sum_{p \in   \Lambda_+^*} |\widehat{V} (p/N)| \langle \xi , e^{-B(\eta)} b_p^* b_p e^{B(\eta)} \xi \rangle \\& \leq C \kappa \langle \xi , e^{-B(\eta)} \cN_+  e^{B(\eta)} \xi \rangle \\ &\leq C \kappa \| (\cN_+ +1)^{1/2} \xi \|^2 \end{split} \end{equation}
The contribution proportional to $-N^{-1} a_p^* a_p$ on the r.h.s. of (\ref{eq:wtcL2-de}) can be bounded analogously. So, let us focus on the last sum on the r.h.s. of (\ref{eq:wtcL2-de}). According to Lemma \ref{lm:conv-series}, we can expand
\begin{equation}\label{eq:bb-deco-in}\begin{split} \kappa \sum_{p \in \Lambda^*_+} \widehat{V} (p/N) & e^{-B(\eta)} b_p b_{-p} e^{B(\eta)}  \\ = \; &\sum_{n,k \geq 0} \frac{(-1)^{k+n}}{k!n!} \kappa \sum_{p \in \Lambda_+^*} \widehat{V} (p/N) \text{ad}^{(n)} (b_p) \text{ad}^{(k)} (b_{-p}) 
\\ = \; &\kappa \sum_{p \in \Lambda_+^*} \widehat{V} (p/N) b_p b_{-p} - \kappa \sum_{p \in \Lambda^*_+} \widehat{V} (p/N) b_p [ B(\eta) , b_{-p} ] \\ &+ \sum^*_{n,k} \frac{(-1)^{k+n}}{k!n!} \kappa \sum_{p \in \Lambda^*_+} \widehat{V} (p/N)  \text{ad}^{(n)} (b_p) \text{ad}^{(k)} (b_{-p}) \end{split} \end{equation}
where the sum $\sum^*$ runs over all pairs $(n,k) \not = (0,0), (0,1)$.
The first term on the r.h.s. of (\ref{eq:bb-deco-in}) does not enter the definition (\ref{eq:cE2def}) of the error term $\cE^{(2)}_{N}$. The second term on the r.h.s. of (\ref{eq:bb-deco-in}) is given by 
\begin{equation}\label{eq:bb-sec-dec}
\begin{split} 
- \kappa \sum_{p \in \Lambda^*_+} \widehat{V} &(p/N) b_p \left[ B(\eta) , b_{-p} \right] \\ = \; &\frac{N-\cN_+ }{N} \kappa \sum_{p\in \Lambda^*_+} \widehat{V}(p/N) \eta_p b_b b_p^* - \frac{\kappa}{N} \sum_{p,q \in \Lambda^*_+} \widehat{V} (p/N) \eta_q b_p b_q^* a_{-q}^* a_{-p} \\ = 
\; &\left( \frac{N-\cN_+ }{N} \right)^2 \kappa \sum_{p\in \Lambda^*_+} \widehat{V} (p/N) \eta_p + \frac{N-\cN_+ }{N} \kappa \sum_{p \in\Lambda^*_+} \widehat{V} (p/N) \eta_p \left( b_p^* b_p - \frac{3}{N} a_p^* a_p \right) \\ &- \frac{N-\cN_+ }{N^2} \kappa \sum_{p,q \in \Lambda^*} \widehat{V} (p/N) \eta_q a_q^* a_{-q}^* a_p a_{-p} \end{split} \end{equation}
To bound the expectation of the last term, we observe that 
\begin{equation}\label{eq:bb-sec3} \left| \frac{\kappa}{N} \sum_{p,q \in \Lambda^*} \widehat{V} (p/N) \eta_q \langle \xi , a_q^* a_{-q}^* a_p a_{-p} \xi \rangle \right| \leq \frac{\kappa}{N}  \Big\| \sum_{q\in \Lambda_+^*} \eta_q a_q a_{-q} \xi \Big\| \Big\| \sum_{p \in \Lambda_+^*} \widehat{V} (p/N) a_p a_{-p} \xi \Big\| \end{equation}
On the one hand, 
\[ \begin{split} \Big\| \sum_{q \in \Lambda^*_+} \eta_q a_q a_{-q} \xi \Big\| &\leq \sum_{q \in \Lambda_+^*} |\eta_q| \| a_q (\cN_+ +1)^{1/2} \xi \| \\ &\leq C \kappa \| (\cN_+ +1) \xi \| \leq C N^{1/2} \kappa \| (\cN_+ +1)^{1/2} \xi \| \end{split} \]
On the other hand, switching to position space, 
\[ \begin{split}  \kappa \Big\| \sum_{p \in \Lambda_+^*} \widehat{V} (p/N) a_p a_{-p} \xi \Big\|  &\leq \kappa \int_{\Lambda \times \Lambda} dx dy N^3 V(N(x-y)) \| \check{a}_x \check{a}_y \xi \| \\ & \leq C N^{1/2} \left(  \kappa^{1/2} \| \cV_N^{1/2} \xi \|  + C \kappa \| (\cN_+ +1)^{1/2} \xi \| \right) \end{split} \]
{F}rom (\ref{eq:bb-sec3}), we find
\begin{equation}\label{eq:bb-sec-3} \begin{split} & \left| \frac{\kappa}{N} \sum_{p,q \in \Lambda^*} \widehat{V} (p/N) \eta_q \langle \xi , a_q^* a_{-q}^* a_p a_{-p} \xi \rangle \right| \\ &\hspace{2cm} 
\leq C \kappa^2 \| (\cN_+ +1)^{1/2} \xi \|^2 + C \kappa^{3/2} \| (\cN_+ +1)^{1/2} \xi \| \| \cV_N^{1/2} \xi \| 
\end{split} \end{equation}
To control the first and second term on the r.h.s. of (\ref{eq:bb-sec-dec}), we observe that 
\begin{equation}\label{eq:riemann} \begin{split} 
\frac{\kappa}{N} \sum_{p \in \Lambda_+^*} |\widehat{V} (p/N)| \eta_p &\leq \frac{C \kappa^2}{N} \sum_{p \in \Lambda_+^*} \frac{|\widehat{V} (p/N)|}{p^2} \\ &\leq C \kappa^2 \sum_{q \in  \Lambda_+^*/N} \frac{1}{N^3} \frac{|\widehat{V} (q)|}{q^2} \leq C \kappa^2 \int_{\bR^3} \frac{|\widehat{V} (q)|}{q^2} dq \leq C \kappa^2 \end{split} \end{equation}
since the sum over the rescaled lattice $N^{-1} \Lambda^*_+$ 
can be interpreted as a Riemann sum. Together with (\ref{eq:bb-sec-3}), this remark implies that 
\begin{equation}\label{eq:bb-second} 
\begin{split} 
\Big| - &\kappa \sum_{p\in \Lambda_+^*} \widehat{V} (p/N) \langle \xi, b_p [B(\eta), b_{-p} ] \xi \rangle - \kappa \sum_{p \in \Lambda_+^*} \widehat{V} (p/N) \eta_p \Big| \\ &\hspace{3cm} \leq C \kappa^2 \| (\cN_+ +1)^{1/2} \xi \|^2 + C \kappa^{3/2} \| (\cN_+ +1)^{1/2} \xi \| \| \cV_N^{1/2} \xi \| \end{split} \end{equation}

Let us now focus on the sum $\sum^*$ over all pairs $(n,k) \not = (0,0), (0,1)$ on the r.h.s. of (\ref{eq:bb-deco-in}). According to Lemma \ref{lm:conv-series}, the operator 
\begin{equation}\label{eq:bb-deco} \kappa \sum_{p \in \Lambda^*_+} \widehat{V} (p/N)  \text{ad}^{(n)} (b_p) \text{ad}^{(k)} (b_{-p}) \end{equation}
can be expanded as the sum of $2^{n+k} n!k!$ terms having the form
\[ \begin{split} \text{I} = &\kappa \sum_{p \in \Lambda^*_+} \widehat{V} (p/N)  \Lambda_1 \dots \Lambda_{i_1} N^{-k_1} \Pi^{(1)}_{\sharp,\flat} (\eta^{j_1} , \dots , \eta^{j_{k_1}} ; \eta_p^{\ell_1} \ph_{\alpha_1 p} ) \\ &\hspace{3cm} \times \Lambda'_1 \dots \Lambda'_{i_2} N^{-k_2} \Pi^{(1)}_{\sharp', \flat'} (\eta^{m_1}, \dots , \eta^{m_{k_2}} ; \eta_p^{\ell_2} \ph_{-\alpha_2 p} ) \end{split} \]
where $i_1, i_2, k_1, k_2, \ell_1, \ell_2 \in \bN$, $j_1, \dots, j_{k_1}, m_1, \dots, m_{k_2} \in \bN \backslash \{ 0 \}$, $\alpha_i = (-1)^{\ell_i}$ for $i=1,2$ and where each operator $\Lambda_i, \Lambda'_i$ is either a a factor $(N-\cN_+ )/N$, a factor $(N-\cN_+ +1)/N$ or a $\Pi^{(2)}$-operator of order $h \in \bN \backslash \{ 0 \}$ having the form
\begin{equation}\label{eq:Pi2-bb} N^{-h} \, \Pi^{(2)}_{\underline{\sharp}
, \underline{\flat}} (\eta^{z_1}, \dots \eta^{z_h}) \end{equation}
with $z_1, \dots ,z_h \in \bN \backslash \{ 0 \}$. 
To bound the expectation of an operator of the form $\text{I}$ we consider first the case $\ell_1 + \ell_2 \geq 1$. Combining the bounds (\ref{eq:E11}) (if $\ell_1 + \ell_2 \geq 2$), (\ref{eq:E10}) (if $(\ell_1,\ell_2) = (1,0)$) and (\ref{eq:E01}) (if $(\ell_1,\ell_2) = (0,1)$) from Lemma \ref{lm:prel}, we obtain 
\begin{equation}\label{eq:bb-ellgeq2} \begin{split} 
|\langle \xi, \text{I} \xi \rangle | \leq  \; & \kappa \sum_{p \in \Lambda^*_+} |\widehat{V} (p/N)| \| (\cN_+ +1)^{1/2} \xi \| \| (\cN_+ +1)^{-1/2}  \Lambda_1 \dots \Lambda_{i_1}  \\ & \times N^{-k_1} \Pi^{(1)}_{\sharp,\flat} (\eta^{j_1} , \dots , \eta^{j_{k_1}} ; \eta_p^{\ell_1} \ph_p ) \Lambda'_1 \dots \Lambda'_{i_2} N^{-k_2} \Pi^{(1)}_{\sharp', \flat'} (\eta^{m_1}, \dots , \eta^{m_{k_2}} ; \eta_p^{\ell_2} \ph_p )\xi \| \\
\leq \; &C^{n+k} \kappa^{n+k+1} \| (\cN_+ +1)^{1/2} \xi \|  \\ &\hspace{1cm} \times  \sum_{p \in \Lambda^*_+} |\widehat{V} (p/N)|\Big\{  (1+k/N) p^{-4} \| (\cN_+ +1)^{1/2} \xi \| \\ &\hspace{4cm} + p^{-2} \| a_p \xi \|   + N^{-1} p^{-2} \| (\cN_+ +1)^{1/2} \xi \| \Big\} \\
\leq \; &C^{k+n} (1+k/N) \kappa^{n+k+1}  \| (\cN_+ +1)^{1/2} \xi \|^2 
\end{split} 
\end{equation}
where we used again the bound (\ref{eq:riemann}). If instead $\ell_1 = \ell_2 = 0$, we use (\ref{eq:Edeco}) to decompose
\[ \begin{split} \langle \xi, \text{I} \, \xi \rangle  = \; &\kappa \sum_{p\in \Lambda^*_+} \widehat{V} (p/N) \langle (\cN_+ +1)^{1/2} \xi, \text{E}_1 (p,-p) \rangle \\ &+ \kappa \sum_{p\in \Lambda^*_+} \widehat{V} (p/N) \langle (\cN_+ +1)^{1/2} \xi , \text{E}_2 a_p a_{-p} \xi \rangle \end{split} \]
The r.h.s. of the last equation can be estimated exactly as we did with the r.h.s. of (\ref{eq:Edeco-kin}). We obtain, similarly to (\ref{eq:G00}), that for $\ell_1 = \ell_2 = 0$, 
\[ \begin{split} |\langle \xi , \text{I} \xi \rangle | \leq \; &C^{k+n} (1+k/N) \kappa^{n+k+1} \| (\cN_+ +1)^{1/2} \xi \|^2 \\ &+ C^{k+n} \kappa^{n+k+1/2} \| (\cN_+ +1)^{1/2} \xi \| \| \cV_N^{1/2} \xi \| \, . \end{split}  \]
Combining this bound with (\ref{eq:bb-ellgeq2}), we find from (\ref{eq:bb-deco}) that for sufficiently small $\kappa$, 
\[ \begin{split} \Big|\sum^*_{n,k} \frac{(-1)^{k+n}}{k!n!} \kappa \sum_{p\in \Lambda_+^*} &\widehat{V} (p/N) \langle \xi , \text{ad}^{(n)} (b_p) \text{ad}^{(k)} (b_{-p}) \xi \rangle \Big| \\ &\leq C \kappa \| (\cN_+ +1)^{1/2} \xi \|^2 + C \kappa^{1/2} \| (\cN_+ +1)^{1/2} \xi \| \| \cV_N^{1/2} \xi \| \end{split} \]
Together with (\ref{eq:b*b-in}), (\ref{eq:bb-deco-in}) and (\ref{eq:bb-second}), we conclude that
\[ |\langle \xi , \cE^{(2)}_N \xi \rangle | \leq C \kappa \| (\cN_+ +1)^{1/2} \xi \|^2 + C \kappa^{1/2} \| (\cN_+ +1)^{1/2} \xi \| \| \cV_N^{1/2} \xi \|  \]
Hence, for every $\delta > 0$ we can find a constant $C > 0$ such that 
\[ \pm \cE_N^{(2)} \leq \delta \cV_N + C \kappa \| (\cN_+ +1)^{1/2} \xi \|^2 \]
\end{proof}

\subsection{Analysis of $\cG_N^{(3)}$}
\label{sebsec:L3}

{F}rom (\ref{eq:cLNj}) and (\ref{eq:cGN-deco}), we have 
\begin{equation}\label{eq:G3N2} \cG_N^{(3)} = \frac{\kappa}{\sqrt{N}} \sum_{p,q \in \Lambda_+^* , p+q \not = 0} \widehat{V} (p/N)  e^{-B(\eta)} b_{p+q}^* a_{-p}^* a_q e^{B(\eta)} + \text{h.c.} \end{equation}
In the next proposition, we show how to bound $\cG_N^{(3)}$.
\begin{prop}\label{prop:G3}
Let the assumptions of Proposition \ref{prop:gene} be satisfied (in particular, suppose $\kappa \geq 0$ is small enough). Then, for every $\delta > 0$ there exists $C > 0$ such that, on $\cF_+^*$,  
\[ \pm \cG_N^{(3)} \leq \delta \cV_N + C \kappa (\cN_+ +1) \]
\end{prop}  

Since some of the terms in $\cG_N^{(3)}$ (and many terms in $\cG_N^{(4)}$, which will be analyzed in the next subsection) 
have to be bounded with the potential energy operator, in the proof of Prop. \ref{prop:G3} (and in the proof of Prop. \ref{prop:E4} in the next subsection) we will often need to switch to position space. For this reason it is convenient to show a version of the estimates in Lemma \ref{lm:prel} stated in position space. The proof of the following Lemma follows closely the proof of Lemma 5.2 in \cite{BS}.
\begin{lemma} \label{lm:prel4}
Let $\xi \in \cF^{\leq N}_+$, $\beta \in \bN$, $i_1, i_2, k_1, k_2, \ell_1, \ell_2 \in \bN$, $j_1, \dots, j_{k_1}, m_1, \dots , m_{k_2} \in \bN\backslash \{ 0 \}$, For every $s=1, \dots, \max \{ i_1, i_2 \}$, let $\Lambda_s, \Lambda'_s$ be either a factor $(N-\cN_+ )/N$, $(N-\cN_+ +1)/N$ or a $\Pi^{(2)}$-operator of the form
\begin{equation}\label{eq:Pi-prel4} N^{-h} \Pi^{(2)}_{\underline{\sharp}, \underline{\flat}} (\eta^{z_1}, \dots , \eta^{z_h}) \end{equation}
for some $h \in \bN \backslash \{ 0 \}$, $z_1, \dots ,z_h \in \bN \backslash \{ 0 \}$ and $\underline{\sharp},\underline{\flat} \in \{ \cdot , * \}^h$. Suppose that the operators 
\[ \begin{split} \Lambda_1 \dots \Lambda_{i_1} N^{-k_1} \Pi^{(1)}_{\sharp, \flat} (\eta^{j_1}, \dots , \eta^{j_{k_1}} ; \check{\eta}^{\ell_1}_x ) \\ \Lambda'_1 \dots \Lambda'_{i_2} N^{-k_2} \Pi^{(1)}_{\sharp', \flat'} (\eta^{m_1}, \dots , \eta^{m_{k_2}} ; \check{\eta}^{\ell_2}_y )
 \end{split} \]
for some $\sharp \in \{ \cdot, * \}^{k_1}, \flat \in \{ \cdot, * \}^{k_1+1}, \sharp' \in \{ \cdot , * \}^{k_2}, \flat' \in \{ \cdot, * \}^{k_2+1}$ appear in the expansion of $\text{ad}^{(n)}_{B(\eta)} (\check{b}_x)$ and of $\text{ad}^{(k)}_{B(\eta)} (\check{b}_y)$ for some $n,k \in \bN$, as described in Lemma \ref{lm:indu}. Here we use the notation $\check{\eta}^{\ell_1}_x$ for the function $z \to \check{\eta}^{\ell_1} (x-z)$, where $\check{\eta}^{\ell_1}$ denotes the Fourier transform of the function $\eta^{\ell_1}$ defined on $\Lambda^*_+$. Let
\[ \begin{split} 
\text{S} = & \; \| (\cN_+ +1)^{\beta/2}  \Lambda'_1 \dots \Lambda'_{i_2} N^{-k_2} \Pi^{(1)}_{\sharp', \flat'} (\eta^{m_1}, \dots , \eta^{m_{k_2}} ; \check{\eta}^{\ell_2}_y )\\ &\hspace{2cm} \times  \Lambda_1 \dots \Lambda_{i_1} N^{-k_1} \Pi^{(1)}_{\sharp, \flat} (\eta^{j_1}, \dots , \eta^{j_{k_1}} ; \check{\eta}^{\ell_1}_x ) \xi \| \end{split} \]
Then we have the following bounds. If $\ell_1, \ell_2 \geq 1$, 
\begin{equation}\label{eq:S1} \text{S} \leq C^{n+k} \kappa^{n+k} \| (\cN_+ +1)^{(\beta+2)/2} \xi \|  \end{equation}
If $\ell_1 =0$ and $\ell_2 \geq 1$, \[ \text{S} \leq C^{n+k} \kappa^{n+k}  \| \check{a}_x (\cN_+ +1)^{(\beta+1)/2} \xi \| \]
If $\ell_1 \geq 1$ and $\ell_2 =0$,
\begin{equation}\label{eq:S2} \begin{split} \text{S} \leq \; &C^{n+k} \kappa^{n+k} n N^{-1} \| (\cN_+ +1)^{(\beta+2)/2} \xi \| \\ &+ C^{n+k} \kappa^{n+k-\ell_1} \mu_{\ell_1} |\check{\eta}^{\ell_1} (x-y)| \| (\cN_+ +1)^{\beta/2} \xi \| \\ &+C^{n+k} \kappa^{n+k} \| \check{a}_y (\cN_+ +1)^{(\beta+1)/2} \xi \| \end{split} \end{equation}
where $\mu_{\ell_1} = 1$ if $\ell_1$ is odd, while $\mu_{\ell_1} = 0$ if $\ell_1$ is even. If $\ell_1 \geq 1$ and $\ell_2 =0$ and we additionally assume that $k_1 > 0$ or $k_2 > 0$ or at least one of the $\Lambda$- or $\Lambda'$-operators is a $\Pi^{(2)}$-operator of the form (\ref{eq:Pi-prel4}), we obtain the improved estimate
\begin{equation}\label{eq:S2-imp} 
\begin{split} \text{S} \leq \; &C^{n+k} \kappa^{n+k} n N^{-1} \| (\cN_+ +1)^{(\beta+2)/2} \xi \| \\ &+ C^{n+k} \kappa^{n+k-\ell_1} \mu_{\ell_1} N^{-1} |\check{\eta}^{\ell_1} (x-y)| \| (\cN_+ +1)^{(\beta+2)/2} \xi \| \\ &+ C^{n+k} \kappa^{n+k} \| \check{a}_y (\cN_+ +1)^{(\beta+1)/2} \xi \| \end{split} \end{equation}
Finally, if $\ell_1 = \ell_2 = 0$, 
\[ \text{S} \leq C^{n+k} \kappa^{n+k} n N^{-1} \| \check{a}_x (\cN_+ +1)^{(\beta+1)/2} \xi \| +C^{n+k} \kappa^{n+k}  \| \check{a}_x \check{a}_y (\cN_+ +1)^{\beta/2} \xi \| \]
\end{lemma}

We are now ready to proceed with the proof of Prop. \ref{prop:G3}. 
\begin{proof}[Proof of Prop. \ref{prop:G3}] 
We start by writing 
\[ \begin{split} e^{-B(\eta)} a_{-p}^* a_q e^{B(\eta)} &= a_{-p}^* a_q + \int_0^1 ds \, e^{-sB(\eta)} [a_{-p}^* a_q , B(\eta)] e^{sB(\eta)} \\ &= a_{-p}^* a_q + \int_0^1 e^{-sB(\eta)} (\eta_q b_{-p}^* b^*_{-q} + \eta_p b_q b_p) e^{s B(\eta)} \end{split} \]
With Lemma \ref{lm:conv-series}, we obtain
\[ \begin{split} e^{-B(\eta)} &a_{-p}^* a_q e^{B(\eta)} \\ = \; &a_{-p}^* a_q \\ &+ \sum_{n,k \geq 0} \frac{(-1)^{n+k}}{n!k!(n+k+1)} \left[ \eta_q \text{ad}^{(n)}_{B(\eta)} (b_{-p}^*) \text{ad}^{(k)}_{B(\eta)} (b^*_{-q}) + \eta_p \text{ad}^{(n)}_{B(\eta)} (b_q) \text{ad}^{(k)}_{B(\eta)} (b_p) \right]
\end{split} \]
{F}rom (\ref{eq:G3N2}), we find
\begin{equation}\label{eq:GN3-in} \begin{split} \cG_N^{(3)} = \; &\sum_{r \geq 0} \frac{(-1)^r}{r!} \frac{\kappa}{\sqrt{N}} \sum_{p,q \in \Lambda^*_+ : p+q \not = 0} \widehat{V} (p/N) \text{ad}^{(r)}_{B(\eta)} (b^*_{p+q}) a_{-p}^* a_q \\
&+ \sum_{n,k,r \geq 0} \frac{(-1)^{n+k+r}}{n!k!r!(n_k+1)} \\ &\hspace{1cm} \times \frac{\kappa}{\sqrt{N}} \sum_{p,q \in \Lambda_+^* , p+q \not = 0} \widehat{V} (p/N) \eta_q \, \text{ad}^{(r)}_{B(\eta)} (b^*_{p+q}) \text{ad}^{(n)}_{B(\eta)} (b_{-p}^*) \text{ad}^{(k)}_{B(\eta)} (b^*_{-q}) \\
&+ \sum_{n,k,r \geq 0} \frac{(-1)^{n+k+r}}{n!k!r!(n+k+1)} \\ &\hspace{1cm} \times \frac{\kappa}{\sqrt{N}} \sum_{p,q \in \Lambda_+^* , p+q \not = 0} \widehat{V} (p/N) \eta_p \, \text{ad}^{(r)}_{B(\eta)} (b^*_{p+q}) \text{ad}^{(n)}_{B(\eta)} (b_{p}) \text{ad}^{(k)}_{B(\eta)} (b_{q})
\\ &+\text{h.c.} 
\end{split} \end{equation}

We start by analyzing the last sum on the r.h.s. of (\ref{eq:GN3-in}). From Lemma \ref{lm:indu}, each operator
\begin{equation}\label{eq:L-in} \frac{\kappa}{\sqrt{N}} \sum_{p,q \in \Lambda_+^* , p+q \not = 0} \widehat{V} (p/N) \eta_p \, \text{ad}^{(r)}_{B(\eta)} (b^*_{p+q}) \text{ad}^{(n)}_{B(\eta)} (b_{p}) \text{ad}^{(k)}_{B(\eta)} (b_{q}) \end{equation}
can be expanded in the sum of $2^{n+k+r} n!k!r!$ terms having the form 
\begin{equation}\label{eq:typG}
\begin{split} \text{L} = \; &\frac{\kappa}{\sqrt{N}} \sum_{p,q \in \Lambda_+^* , p+q \not = 0} \widehat{V} (p/N) \eta_p \, \Pi^{(1)}_{\sharp , \flat} (\eta^{j_1} , \dots , \eta^{j_{k_1}}; \eta^{\ell_1}_{p+q} \ph_{\alpha_1 (p+q)})^*  \Lambda^*_{i_1} \dots \Lambda^*_{i_1} \\ &\hspace{3cm} \times  \Lambda'_1 \dots \Lambda'_{i_2} N^{-k_2} \Pi^{(1)}_{\sharp' , \flat'} (\eta^{m_1} , \dots , \eta^{m_{k_2}}; \eta^{\ell_2}_p \ph_{\alpha_2 p})  \\ &\hspace{3cm} \times \Lambda''_1 \dots \Lambda''_{i_3} N^{-k_3} \Pi^{(1)}_{\sharp'' , \flat''} (\eta^{s_1} , \dots , \eta^{s_{k_3}}; \eta^{\ell_3} \ph_{\alpha_3 q}) \end{split} \end{equation}
where $i_1, i_2, i_3, k_1, k_2, k_3, \ell_1, \ell_2 , \ell_3 \in \bN$, $j_1, \dots , j_{k_1}, m_1, \dots, m_{k_2}, s_1 , \dots , s_{k_3} \in \bN \backslash \{ 0 \}$, $\alpha_i = (-1)^{\ell_i}$ for $i=1,2$ and where each operator $\Lambda_i, \Lambda'_i, \Lambda''_i$ is either a factor $(N-\cN_+ )/N$, a factor $(N+1 - \cN_+ )/N$ or a $\Pi^{(2)}$-operator of the form
\[ N^{-h} \Pi^{(2)}_{\underline{\sharp}, \underline{\flat}} (\eta^{z_1} , \dots , \eta^{z_s}) \]
for some $h, z_1 \dots, z_h \in \bN\backslash \{ 0 \}$. The expectation of (\ref{eq:typG}) can be bounded by 
\[ \begin{split} 
|\langle \xi , \text{L} \xi \rangle | \leq \; &\frac{\kappa}{\sqrt{N}} \sum_{\substack{p,q \in \Lambda^*_+ : \\ p \not =-q}} |\widehat{V} (p/N)| |\eta_p| \| \Lambda_1 \dots \Lambda_{i_1} N^{-{k_1}} \Pi^{(1)}_{\sharp, \flat} (\eta^{j_1}, \dots , \eta^{j_{k_1}} ; \eta^{\ell_1}_{p+q} \ph_{\alpha_1 (p+q)} ) \xi \| \\ &\hspace{2cm} \times \| \Lambda'_1 \dots \Lambda'_{i_2} N^{-k_2} \Pi^{(1)}_{\sharp' , \flat'} (\eta^{m_1} , \dots , \eta^{m_{k_2}}; \eta^{\ell_2}_p \ph_{\alpha_2 p})  \\ &\hspace{3cm} \times \Lambda''_1 \dots \Lambda''_{i_3} N^{-k_3} \Pi^{(1)}_{\sharp'' , \flat''} (\eta^{s_1} , \dots , \eta^{s_{k_3}}; \eta^{\ell_3} \ph_{\alpha_3 q}) \xi \| \end{split} \]
Combining the bounds (\ref{eq:D1}) (if $\ell_1 \geq 1$) and (\ref{eq:D0}) (if $\ell_1 = 0$) on the one hand, and the bounds (\ref{eq:E11}) (if $\ell_2,\ell_3 \geq 1$), (\ref{eq:E10}) (if $\ell_2 \geq 1$ and $\ell_3 = 0$), (\ref{eq:E01-b}) (if $\ell_2 = 0$ and $\ell_3 \geq 1$) and (\ref{eq:Edeco}) (if $\ell_1 = \ell_2 = 0$) on the other hand, we conclude that 
\[  \begin{split} 
|\langle \xi , \text{L} \xi \rangle | \leq \; &C^{n+k+r} \kappa^{n+k+r+2} \\ &\times \frac{1}{\sqrt{N}} \sum_{\substack{p,q \in \Lambda^*_+ : \\ p \not = -q}} \frac{1}{p^{2}} \Big\{ \frac{1}{(p+q)^{2}} \| (\cN_+ +1)^{1/2} \xi \| + \| a_{p+q} \xi \| \Big\} \\ & \hspace{1cm} \times \Big\{ \frac{(1+ r/N)}{p^{2} q^{2}} \| (\cN_+ +1) \xi \| + \frac{(1+r/N)}{p^{2}} \| a_q (\cN_+ +1)^{1/2} \xi \| \\ &\hspace{2cm}  + \frac{1}{q^{2}} \| a_p (\cN_+ +1)^{1/2} \xi \| + \| a_p a_q \xi \| \Big\} \\
\leq \; &C^{n+k+r} (1+r/N) \kappa^{n+k+r+2} \| (\cN_+ +1)^{1/2} \xi \|^2 \end{split} \] 
{F}rom (\ref{eq:L-in}), we obtain that the expectation of the last sum on the r.h.s. of (\ref{eq:GN3-in}) is bounded by   
\begin{equation}\label{eq:fir-3}\begin{split} \Big| \sum_{n,k,r \geq 0} &\frac{(-1)^{n+k+r}}{n!k!r!(n+k+1)}  \\ &\times \frac{\kappa}{\sqrt{N}} \sum_{p,q \in \Lambda_+^* , p+q \not = 0} \widehat{V} (p/N) \eta_p \, \langle \xi, \text{ad}^{(r)}_{B(\eta)} (b^*_{p+q}) \text{ad}^{(n)}_{B(\eta)} (b_{p}) \text{ad}^{(k)}_{B(\eta)} (b_{q}) \xi \rangle \Big| \\ &\hspace{8cm}  \leq C \kappa^2  \| (\cN_+ +1)^{1/2} \xi \|^2 \end{split} \end{equation}

Next, we consider the second sum on the r.h.s. of (\ref{eq:GN3-in}) (we take the hermitian conjugated operator). To bound the expectation of this term, we will need to use the potential energy operator. For this reason, it is convenient 
to switch to position space. We find 
\begin{equation}\label{eq:adbadbadb}  \begin{split} 
\frac{\kappa}{\sqrt{N}} &\sum_{p,q \in \Lambda_+^* , p+q \not = 0} \widehat{V} (p/N) \eta_q \, \text{ad}^{(r)}_{B(\eta)} (b_{-q}) \text{ad}^{(n)}_{B(\eta)} (b_{-p}) \text{ad}^{(k)}_{B(\eta)} (b_{p+q})  \\ &=\kappa \int_{\Lambda \times \Lambda} dx dy \, N^{5/2} V(N(x-y)) \text{ad}^{(r)}_{B(\eta)} (b (\check{\eta}^{1+\ell_1}_x) \text{ad}^{(n)}_{B(\eta)} (\check{b}_y) \text{ad}^{(k)}_{B(\eta)} (\check{b}_x) 
\end{split} \end{equation}
where we used the notation $\check{\eta}^{s}$ to indicate the Fourier transform of the sequence $\Lambda^* \ni p \to \eta^s_p$, and $\check{\eta}^s_x$ denotes the function (or the distribution, if $s =0$) $z \to \check{\eta}^s_x (z) = \check{\eta}^s (z-x)$.  
With Lemma \ref{lm:indu}, the r.h.s. of (\ref{eq:adbadbadb}) can be written as the sum of $2^{n+k+r} n!k!r!$ terms, all having the form 
\begin{equation}\label{eq:Hop} \begin{split} 
\text{M}  = \; & \kappa \int_{\Lambda \times \Lambda} dx dy \, N^{5/2} V(N(x-y)) \, \Lambda_1 \dots \Lambda_{i_1} N^{-k_1} 
\Pi^{(1)}_{\sharp , \flat} (\eta^{j_1} , \dots , \eta^{j_{k_1}}; \check{\eta}^{1+\ell_1}_x )  \\ & \hspace{3cm} \times \Lambda'_1 \dots \Lambda'_{i_2}  N^{-k_2} \Pi^{(1)}_{\sharp' , \flat'} (\eta^{m_1} , \dots , \eta^{m_{k_2}}; \check{\eta}^{\ell_2}_y)  \\ & \hspace{3cm} \times \Lambda^{''}_1 \dots \Lambda^{''}_{i_3} N^{-k_3} \Pi^{(1)}_{\sharp^{''} , \flat^{''}} (\eta^{s_1} , \dots , \eta^{s_{k_3}}; \check{\eta}^{\ell_3}_{x}) \end{split} 
\end{equation}
where $i_1, i_2, i_3, k_1, k_2, k_3, \ell_1, \ell_2 , \ell_3 \in \bN$, $j_1, \dots , j_{k_1}, m_1, \dots, m_{k_2}, s_1 , \dots , s_{k_3} \in \bN \backslash \{ 0 \}$ and where each operator $\Lambda_i, \Lambda'_i, \Lambda''_i$ is either a factor $(N-\cN_+ )/N$, a factor $(N+1 - \cN_+ )/N$ or a $\Pi^{(2)}$-operator of the form
\begin{equation}\label{eq:Pi2-M} N^{-h} \Pi^{(2)}_{\underline{\sharp}, \underline{\flat}} (\eta^{z_1} , \dots , \eta^{z_h}) \end{equation}
for some $h, z_1, \dots , z_h \in \bN \backslash \{ 0 \}$. To bound the expectation of (\ref{eq:Hop}), we first assume that $(n,k) \not = (0,1)$. Under this condition, we bound
\begin{equation}\label{eq:MM} \begin{split} 
|\langle \xi , \text{M} \xi \rangle | \leq \; &\kappa \int_{\Lambda \times \Lambda} dx dy \, N^{5/2} V(N(x-y)) \| N^{-k_1} \Pi^{(1)}_{\sharp,\flat} (\eta^{j_1}, \dots, \eta^{j_{k_1}} ; \check{\eta}^{\ell_1+ 1}_x )^* \Lambda_{i_1}^* \dots \Lambda^*_1 \xi \| \\ &\hspace{1cm} \times \Big\| \Lambda'_1 \dots \Lambda'_{i_2}  N^{-k_2} \Pi^{(1)}_{\sharp' , \flat'} (\eta^{m_1} , \dots , \eta^{m_{k_2}}; \check{\eta}^{\ell_2}_y)  \\ & \hspace{2cm} \times \Lambda^{''}_1 \dots \Lambda^{''}_{i_3} N^{-k_3} \Pi^{(1)}_{\sharp^{''} , \flat^{''}} (\eta^{s_1} , \dots , \eta^{s_{k_3}}; \check{\eta}^{\ell_3}_{x}) \xi \Big\| \end{split} \end{equation}
With Lemma \ref{lm:prel4}, we estimate
\begin{equation}\label{eq:MM1} \| N^{-k_1} \Pi^{(1)}_{\sharp,\flat} (\eta^{j_1}, \dots, \eta^{j_{k_1}} ; \check{\eta}^{\ell_1+ 1}_x )^* \Lambda_{i_1}^* \dots \Lambda^*_1 \xi \| \leq C^r \kappa^{r+1} 
\| (\cN_+ +1)^{1/2} \xi \| \end{equation}
Considering separately all possible choices for the parameters $\ell_2, \ell_3$, Lemma \ref{lm:prel4} also implies that 
\begin{equation}\label{eq:MM2} \begin{split}  &\Big\| \Lambda'_1 \dots \Lambda'_{i_2}  N^{-k_2} \Pi^{(1)}_{\sharp' , \flat'} (\eta^{m_1} , \dots , \eta^{m_{k_2}}; \check{\eta}^{\ell_2}_y)   \Lambda^{''}_1 \dots \Lambda^{''}_{i_3} N^{-k_3} \Pi^{(1)}_{\sharp^{''} , \flat^{''}} (\eta^{s_1} , \dots , \eta^{s_{k_3}}; \check{\eta}^{\ell_3}_{x}) \xi \Big\| \\ &\hspace{1cm} \leq C^{n+k} \kappa^{n+k} \Big\{ (1+k/N) \| (\cN_+ +1) \xi \| + (1+k/N) \| \check{a}_x (\cN_+ +1)^{1/2} \xi \| \\ &\hspace{8cm} + \|\check{a}_y (\cN_+  +1)^{1/2} \xi \| + \| \check{a}_x \check{a}_y \xi \|  \Big\} \end{split} \end{equation}
When dealing with the choice $(\ell_2, \ell_3) = (0,1)$, we used here the exclusion of the pair $(n,k) = (0,1)$, which implies that $n+k \geq 1$ (because $n \geq \ell_2, k \geq \ell_3$) and therefore that either $k_2 > 0$ or $k_3 > 0$ or that at least one of the $\Lambda'$- or of the $\Lambda''$-operators is a $\Pi^{(2)}$-operator of the form (\ref{eq:Pi2-M}); this observation allowed us to use the bound (\ref{eq:S2-imp}), which together with $|\check{\eta} (x-y)| \leq C N \| V \|_1$, led us to (\ref{eq:MM2}). 
Inserting (\ref{eq:MM1}) and (\ref{eq:MM2}) in (\ref{eq:MM}), we arrive at
\begin{equation}\label{eq:Mbd} \begin{split} |\langle \xi , \text{M} \xi \rangle | \leq \; &C^{n+k+r} (1+k/N)
\kappa^{n+k+r+2} \| (\cN_+ +1)^{1/2} \xi \|  \int_{\Lambda \times \Lambda} dx dy \, N^{5/2} V(N(x-y)) \\ &\hspace{1cm} \times \Big\{ \| (\cN_+ +1) \xi \| + \| \check{a}_x (\cN_+ +1)^{1/2} \| + \| \check{a}_y (\cN_+ +1)^{1/2} \xi \| + \| \check{a}_x \check{a}_y \xi \| \Big\}  \\ \leq \; &C^{n+k+r} (1+k/N) \kappa^{n+k+r+2}  \| (\cN_+ +1)^{1/2} \xi \|^2 \\ &+ C^{n+k+r} (1+k/N)
\kappa^{n+k+r+3/2}  \| (\cN_+ + 1)^{1/2} \xi \| \| \cV_N^{1/2} \xi \| 
\end{split} \end{equation}

Finally, let us consider the expectation of (\ref{eq:Hop}) in the case $(n,k) = (0,1)$. In fact, we can further restrict our attention to the choice $(\ell_2, \ell_3) = (0,1)$, because for all other choices of $(\ell_2, \ell_3)$, the bound (\ref{eq:MM2}) remains true even if $(n,k) = (0,1)$. If $(\ell_2, \ell_3) = (n,k) = (0,1)$, by Lemma \ref{lm:indu}, part iii) and iv), the operator (\ref{eq:Hop}) has the form 
\begin{equation}\label{eq:MM1M2} \begin{split} \text{M} = \;& \kappa \int_{\Lambda \times \Lambda} dx dy \, N^{5/2} V(N(x-y)) \\ &\hspace{1cm} \times \Lambda_1 \dots \Lambda_{i_1} N^{-k_1} \Pi^{(1)}_{\sharp,\flat} (\eta^{j_1} , \dots , \eta^{j_{k_1}} , \check{\eta}_x^{1+\ell_1} ) \check{b}_y  \frac{(N-\cN_+ )}{N} b^* (\check{\eta}_x) \\  = \; & \kappa 
\int_{\Lambda \times \Lambda} dx dy \, N^{5/2} V(N(x-y)) \\ &\hspace{1cm} \times \Lambda_1 \dots \Lambda_{i_1} N^{-k_1} \Pi^{(1)}_{\sharp,\flat} (\eta^{j_1} , \dots , \eta^{j_{k_1}} , \check{\eta}_x^{1+\ell_1} ) a^* (\check{\eta}_x) \frac{(N+1-\cN_+)}{N}  \check{a}_y \\ &+\kappa  \int_{\Lambda \times \Lambda} dx dy \, N^{5/2} V(N(x-y)) \\ &\hspace{1cm} \times \Lambda_1 \dots \Lambda_{i_1} N^{-k_1} \Pi^{(1)}_{\sharp,\flat} (\eta^{j_1} , \dots , \eta^{j_{k_1}} , \check{\eta}_x^{1+\ell_1} ) \frac{(N+1-\cN_+)}{N} \, \frac{(N-\cN_+)}{N} \check{\eta} (x-y) \\
=: \; & \text{M}_1 + \text{M}_2 
\end{split} \end{equation}
The expectation of the first term can be bounded by 
\begin{equation}\label{eq:MM1M22} \begin{split} |\langle \xi , \text{M}_1 \xi \rangle | &\leq C^r \kappa^{r+2} \int_{\Lambda \times \Lambda} dx dy \, N^{5/2} V(N(x-y)) \| (\cN_+ +1)^{1/2} \xi \| \| \check{a}_y (\cN_+ +1)^{1/2} \xi \| \\ &\leq C^r \kappa^{r+2} \| (\cN_+ +1)^{1/2} \xi \|^2 \end{split}  \end{equation}
As for the second term on the r.h.s. of (\ref{eq:MM1M2}), its expectation vanishes on vectors $\xi \in \cF_+^{\leq N}$ (because of the orthogonality to the constant orbital $\ph_0$). 

Combining (\ref{eq:Mbd}) with (\ref{eq:MM1M2}) and (\ref{eq:MM1M22}), and summing over all $n,k,r \in \bN$, we conclude that, if $\| V \|_1$ is small enough, the expectation of the second sum on the r.h.s. of (\ref{eq:GN3-in}) is bounded by
\begin{equation}\label{eq:sec-3} \begin{split} 
\Big| \sum_{n,k,r \geq 0} &\frac{(-1)^{n+k+r}}{n!k!r!(n+k+1)}  \\ \frac{\kappa}{\sqrt{N}} & \times \sum_{p,q \in \Lambda_+^* , p+q \not = 0} \widehat{V} (p/N) \eta_q \, \langle \xi, \text{ad}^{(r)}_{B(\eta)} (b^*_{p+q}) \text{ad}^{(n)}_{B(\eta)} (b_{-p}^*) \text{ad}^{(k)}_{B(\eta)} (b^*_{-q}) \xi \rangle \Big| \\ &\hspace{2cm} \leq C \kappa^2 \| (\cN_+ +1)^{1/2} \xi \|^2 + C \kappa^{3/2} \| (\cN_+ +1)^{1/2} \xi \| \| \cV_N^{1/2} \xi \| \end{split} 
\end{equation}

Finally, we consider the first sum on the r.h.s. of (\ref{eq:GN3-in}). {F}rom Lemma \ref{lm:indu}, each operator
\begin{equation}\label{eq:first-3} \frac{\kappa}{\sqrt{N}} \sum_{p,q \in \Lambda^*_+ : p+q \not = 0} \widehat{V} (p/N) \text{ad}^{(r)}_{B(\eta)} (b^*_{p+q}) a^*_{-p} a_q \end{equation}
can be written as the sum of $2^r r!$ terms having the form
\begin{equation}\label{eq:rsum} \text{P} = \frac{\kappa}{\sqrt{N}} \sum_{p,q \in \Lambda^*_+ : p+q \not = 0} \widehat{V} (p/N) N^{-k_1} \Pi^{(1)}_{\sharp, \flat} (\eta^{j_1}, \dots , \eta^{j_{k_1}} ; \eta_{p+q}^{\ell_1}  \ph_{\alpha_1 (p+q)} )^* \Lambda_{i_1}^* \dots \Lambda_1^* a_{-p}^* a_q \end{equation}
for $i_1, k_1, \ell_1 \in \bN$, $j_1, \dots , j_{k_1} \in \bN \backslash \{ 0 \}$, $\alpha_1 = 1$ if $\ell_1$ is even, $\alpha_1 = -1$ if $\ell_1$ is odd. To bound the expectation of $\text{P}$ we distinguish three cases.

If $\ell_1 \geq 2$, we bound (proceeding as in Lemma \ref{lm:prel})  
\[ \begin{split} |\langle \xi , \text{P} \xi \rangle | &\leq \frac{C \kappa}{\sqrt{N}} \sum_{p,q \in \Lambda^*_+ : p \not = -q} |\eta_{p+q}|^{\ell_1} \, \| a_{-p} \Lambda_1 .. \Lambda_{i_1} N^{-k_1} \Pi^{(1)} (\eta^{j_1}, \dots , \eta^{j_{k_1}} ;  \ph_{\alpha_1 (p+q)}) \xi \| \| a_q \xi \| \\ &\leq C^r \kappa^{r+1} \sum_{p,q \in \Lambda^*_+ : p \not = -q} \frac{1}{(p+q)^{4}} \left\{ \| a_{-p} \xi \| + \frac{r}{N p^2} \| (\cN_+ +1)^{1/2} \xi \| \right\} \| a_q \xi \| 
\\ &\leq C^r (1+r/N) \kappa^{r+1} \| (\cN_+ +1)^{1/2} \xi \|^2 \end{split} \]

If $\ell_1 = 1$, we commute the operator $a_{-(p+q)}$ (or the $b_{-(p+q)}$ operator) appearing in the $\Pi^{(1)}$-operator in (\ref{eq:rsum}) to the right, and the operator $a_{-p}^*$ to the left (it is important to note that $[a_{-(p+q)}, a_{-p}^*] = 0$ since $q \not = 0$). We find
\[ \begin{split} 
|\langle \xi , \text{P} \xi \rangle | &\leq \frac{C^r \kappa^{r+1}}{\sqrt{N}} \sum_{p,q \in \Lambda^*_+: p \not = -q} |\widehat{V} (p/N)| \frac{1}{(p+q)^2} \, \Big\{ \frac{r}{N p^2} \| (\cN_+ +1) \xi \| \| a_q \xi \|  \\ & \hspace{3cm} + \frac{1}{N (p+q)^2} \| a_{-p} (\cN_+ +1)^{1/2} \xi \| \| a_q \xi \| + \| a_{-p} \xi \| \| a_{-(p+q)} a_q \xi \| \Big\}  \\ &\leq C^r \kappa^{r+1} \| (\cN_+ +1)^{1/2} \xi \|^2 \end{split} \]

Finally, if $\ell_1 = 0$ we only commute $a^*_{-p}$ to the left. We find (similarly as in Lemma~\ref{lm:prel})
\begin{equation}
\begin{split}  
|\langle \xi , \text{P} \xi \rangle | \leq \; & \left| \frac{\kappa}{\sqrt{N}} \sum_{p,q \in \Lambda^*_+ : p+q \not = 0} \widehat{V} (p/N) \langle \text{R} \, a_{p+q} a_{-p} \xi,  a_q \xi \rangle \right| \\ &+ \frac{C^r r \kappa^{r+1}}{N} \sum_{p,q \in \Lambda^*_+ : p \not = -q} \frac{|\widehat{V} (p/N)|}{p^2}  \| a_{p+q} \xi \| \| a_q \xi \| 
\end{split} 
\end{equation}
for an operator $\text{R}$ with $\| R \xi \| \leq C^r \kappa^r$. To bound the first term, we switch to position space. We find, similarly to (\ref{eq:riemann}), 
\[ \begin{split}  
|\langle \xi , \text{P} \xi \rangle | \leq \; &\kappa \int_{\Lambda \times \Lambda} dx dy \, N^{5/2} V(N(x-y)) \| \text{R} \, \check{a}_x \check{a}_y \xi \| \|\check{a}_x \xi \| \\ &+ \frac{C^r r \kappa^{r+1}}{N} \sum_{p,q \in \Lambda^*_+ : p \not = -q} \frac{|\widehat{V} (p/N)|}{p^2}  \| a_{p+q} \xi \| \| a_q \xi \| 
\\ \leq \; &C^r \kappa^{r+1/2} \| (\cN_+ +1)^{1/2} \xi \| \| \cV_N^{1/2} \xi \| + C^r r \kappa^{r+1} \| (\cN_+ +1)^{1/2} \xi \|^2   
\end{split} \]
{F}rom (\ref{eq:first-3}), summing over all $r \in \bN$, we conclude that the expectation of the first sum on the r.h.s. of (\ref{eq:GN3-in}) is bounded, if $\| V \|_1$ is small enough, by 
\[ \begin{split} \Big| \sum_{r \geq 0} \frac{(-1)^r}{r!} &\frac{\kappa}{\sqrt{N}} \sum_{p,q \in \Lambda^*_+ : p+q \not = 0} \widehat{V} (p/n) \langle \xi, \text{ad}^{(r)}_{B(\eta)} (b^*_{p+q}) a_{-p}^* a_q \xi \rangle \Big| \\ &\leq C \kappa \| (\cN_+ +1)^{1/2} \xi \|^2 + C \kappa^{1/2} \| (\cN_+ +1)^{1/2} \xi \| \| \cV_N^{1/2} \xi \| \end{split} \]
{F}rom (\ref{eq:GN3-in}), (\ref{eq:fir-3}), (\ref{eq:sec-3}) and the last equation, it follows that for every $\delta > 0$ there exists $C > 0$ such that 
\[ \pm \cG_N^{(3)} \leq \delta \cV_N + C \kappa (\cN_+ +1)\]
\end{proof}

\subsection{Analysis of $\cG_N^{(4)}$}
\label{sebsec:L4}

With $\cL^{(4)}_N$ as defined in (\ref{eq:cLNj}), we write
\begin{equation}\label{eq:def-E4} \begin{split} \cG_N^{(4)} = \; &e^{-B(\eta)} \cL^{(4)}_N e^{B(\eta)} \\ = \; &\cV_N + \frac{\kappa}{2N} \sum_{q\in \Lambda^*_+, r\in \Lambda^*: r\not = -q} \widehat{V} (r/N) \eta_{q+r} \eta_q \\ &+ \frac{\kappa}{2N} \sum_{q,r \in \Lambda^*_+ : r \not = -q} \widehat{V} (r/N) \, \eta_{q+r} \left(  b_q b_{-q} + b^*_q b^*_{-q} \right)  + \cE^{(4)}_N \end{split} \end{equation}
In the next proposition, we estimate the error term $\cE^{(4)}_N$. 
\begin{prop}\label{prop:E4}
Let the assumptions of Proposition \ref{prop:gene} be satisfied (in particular, suppose $\kappa \geq 0$ is small enough). Then, for every $\delta > 0$ there exists $C > 0$ such that, on $\cF_+^*$,  
\[ \pm \cE_N^{(4)} \leq \delta \cV_N + C \kappa  (\cN_+ +1) \]
\end{prop}

\begin{proof}
We have
\begin{equation}\label{eq:quartic1} \begin{split} 
e^{-B(\eta)} &\cL^{(4)}_N e^{B(\eta)} \\ =\;& \frac{\kappa}{2N} \sum_{p,q \in \Lambda_+^*, r \in \Lambda^* : r \not = -p,q} \widehat{V} (r/N) e^{-B(\eta)} a_p^* a_q^* a_{q-r} a_{p+r}  e^{B(\eta)} \\ =\; & \cV_N + \frac{\kappa}{2N} \sum_{p,q \in \Lambda_+^*, r \in \Lambda^* : r \not = -p,q} \widehat{V} (r/N) \int_0^1 ds \,  e^{-sB(\eta)} \left[ a_p^* a_q^* a_{q-r} a_{p+r} , B(\eta) \right] e^{sB(\eta)}
 \\ = \; &\cV_N + \frac{\kappa}{2N} \sum_{q \in \Lambda_+^*, r \in \Lambda^* : r \not = -q} \widehat{V} (r/N) \eta_{q+r} \int_0^1 ds \, \left( e^{-sB(\eta)} b_q^* b_{-q}^* e^{sB(\eta)} + \text{h.c.} \right)  \\
& +\frac{\kappa}{N} \sum_{p,q \in \Lambda_+^* , r \in \Lambda^* : r \not = p,-q} \widehat{V} (r/N) \eta_{q+r} \int_0^1 ds \left( e^{-s B(\eta)} b_{p+r}^* b_q^* a^*_{-q-r} a_p e^{sB(\eta)} + \text{h.c.} \right) \end{split} \end{equation}
Now we observe that
\[ \begin{split}  e^{-sB(\eta)} a^*_{-q-r} a_p e^{s B(\eta)}  &=
a^*_{-q-r} a_p + \int_0^s d\tau \, e^{-\tau B(\eta)} \left[ a^*_{-q-r} a_p , B(\eta) \right] e^{-\tau B(\eta)} \\ &= a^*_{-q-r} a_p + \int_0^s d\tau \, e^{-\tau B(\eta)} \left( \eta_p b^*_{-p} b^*_{-q-r} + \eta_{q+r} b_p b_{q+r} \right) e^{-\tau B(\eta)} \end{split}  \]
Inserting in (\ref{eq:quartic1}) and using Lemma \ref{lm:conv-series}, we obtain 
\[  e^{-B(\eta)} \cL^{(4)}_N e^{B(\eta)} -\cV_N = \text{W}_1 + \text{W}_2 + \text{W}_3 + \text{W}_4 \]
where we defined 
\begin{equation}\label{eq:defW} \begin{split} \text{W}_1 = \; &
\sum_{n,k =0}^\infty \frac{(-1)^{n+k}}{n!k!(n+k+1)} 
 \\ &\hspace{.5cm} \times \frac{\kappa}{2N} \sum_{q \in \Lambda_+^*, r \in \Lambda^* : r \not = -q} \widehat{V} (r/N) \eta_{q+r} \left( \text{ad}^{(n)}_{B(\eta)} (b_q) \text{ad}^{(k)}_{B(\eta)} ( b_{-q}) + \text{h.c.} \right) \\ \text{W}_2 = \; & \sum_{n,k =0}^\infty \frac{(-1)^{n+k}}{n!k!(n+k+1)}  \\ &\hspace{.5cm} \times  \frac{\kappa}{N} \sum_{p,q \in \Lambda_+^* , r \in \Lambda^* : r \not = p,-q} \widehat{V} (r/N) \eta_{q+r} \, \left( \text{ad}^{(n)}_{B(\eta)} (b_{p+r}^*) \text{ad}^{(k)}_{B(\eta)} (b_q^*)  a^*_{-q-r} a_p + \text{h.c.} \right)  
\end{split} \end{equation} 
and 
\begin{equation}\label{eq:defW2}
\begin{split}
\text{W}_3 = \; &  \sum_{n,k,i,j =0}^\infty \frac{(-1)^{n+k+i+j}}{n!k!i!j!(i+j+1)(n+k+i+j+2)} \\ &\hspace{1cm}  \frac{\kappa}{N} \sum_{p,q\in \Lambda^*_+, r \in \Lambda^* : r \not = -p -q} \widehat{V} (r/N) \eta_{q+r} \eta_p  \\ &\hspace{2cm} \times \left( \text{ad}^{(n)}_{B(\eta)} (b^*_{p+r}) \text{ad}^{(k)}_{B(\eta)} (b^*_q) \text{ad}^{(i)}_{B(\eta)} (b_{-p}^*) \text{ad}^{(j)}_{B(\eta)} (b_{-q-r}^*) + \text{h.c.} \right) \\ \text{W}_4 = \; & \sum_{n,k,i,j =0}^\infty \frac{(-1)^{n+k+i+j}}{n!k!i!j!(i+j+1)(n+k+i+j+2)} \\ &\hspace{1cm} \times  \frac{\kappa}{N} \sum_{p,q\in \Lambda^*_+, r \in \Lambda^* : r \not = -p -q} \widehat{V} (r/N) \eta^2_{q+r}  \\ & \hspace{2cm} \times  \left( \text{ad}^{(n)}_{B(\eta)} (b^*_{p+r}) \text{ad}^{(k)}_{B(\eta)} (b^*_q) \text{ad}^{(i)}_{B(\eta)} (b_{p}) \text{ad}^{(j)}_{B(\eta)} (b_{q+r}) + \text{h.c.} \right) 
  \end{split} \end{equation}
We consider, first of all, the expectation of the term $\text{W}_2$. Since we will need the potential energy operator to bound this term, it is convenient to switch to position space. On $\cF^{\leq N}_+$, we find
\begin{equation} \begin{split} \text{W}_2 = \; &\sum_{n,k= 0}^\infty \frac{(-1)^{n+k}}{n!k!(n+k+1)} \\ & \times \kappa \int_{\Lambda \times \Lambda} dx dy N^2 V(N(x-y)) \left( \text{ad}^{(n)}_{B(\eta)} (\check{b}^*_x) \text{ad}_{B(\eta)}^{(k)} (\check{b}^*_y) a^* (\check{\eta}_x) \check{a}_y  + \text{h.c.} \right) \end{split} 
\end{equation}
with the notation $\check{\eta}_x (z) = \check{\eta} (x-z)$. With  Cauchy-Schwarz, we find
\[ \begin{split} 
\Big| \kappa \int_{\Lambda \times \Lambda} dx dy \, &N^2 V(N(x-y)) \langle \xi , \text{ad}^{(n)}_{B(\eta)} (\check{b}_x^*) \text{ad}^{(k)}_{B(\eta)} (\check{b}_y^*) a^* (\check{\eta}_x) \check{a}_y \xi \rangle \Big| \\ 
&\leq \kappa  \int_{\Lambda \times \Lambda} dx dy \, N^2 V(N(x-y)) \\ &\hspace{1cm} \times  \| (\cN_+ +1)^{1/2} \text{ad}^{(k)}_{B(\eta)} (\check{b}_y) \text{ad}^{(n)}_{B(\eta)} (\check{b}_x) \xi \| \| (\cN_+ +1)^{-1/2} a^* (\check{\eta}_x) \check{a}_y \xi \| 
\end{split} \]
We bound
\[ \| (\cN_+ +1)^{-1/2} a^* (\check{\eta}_x) \check{a}_y \xi \| \leq C \kappa \| \check{a}_y \xi \|  \]
With Lemma \ref{lm:indu}, we estimate $\| (\cN_+ +1)^{1/2} \text{ad}^{(k)}_{B(\eta)} (\check{b}_y) \text{ad}^{(n)}_{B(\eta)} (\check{b}_x) \xi \|$ by the sum of $2^{n+k} n!k!$ terms of the form 
\begin{equation}\label{eq:norm-T} \begin{split} 
\text{T} = &\left\| (\cN_+ +1)^{1/2} \Lambda_1 \dots \Lambda_{i_1} N^{-k_1} \Pi^{(1)}_{\sharp, \flat} (\eta^{j_1}, \dots , \eta^{j_{k_1}} ; \check{\eta}_{y}^{\ell_1}) \right. \\ &\hspace{3cm} \left. \times  \Lambda'_1 \dots \Lambda'_{i_2} N^{-k_2} \Pi^{(1)}_{\sharp,\flat} (\eta^{m_1} , \dots, \eta^{m_{k_2}} ; \check{\eta}^{\ell_2}_{x}) \xi \right\| \end{split} \end{equation}
with $i_1, i_2, k_1, k_2, \ell_1, \ell_2 \geq 0$, $j_1, \dots , j_{k_1}, m_1, \dots , m_{k_2} \geq 0$ and where each $\Lambda_i$ and $\Lambda'_i$ operator is either a factor $(N-\cN_+ )/N$, $(N-\cN_+ +1)/N$ or a $\Pi^{(2)}$-operator (here $\check{\eta}^{\ell_1}$ indicates the function with Fourier coefficients given by $\eta^{\ell_1}_p$, for all $p \in \Lambda^*_+$).  

With Lemma \ref{lm:prel4}, we find 
\begin{equation}\label{eq:Tf} \begin{split} \text{T} &\leq (n+1) C^{k+n} \kappa^{k+n} \Big\{ \| (\cN_+ +1)^{3/2} \xi \| + \| \check{a}_y (\cN_+ +1) \xi \| + \| \check{a}_x (\cN_+ +1) \xi \| \\ &\hspace{5cm} + N \| (\cN_+ +1)^{1/2} \xi \| + \sqrt{N} \| \check{a}_x \check{a}_y \xi \| \Big\}  \end{split} \end{equation}
For $\xi \in \cF^{\leq N}_+$, we obtain 
\[ \begin{split} 
&\left| \kappa \int_{\Lambda \times \Lambda} dx dy N^2 V(N(x-y)) \langle \xi,  \text{ad}^{(n)}_{B(\eta)} (\check{b}_x^*) \text{ad}^{(k)}_{B(\eta)} (\check{b}^*_y)  a^* (\check{\eta}_x) \check{a}_y \xi \rangle \right| \\ &\hspace{3cm} \leq  (n+1)!k! \, C^{n+k} \kappa^{n+k+2}  \int_{\Lambda \times \Lambda}  dx dy \, N^2 V(N(x-y))  \| a_y \xi \| \\ &\hspace{4cm} \times \Big\{ N \| (\cN_+ +1)^{1/2} \xi \| + N  \| \check{a}_x \xi \| + N \| \check{a}_y \xi \| + N^{1/2}  \| \check{a}_x \check{a}_y \xi \| \Big\} 
\\ &\hspace{3cm} \leq (n+1)! k! \, C^{n+k} \kappa^{n+k+3/2} \| (\cN_+ +1)^{1/2} \xi \| \| (\cV_N + \cN_+  + 1 )^{1/2} \xi \| \end{split} \] 
and therefore, if $\kappa$ is small enough, 
\begin{equation}\label{eq:W2-end}|\langle \xi , \text{W}_2 \xi \rangle | \leq  C \kappa^2 \| (\cN_+ +1)^{1/2} \xi \|^2 + C \kappa^{3/2} \| (\cN_+ +1)^{1/2} \xi \| \| \cV_N^{1/2} \xi \| \, .  \end{equation}

Next, let us consider the term $\text{W}_3$, defined in (\ref{eq:defW2}). As above, we switch to position space. We find
\begin{equation}\label{eq:W3}
\begin{split}  \text{W}_3 = &\; \sum_{n,k,i,j = 0}^\infty \frac{(-1)^{n+k+i+j}}{n!k!i!j! (i+j+1)(n+k+i+j+2)} \\ & \times \kappa \int dx dy \, N^2 V(N(x-y)) \\ &\hspace{1cm} \times  \left( \text{ad}^{(n)} (\check{b}^*_x) \text{ad}^{(k)}_{B(\eta)} (\check{b}^*_y) \text{ad}_{B(\eta)}^{(i)} (b^* (\check{\eta}_x)) \text{ad}^{(j)}_{B(\eta)} (b^* (\check{\eta}_y)) + \text{h.c.} \right) \end{split} \end{equation}
With Cauchy-Schwarz, we have
\[ \begin{split} & \left| \kappa \int dx dy N^2 V(N(x-y)) \langle \xi , \text{ad}^{(n)}_{B(\eta)} (\check{b}_x^*) \text{ad}^{(k)}_{B(\eta)} (\check{b}_y^*) \text{ad}^{(i)}_{B(\eta)} (\check{b}^* (\check{\eta}_x)) \text{ad}^{(j)}_{B(\eta)}  (\check{b} (\check{\eta}_y)) \xi \rangle \right| \\
&\hspace{2cm} \leq \kappa \int dx dy \, N^2 V(N(x-y)) \, \| (\cN_+ +1)^{1/2} \text{ad}^{(k)}_{B(\eta)} (\check{b}_y)  \text{ad}^{(n)}_{B(\eta)} (\check{b}_x) \xi \| \\ &\hspace{5cm} \times  \|  (\cN_+ +1)^{-1/2} 
\text{ad}^{(i)}_{B(\eta)} (b (\check{\eta}_x)) \text{ad}^{(j)} (b (\check{\eta}_y)) \xi \| 
\end{split} \]
Expanding $\text{ad}^{(i)}_{B(\eta_t)} (b (\check{\eta}_x)) \text{ad}^{(j)} (b (\check{\eta}_y))$ as in Lemma \ref{lm:indu} and using Lemma \ref{lm:prel4} (with $\ell_1$ and $\ell_2$ replaced by $\ell_1 + 1$ and $\ell_2 +1$, so that we can always use the inequality (\ref{eq:S1})), we obtain  
\begin{equation}\label{eq:W3-1} \begin{split} 
\| (\cN_+ +1)^{-1/2} &\text{ad}^{(i)}_{B(\eta)} (b (\check{\eta}_x)) \text{ad}^{(j)} (b (\check{\eta}_y)) \xi \| \leq i!j! \, C^{i+j} \kappa^{i+j+2}   \| (\cN_+ +1)^{1/2} \xi \|  \end{split} \end{equation}
As for the norm $\| (\cN_+ +1)^{1/2} \text{ad}^{(k)}_{B(\eta)} (\check{b}_y)  \text{ad}^{(n)}_{B(\eta)} (\check{b}_x) \xi \|$, we can estimate by the sum of $2^{n+k} n!k!$ contributions of the form (\ref{eq:norm-T}). With (\ref{eq:Tf}), we conclude that, if $\kappa$ is small enough,  
\begin{equation}\label{eq:W3end} |\langle \xi , \text{W}_3 \xi \rangle | \leq C \kappa^2 \| (\cN_+ +1)^{1/2} \xi \|^2 + C \kappa^{3/2} \| (\cN_+ +1)^{1/2} \xi \| \| \cV_N^{1/2} \xi \| \end{equation}

The term $\text{W}_4$ in (\ref{eq:defW2}) can be bounded similarly. First, we switch to position space. We find
\begin{equation}\label{eq:W4} \begin{split} \text{W}_4 = \; &\sum_{n,k,i,j =0}^\infty \frac{(-1)^{n+k+i+j}}{n!k!i!j! (i+j+1) (n+k+i+j+2)} \\ &\times \kappa \int dxdy \, N^2 V(N(x-y)) \, \left( \text{ad}^{(n)} (\check{b}_x) \text{ad}^{(k)} (\check{b}_y) \text{ad}^{(i)} (b (\check{\eta}^2_x)) \text{ad}^{(j)} (\check{b}_y) + \text{h.c.} \right) \end{split} \end{equation}
The expectation of the operators on the r.h.s. of (\ref{eq:W4}) can be bounded similarly as we did for the operators on the r.h.s. of (\ref{eq:W3}). The only difference is the fact that now we have to replace the estimate (\ref{eq:W3-1}) with 
\[ \| (\cN_+ +1)^{-1/2} \text{ad}^{(i)} (b (\check{\eta}^2_x)) \text{ad}^{(j)} (\check{b}_y) \xi \| \leq i! j! C^{i+j} \kappa^{i+j+2}  \left[ \| (\cN_+ +1)^{1/2} \xi \| + \| a_y \xi \| \right] \] 
We arrive at
\begin{equation}\label{eq:W4end} |\langle \xi , \text{W}_4 \xi \rangle | \leq C \kappa^2 \| (\cN_+ +1)^{1/2} \xi \|^2 + C \kappa^{3/2} \| (\cN_+ +1)^{1/2} \xi \| \| \cV_N^{1/2} \xi \| \end{equation}

Finally, we consider the term $\text{W}_1$ in (\ref{eq:defW}). 
Here, we separate contributions with $(n,k) = (0,0), (0,1)$ by writing: 
\begin{equation}\label{eq:W1-sep} \begin{split} 
\text{W}_1 = \; &\frac{\kappa}{2N} \sum_{q\Lambda_+^*, r\in \Lambda^* : r \not = -q} \widehat{V} (r/N) \eta_{r+q} ( b_q b_{-q} + \text{h.c.}) \\ &- \frac{\kappa}{4N} \sum_{q\in \Lambda^*_+, r\in \Lambda^*: r \not = -q} \widehat{V} (r/N) \eta_{q+r} \left( b_q \left[ B(\eta) , b_{-q} \right] + \text{h.c.} \right)  + \wt{W}_1 \end{split} \end{equation}
where 
\begin{equation}\label{eq:wtW1} 
\wt{W}_1 = \sum_{n,k}^* \frac{(-1)^{n+k}}{n!k!(n+k+1)} \frac{\kappa}{2N} \sum_{q\in \Lambda^*_+, r\in \Lambda^* : r \not = -q} \widehat{V} (r/N) \eta_{q+r} \left( \text{ad}^{(n)}_{B(\eta)} (b_q) \text{ad}^{(k)}_{B(\eta)} (b_{-q}) + \text{h.c.} \right) \end{equation}
and where the sum $\sum^*_{n,k}$ runs over all pairs $(n,k) \not = (0,0), (0,1)$. 

The first term on the r.h.s. of (\ref{eq:W1-sep}) does not enter the definition (\ref{eq:def-E4}) of the error term $\cE_N^{(4)}$. We do not have to estimate it. As for the second term on the r.h.s. of (\ref{eq:W1-sep}), we compute the commutator
\[ [B(\eta), b_{-q} ] = -\eta_q(1-\cN_+ /N) b_q^*   + \frac{1}{N} \sum_{m \in \Lambda^*_+}  \eta_m b^*_m a^*_{-m} a_{-q} \]
Hence
\[ \begin{split} \frac{\kappa}{N} \sum_{q\in \Lambda^*_+, r\in \Lambda^* : r \not = -q} \widehat{V} (r/N) \eta_{r+q} &b_q \left[ B(\eta), b_{-q} \right] \\ = \; &-\frac{\kappa}{N} \sum_{q\in \Lambda^*_+ , r\in \Lambda^*: r \not = -q} \widehat{V} (r/N) \eta_{r+q} \eta_q b_q b_q^* \left( 1 - \frac{\cN_+ +1}{N} \right) \\ &+ \frac{\kappa}{N^2} \sum_{q,m \in \Lambda^*_+, r\in \Lambda^* :r \not = -q} \widehat{V} (r/N) \eta_{r+q} \eta_m  b_q b_m^* a^*_{-m} a_{-q} \end{split} \]
and therefore
\[ \begin{split} \frac{\kappa}{N} \sum_{q \in \Lambda^*_+ , r\in \Lambda^*: r \not = -q} \widehat{V} (r/N)& \eta_{r+q} b_q \left[ B(\eta), b_{-q} \right] \\ = \; &-\frac{\kappa}{N} \sum_{q \in \Lambda^*_+ , r\in \Lambda^*: r \not = -q} \widehat{V} (r/N) \eta_{r+q} \eta_q \left( 1-\frac{\cN_+ }{N} \right) \left( 1 - \frac{\cN_+ +1}{N} \right) \\ &+\frac{2\kappa}{N^2} \sum_{q \in \Lambda^*_+, r\in \Lambda^* : r \not = -q} \widehat{V} (r/N) \eta_{r+q} \eta_q a_q^* a_q \left( 1 - \frac{\cN_+ +1}{N} \right) \\ &+ \frac{\kappa}{N^3} \sum_{q,m \in \Lambda^*_+, r\in \Lambda^* :r \not = -q} \widehat{V} (r/N) \eta_{r+q} \eta_m  a_m^* a^*_{-m} a_q a_{-q}  
\end{split} \]
We conclude that
\[ \begin{split} \frac{\kappa}{N} \sum_{q \in \Lambda^*_+, r\in \Lambda^* : r \not = -q} \widehat{V} (r/N) \eta_{r+q} b_q \left[ B(\eta), b_{-q} \right] + &\frac{\kappa}{N} \sum_{q \in \Lambda^*_+, r\in \Lambda^* : r \not = -q} \widehat{V} (r/N) \eta_{r+q} \eta_q \\ &\hspace{3cm} = \text{T}_1 + \text{T}_2 + \text{T}_3 \end{split} \]
with 
\[ \begin{split}  \text{T}_1 &= \frac{\kappa}{N^2} \sum_{q \in \Lambda^*_+, r\in \Lambda^* : r \not = -q} \widehat{V} (r/N) \eta_{r+q} \eta_q (2\cN_+ +1 + \cN_+ /N + \cN_+ ^2/N )  \\
\text{T}_2 &= \frac{2\kappa}{N^2} \sum_{q\in \Lambda^*_+, r\in \Lambda^* : r \not = -q} \widehat{V} (r/N) \eta_{r+q} \eta_q a_q^* a_q \left( 1 - \frac{\cN_+ +1}{N} \right) \\
\text{T}_3 &= \frac{\kappa}{N^3} \sum_{q,m \in \Lambda^*_+, r\in \Lambda^* :r \not = -q} \widehat{V} (r/N) \eta_{r+q} \eta_m  a_m^* a^*_{-m} a_q a_{-q} \end{split} \]
Since 
\begin{equation}\label{eq:bd1} \frac{\kappa}{N^2} \sum_{q \in \Lambda^*_+, r\in \Lambda^* : r \not = -q} \widehat{V} (r/N) \frac{1}{(r+q)^2 q^2} \leq C < \infty \end{equation}
uniformly in $N$, we easily find
\[ |\langle \xi , \text{T}_1 \xi \rangle | \leq C \kappa \| (\cN_+ +1)^{1/2} \xi \|^2 \]
Furthermore,
\[ | \langle \xi , \text{T}_2 \xi \rangle | \leq \frac{2\kappa^3}{N^2} \sum_{q \in \Lambda^*_+, r\in \Lambda^* : r \not = -q} |\widehat{V} (r/N)| \frac{1}{(r+q)^2 q^2}  \| a_q \xi \|^2 \leq C N^{-1} \kappa^3 \| (\cN_+ +1)^{1/2} \xi \|^2 \]
Finally, we consider the term $\text{T}_3$. To this end, we switch to position space. We find
\[ \begin{split} \text{T}_3 &= \frac{\kappa}{N^3} \sum_{q,m \in \Lambda^*_+, r\in \Lambda^* :r \not = -q} \widehat{V} (r/N) \eta_{r+q} \eta_m  a_m^* a^*_{-m} a_q a_{-q} \\ &=  \kappa \int_{\Lambda \times \Lambda} dx dy \, V(N(x-y)) \check{\eta} (x-y) B \check{a}_x \check{a}_y \end{split} \]
where $B = \sum_{m\in \Lambda_+^*} \eta_m a_m^* a_{-m}^*$. Since $\| B^* \xi \| \leq C \kappa \| (\cN_+ +1) \xi \|$, we obtain
\[ \begin{split}  | \langle \xi , \text{T}_3 \xi \rangle | &\leq C \kappa^2  \| (\cN_+ +1)^{1/2} \xi \| \int_{\Lambda \times \Lambda} dxdy \, N^{1/2} V(N(x-y)) |\check{\eta} (x-y)| \| \check{a}_x \check{a}_y \xi \| \\ 
&\leq C \kappa^2 \| (\cN_+ +1)^{1/2} \xi \| \int_{\Lambda \times \Lambda} N^{3/2} V(N(x-y))  \| a_x a_y \xi \| \\ &\leq C N^{-1} \kappa^{3/2} \| (\cN_+ +1)^{1/2} \xi \|  \| \cV_N^{1/2} \xi \| \end{split} \]

Let us now focus on the expectation of (\ref{eq:wtW1}). 
According to Lemma \ref{lm:indu}, the operator
\[ \frac{\kappa}{N} \sum_{q \in \Lambda^*_+, r\in \Lambda^* : r \not = -q} \widehat{V} (r/N) \eta_{q+r} \text{ad}^{(n)} (b_q) \text{ad}^{(k)} (b_{-q})  \]
can be written as the sum of $2^{n+k} n!k!$ terms having the form
\[ \begin{split} \text{X} = \frac{\kappa}{N} \sum_{q\in \Lambda^*_+, r\in \Lambda^* : r \not = -q} \widehat{V} (r/N) \eta_{q+r} \Lambda_1 &\dots \Lambda_{i_1} N^{-k_1} \Pi^{(1)}_{\sharp,\flat} (\eta^{j_1} , \dots , \eta^{j_{k_1}} ; \eta^{\ell_1}_q \ph_{\alpha_1 q} ) \\ &\times \Lambda'_1 \dots \Lambda'_{i_2} N^{-k_2} \Pi^{(1)}_{\sharp', \flat'} (\eta^{m_1} , \dots , \eta^{m_{k_2}} ; \eta^{\ell_2}_q \ph_{-\alpha_2 q} ) \end{split} \]
where $i_1, i_2, k_1, k_2, \ell_1, \ell_2 \in \bN$, $j_1, \dots , j_{k_1}, m_1, \dots , m_{k_2} \in \bN \backslash \{0 \}$, $\alpha_i = 1$ if $\ell_i$ is even and $\alpha_i = -1$ if $\ell_i$ is odd. To bound the expectation of the operator $\text{X}$, we distinguish two cases.

If $\ell_1 + \ell_2 \geq 1$, we use Lemma \ref{lm:prel} to estimate
\[ \begin{split} |\langle \xi , \text{X} \xi  \rangle | \leq \; & C^{n+k} \kappa^{n+k+2} \| (\cN_+ +1)^{1/2} \xi \| N^{-1} \sum_{q,r \in \Lambda^*_+ : r \not = -q} \frac{|\widehat{V} (r/N)|}{(q+r)^2} \\ &\times \left\{ \frac{1}{q^4} (1+k/N) \| (\cN_+ +1)^{1/2} \xi \| + \frac{1}{q^2} \| a_q \xi \| + \frac{1}{Nq^2} \| (\cN_+ +1)^{1/2} \xi \| \right\} \end{split} \]
Here we used the fact that we excluded the pairs $(n,k) = (0,0), (0,1)$ to make sure that, if $\ell_1 =0$ and $\ell_2 = 1$, then either $k_1 > 0$ or $k_2 > 0$ or at least one of the operators $\Lambda$ or $\Lambda'$ has to be a $\Pi^{(2)}$-operator. {F}rom (\ref{eq:bd1}) and from the similar bound 
\[ \sup_{q} \frac{1}{N} \sum_r |\widehat{V} (r/N)| \frac{1}{(q+r)^2} \leq C < \infty \]
uniformly in $N$, we conclude that, for $\ell_1 + \ell_2 \geq 1$, 
\begin{equation}\label{eq:X1} |\langle \xi, \text{X} \xi \rangle | \leq C \kappa^{2} \| (\cN_+ +1)^{1/2} \xi \|^2 \end{equation}

For $\ell_1 = \ell_2 = 0$, we use Lemma \ref{lm:prel} to write 
\[ \text{X} = \frac{\kappa}{N} \sum_{q\in \Lambda^*_+, r\in \Lambda^*} \widehat{V}(r/N) \eta_{q+r} \left[ A_{q} + B a_q a_{-q} \right] =: \text{X}_1 + \text{X}_2 \]
where 
\[ |\langle \xi , \text{A}_q \xi \rangle | \leq C^{n+k} \kappa^{n+k} \frac{k}{N q^2} \| (\cN_+ +1)^{1/2} \xi \|^2 \]
and (since we excluded the term with $(n,k) = (0,0)$) 
\[ \| B^* \xi \| \leq C^{n+k} N^{-1}  \kappa^{n+k} \| (\cN_+ +1) \xi \| \] 
We immediately obtain that
\[ \begin{split} 
|\langle \xi , \text{X}_1 \xi \rangle | &\leq \frac{C^{n+k} \kappa^{n+k+2}}{N^2} \sum_{q\in \Lambda^*_+, r\in \Lambda^*} \widehat{V} (r/N) \frac{1}{(q+r)^2 q^2} \| (\cN_+ +1)^{1/2} \xi \|^2 \\ &\leq C^{n+k} \kappa^{n+k+2} \| (\cN_+ +1)^{1/2} \xi \|^2 \end{split} \]
and, switching to position space,  
\[ \begin{split} 
| \langle \xi , \text{X}_2 \xi \rangle |  &=  \Big| \kappa \int_{\Lambda\times \Lambda} dx dy \, N^2 V(N(x-y)) \eta (x-y) \langle B^* \xi ,  \check{a}_x \check{a}_y \xi \rangle  \Big|\\ &\leq \kappa \int_{\Lambda \times \Lambda} dx dy N^2 V(N(x-y)) |\eta (x-y)| \| \check{a}_x \check{a}_y \xi \| \| B^* \xi \| \\ &\leq C \kappa^{n+k+2} \| (\cN_+ +1)^{1/2} \xi \| \int_{\Lambda \times \Lambda}  dx dy N^{5/2} V(N(x-y)) \| \check{a}_x \check{a}_y \xi \| 
\\ &\leq C \kappa^{n+k+3/2}  \| (\cN_+ +1)^{1/2} \xi \| \| \cV_N^{1/2} \xi \| \end{split} \] 
Combining the last two bounds with (\ref{eq:X1}), and then summing over all $n,k$, we find 
\[ |\langle \xi , \wt{W}_1 \xi \rangle | \leq C \kappa^2 \| (\cN_+ +1)^{1/2} \xi \|^2 + C \kappa^{3/2} \| (\cN_+ +1)^{1/2} \xi \| \| \cV_N^{1/2}\xi \|  \]
With (\ref{eq:defW}), (\ref{eq:defW2}), (\ref{eq:W2-end}), (\ref{eq:W3end}), (\ref{eq:W4end}), we conclude the proof of the proposition. 
\end{proof}

\subsection{Proof of Proposition \ref{prop:gene}}
\label{sub:proof}

Combining the results of Prop. \ref{prop:E0}, Prop. \ref{prop:K}, Prop. \ref{prop:cL2}, Prop. \ref{prop:G3} and Prop. \ref{prop:E4}, we conclude that the excitation Hamiltonian $\cG_N$ defined in (\ref{eq:GN}) is such that
\[ \begin{split} 
\cG_N = \; &\frac{(N-1)}{2} \kappa \widehat{V} (0) + \sum_{p\in \Lambda^*_+} \eta_p \left[ p^2 \eta_p + \kappa \widehat{V} (p/N)  + \frac{\kappa}{2N} \sum_{r \in \Lambda^* :r \not = -p} \widehat{V} (r/N) \eta_{p+r} \right]  \\ &+  \sum_{p \in \Lambda^*_+} \left[ p^2 \eta_p + \kappa \frac{\widehat{V} (p/N)}{2} + \frac{\kappa}{2N} \sum_{r \in \Lambda^* : r \not = -q} \widehat{V} (r/N) \eta_{p+r} \right] \left[ b_p b_{-p} + b^*_p b^*_{-p} \right] \\ &+\cK + \cV_N + \cE_N 
\end{split} \]
where the operator $\cE_N$ is such that, for all $\delta > 0$ there exists $C > 0$ with 
\[ \pm \cE_N \leq \delta (\cK + \cV_N) + C \kappa (\cN_+ +1) \]
With (\ref{eq:eta-scat}), we obtain 
\begin{equation}\label{eq:GNaf} \begin{split} \cG_N = \; &\frac{(N-1)}{2} \kappa \widehat{V} (0) + \frac{\kappa}{2} \sum_{p \in \Lambda^*_+} \widehat{V} (p/N) \eta_p + \cK + \cV_N + \cE_N \\ &+ \sum_{p\in \Lambda^*_+} \eta_p \left[  - \kappa \frac{\widehat{V} (p/N) \eta_0}{2N} + \lambda_\ell N^3 \widehat{\chi}_\ell (p) + \lambda_\ell N^2 \sum_{q \in \Lambda^*} \widehat{\chi}_\ell (p-q) \eta_q \right]  \\ 
&+ \sum_{p \in \Lambda^*_+} \left[ N^3 \lambda_\ell \widehat{\chi}_\ell (p) + N^2 \lambda_\ell \sum_{q \in \Lambda^*} \widehat{\chi}_\ell (p-q) \eta_q - \frac{\kappa}{2N} \widehat{V} (p/N) \eta_0 \right] (b_p b_{-p} + b_p^* b_{-p}^*)  \end{split} \end{equation}
With the definition (\ref{eq:ktdef}) and with the estimate (\ref{eq:wteta0}) we find that
\[ \begin{split} \left| \frac{(N-1)}{2} \kappa \widehat{V} (0) + \frac{\kappa}{2} \sum_{p \in \Lambda_+^*} \widehat{V} (p/N) \eta_p - \frac{N\kappa}{2} \int \, N^3 V(Nx) f_\ell (Nx) dx \right| \leq C \kappa \end{split} \]
With the approximate identity (\ref{eq:Vfa0}), we conclude that
\[  \left| \frac{(N-1)}{2} \kappa \widehat{V} (0) + \frac{\kappa}{2} \sum_{p \in \Lambda_+^*} \widehat{V} (p/N) \eta_p - 4 \pi a_0 N \right| \leq C \kappa \, . \]

As for the terms on the second line on the r.h.s. of (\ref{eq:GNaf}), they are all at most of order one. The first term can be estimated with (\ref{eq:wteta0}) by 
\[ \Big| \frac{\kappa}{N} \sum_{p \in \Lambda^*_+} \widehat{V} (p/N) \eta_p \eta_0 \Big| \leq \frac{C \kappa^3}{N} \sum_{p \in \Lambda^*_+} \frac{\widehat{V} (p/N)}{p^2} \leq C \kappa^3  \]
similarly to (\ref{eq:riemann}). The second term can be controlled 
using Lemma \ref{3.0.sceqlemma}, part i), which implies that $\lambda_\ell N^3 \leq C \kappa$. We find 
\[ N^3 \lambda_\ell \sum_{p \in \Lambda_+^*} \widehat{\chi}_\ell (p) \eta_p \leq C \kappa \| \chi_\ell \| \| \eta \| \leq C \kappa^2 \]
As for the third term, we use again the bound $N^3 \lambda_\ell \leq C \kappa$ to estimate
\[ \Big| \lambda_\ell N^2 \sum_{p \in \Lambda^*_+, q \in \Lambda^*} \widehat{\chi}_\ell (p-q) \eta_q \eta_p \Big|  \leq C N^{-1} \kappa \| \eta \|^2 \leq C N^{-1} \kappa^3 \]

Next, we bound the expectation of the operator on the last line on the r.h.s. of (\ref{eq:GNaf}). The first contribution can be estimated by 
\[ \begin{split} \Big| N^3 \lambda_\ell \sum_{p \in \Lambda^*_+} \widehat{\chi}_\ell (p) \langle \xi , b_p b_{-p} \xi \rangle \Big| &\leq C \kappa \| (\cN_+ +1)^{1/2} \xi \| \sum_{p \in \Lambda^*_+} |\widehat{\chi}_\ell (p)  \| a_{-p} \xi \| \\ &\leq C \kappa \| \chi_\ell \| \| (\cN_+ +1)^{1/2} \xi \|^2 \leq C \kappa \| (\cN_+ +1)^{1/2} \xi \|^2 \end{split} \]
Similarly, 
\[ \begin{split} \Big| N^2 \lambda_\ell \sum_{p \in \Lambda^*_+, q \in \Lambda^*} \widehat{\chi}_\ell (p-q) \eta_q \langle \xi , b_p b_{-p} \xi \rangle \Big| &\leq C N^{-1} \kappa \| \widehat{\chi}_\ell * \eta \| \| (\cN_+ +1)^{1/2} \xi \|^2 \\ &\leq C N^{-1} \kappa \| (\cN_+ +1)^{1/2} \xi \|^2 \end{split}  \]
Finally, to estimate the contribution of the last term on the last line on the r.h.s. of (\ref{eq:GNaf}), we switch to position space. We find
\[ \begin{split}  \Big| \frac{\kappa}{2N} \sum_{p \in \Lambda_+^*} \widehat{V} (p/N) \eta_0 \langle \xi , b_p b_{-p} \xi \rangle \Big| &\leq C \kappa \int dx dy \, N^2 V (N(x-y)) \| \check{a}_x \check{a}_y \xi \| \| \xi \| \\ &\leq C N^{-1/2} \kappa^{3/2} \| \cV_N^{1/2} \xi \| \| \xi \| \end{split} \]

We conclude that  
\[ \cG_N = 4 \pi a_0 N + \cK + \cV_N + \wt{\cE}_N \]
where the error term $\wt{\cE}_N$ is such that, for all $\delta > 0$ there exists a constant $C > 0$ such that 
\begin{equation}\label{eq:wtEN} \pm \wt{\cE}_N \leq \delta (\cK + \cV_N) + C \kappa (\cN_+  + 1) \end{equation}
The statement of Prop. \ref{prop:gene} now follows by the remark that, on $\cF_+^{\leq N}$, $\cN_+  \leq (2\pi)^{-2} \cK$ (i.e. the kinetic energy operator on $\cF^{\leq N}_+$ is gapped). Taking for example $\delta = 1$ in (\ref{eq:wtEN}), we find
\[ \cG_N \leq 4\pi a_0 N + 2(\cK+\cV_N) + C (\cN_+ +1) \leq 4 \pi a_0 N + C (\cK+\cV_N + 1) \]
Taking instead $\delta = 1/3$, we find the lower bound
\[ \cG_N \geq 4\pi a_0 N + \frac{2}{3} (\cK+\cV_N) - C \kappa (\cN_+ +1) \geq 4\pi a_0 N + \left[ \frac{2}{3} - \frac{C\kappa}{(2\pi)^2} \right] (\cK + \cV_N) - C  \]
Now, if $\kappa \geq 0$ is small enough, we obtain that 
\[ \cG_N \geq 4\pi a_0 N + \frac{1}{2} (\cK +\cV_N) - C \geq 4\pi a_0 N + 2 \pi^2 \cN_+ - C \]
which completes the proof of Prop. \ref{prop:gene}.

\medskip

{\it Acknowledgement.} B.S. gratefully acknowledge support from the NCCR SwissMAP and from the Swiss National Foundation of Science through the SNF Grants ``Effective equations from quantum dynamics'' and ``Dynamical and energetic properties of Bose-Einstein condensates''.

\end{document}